\documentclass[a4paper,11pt]{article}

\usepackage{article}

\title{Classification of (non)-frustrated 2D Ising models in genus 1\\
on isoradial graphs}
\author{
Béatrice de Tilière%
  \thanks{{\small 
PSL University-Dauphine, CNRS, UMR 7534, CEREMADE, 75016 Paris, France.
    {\small\texttt{ detiliere@ceremade.dauphine.fr}}
}},
Lucas Rey%
  \thanks{{\small %
PSL University-Dauphine, CNRS, UMR 7534, CEREMADE, 75016 Paris, France.
 {\small\texttt{ lucas.rey@inrae.fr}}}}
}

\begin{document}
\maketitle

\begin{abstract}
We prove a complete classification of 2D Ising models defined on isoradial graphs, frustrated or not, whose underlying spectral curve has genus 1. 
As a specific case, we recover Baxter's $Z$-invariant Ising model, thus extending his class of models to \emph{real} coupling constants. We identify two additional families of models, both having \emph{non}-Harnack spectral curves. We show that in all cases the spectral curve is maximal. Moreover, each family undergoes an algebraic phase transition, in the sense that the genus changes from one to zero, explaining the different behaviors observed in the physics literature. In our proof, we use properties of the spectral curve and Fock's approach. This yields a natural framework for a further systematic study of the frustrated Ising model, in particular for proving local formulas. In the course of the proof, we also identify Fock's dimer models corresponding to real algebraic curves of genus 1, and to real dimer models. As an example of application of our main result, we prove a full classification of the frustrated Ising model on the triangular lattice. 
\end{abstract}

\section{Introduction}\label{sec:intro}

Consider a planar, simple graph $G=(V,E)$ such that edges are assigned \emph{real coupling constants} $\eps\Js=(\eps_e\Js_e)_{e\in E}$ with $\eps_e \in \{\pm 1\}$ and $\Js_e \in \RR^+$. 
A \emph{spin configuration}, denoted by $\sigma$, is an element of $\{-1,1\}^V$. When the graph is moreover finite, the \emph{(frustrated) Ising Boltzmann measure with free boundary conditions}, denoted by $\PPising$, is the probability measure on the set of spin configurations defined by
\begin{equation*}
\forall\,\sigma\in\{-1,1\}^V,\quad \PPising(\sigma)=\frac{1}{\Zising(G,\eps\Js)}
e^{\sum_{vv'\in E}\eps_e\Js_e\sigma_v\sigma_{v'}},
\end{equation*}
where $\Zising(G,\eps\Js)$ is the normalizing constant known as the \emph{partition function}. 

Under this model, neighboring spins tend to agree, resp. disagree, when $\eps_e = 1$, resp. $-1$.
However, the fundamental parameter to understand the behaviour of the model is the \emph{frustration function} (a terminology due to P. W. Aspen), denoted by $\delta=(\delta_f)_{f\in F}$, and defined as the product of the signs on edges bounding faces: $\delta_f=\prod_{e \in \partial f}\eps_e$. Prior to describing our main contributions, let us give some historical background on this model, which has been a longstanding source of interest in the physics community. 

\paragraph{Historical perspective: phase transition in the frustrated Ising model}
The notion of frustrated model was coined by Toulouse in 1977 \cite{Toulouse} to describe a model where the spins in a system cannot find an orientation to fully satisfy all the interactions with their neighbors.
The antiferromagnetic Ising model on the triangular lattice has been one of the first frustrated model to attract the attention of physicists, even before the notion of frustration was introduced. 
Shortly after the computation of the free energy of the ferromagnetic Ising model by Onsager \cite{Ons44}, and building on subsequent work by Kauffman, Wannier \cite{Wan70} gave a physical derivation of the free energy of the isotropic antiferromagnetic Ising model on the triangular lattice.
He showed that this model exhibits no phase transition, but gave some heuristics arguments that the model is critical at $T=0$.
He also discusses the degeneracy of the ground state, which is one of the key features of frustrated models.

The possibility of a phase transition in the anisotropic frustrated Ising model on the triangular lattice (where the three types of orientations have different coupling constants $-\Js_1, -\Js_2, -\Js_3$) was first discussed by Houtappel \cite{Hou50} who gave a physical derivation of the free energy of the model using transfer matrix techniques.
He argues that depending on the values of $\Js$, there exists $0 < T_c < \infty$ such that the free energy at $\Js/T_c$ is not regular, \emph{i.e.}, there exists a phase transition, if and only if: either all the coupling constants are different, or the two largest interactions are equal in strength and differ from the smallest, that is $\Js_3 <\Js_2 = \Js_1$. This case was studied in more details in Eggarter~\cite{Egg75}.

In a series of papers \cite{Ste64, Ste70a, Ste70b} Stephenson, using Toeplitz determinants, computed the spin pair correlations in the three directions of the triangular lattice for any values of the coupling constants.
He proves that below the critical value $T_c$ of Houtappel, the model has long-range order, while it has exponential decay of correlations above $T_c$. 
He also proves the existence of a  ``disorder temperature" $T_D \geq T_c$, separating an antiferromagnetic (long-range or short range) order phase $T < T_D$ from a (short range) disorder phase.
This disorder temperature is not associated to a singularity of the free energy.
Stephenson also studies in detail the isotropic case $\Js_3 = \Js_2 = \Js_1$, and gives mathematical and numerical evidence that the correlations decay exponentially at $T>0$ and polynomially at $T=0$, thus predicting a phase transition at $T=0$.
This is related to the fact that the boundaries of clusters in the antiferromagnetic isotropic Ising model are expected to converge in the scaling limit towards CLE$_4$ at $T=0$ and CLE$_6$ at $T > 0$ \cite{BloNie89}. Similar behaviors have been observed for the frustrated Ising model on the square lattice~\cite{YangLee,Villain,VannimenusToulouse,LongaOles,Forgacs,ForgacsFradkin,AuYangPerk_2,PerkAuYang_1}. 

From a rigorous mathematical point of view, little is known about those models.
Recently, a rigorous mathematical derivation of the free energy of the anisotropic antiferromagnetic Ising model on the triangular lattice was obtained by \cite{AthUel25} using the Kac Ward method in the spirit of \cite{KagLisMee12}, formalizing the results of Wannier and Houtappel.
The authors also give a rigorous mathematical derivation of the fact observed by Houtappel that there is a phase transition in this model for $\Js_3<\Js_2=\Js_1$ but no phase transition for $\Js_3=\Js_2 \leq \Js_1$ (including the isotropic case).

In the physics community, the existence of a phase transition has also been discussed in related but different models.
For example the Ising model with competing nearest and next-to-nearest neighbor interactions was studied by \cite{VakLarOvc66} who prove that there may or may not exist a phase transition depending on the ratio of the nearest and next-nearest coupling constants.
The phase diagram of the antiferromagnetic Ising model with non-zero magnetic field has been the subject of a controversy \cite{LinWu79}, \cite{DocHem81}. 
The following papers provide an overview of the different lattices studied, and of the different regimes observed~\cite{Liebmann,Diep}.

Comparatively, the situation for the ferromagnetic Ising model is much more understood: let us only cite the results of Cimasoni and Duminil-Copin \cite{CimDum12} and Li \cite{Li10} who give an algebraic equation for the critical point of the planar Ising model on any periodic graph with any (positive) coupling constants. 
To our knowledge, on general graphs and with general coupling constants, it is not known whether there is a phase transition in the frustrated Ising model, and a fortiori there is no equivalent of the results of Cimasoni and Duminil-Copin \cite{CimDum12} and Li \cite{Li10}.

Our main result consists in classifying all possible (frustrated) Ising models defined on isoradial graphs, when the underlying spectral curve has genus 1,
by means of a related dimer model, using either the Fisher correspondence~\cite{Fisher} or bosonization~\cite{Dubedat,BoutillierdeTiliere:XORloops}.
The phase transition in the dimer model on bipartite graph with positive weights is famously described in Kenyon, Okounkov and Sheffield \cite{KO}.
The appearance of a phase transition is related to the disappearance of the hole in the amoeba of the spectral curve. 
To understand the link between phase transition and amoeba, a key point proven by \cite{KO} is the fact that the spectral curve of a dimer model on a bipartite graph with positive weights is of a special type, namely a simple Harnack curve.
However, the dimer model related to a frustrated Ising model by the Fisher correspondence is defined on a non-bipartite graph, while the model given by bosonization is defined on a bipartite graph with negative weights.
A related physics paper is \cite{NasDen17}, which is interested in classifying the behaviours of the amoeba related to dimer models on non-bipartite graphs or with negative weights.
Via simulations and concrete example, the authors describe the possible deviations of the behaviour of the amoeba from the well-studied Harnack case (singularities, shrinking, disappearance of the hole).
However, they give no general classification and to our understanding the particular spectral curves that we wish to study are not considered there.

\paragraph{Main contributions.} This paper proves a complete classification of the (frustrated) Ising model, where the bracketed term ``frustrated'' emphasizes that this theory includes both the frustrated and the non-frustrated cases; it is a first step in establishing a general and systematic theory of the (frustrated) Ising models. More precisely, we consider the (frustrated) Ising model defined on an infinite, $\ZZ^2$-periodic, isoradial graph $G$, with coupling constants $\eps\Js$ and frustration function $\delta$, and the corresponding Ising-dimer model on the bipartite graph $\GQ$~\cite{Dubedat,BoutillierdeTiliere:XORloops}. The spectral curve $\C$ of the (frustrated) Ising model is that of this bipartite dimer model; we suppose that the curve $\C$ has genus 1. We consider Fock's gauge equivalent dimer model~\cite{Fock} on the graph $\GQ$, with modular parameter $\tau$ (or equivalently elliptic modulus $k$), angle map $\mapalpha$, and parameter $t$; as we will see in Section~\ref{sec:Fock_necessary_condition_II}, there is an additional parameter $\rho$, coming from the central symmetry of the curve $\C$.  
Our main contribution, see Theorem~\ref{thm:main_intro} below and Theorem~\ref{thm:main} in the body of the paper for a precise statement, is a full classification of the (frustrated) Ising models using Fock's dimer models; our proof relies on properties of the spectral curve $\C$ and gauge equivalence of the two dimer models. 

\begin{thm}\label{thm:main_intro}
Consider an Ising model on an infinite, $\ZZ^2$-periodic, isoradial graph $G$ with periodic real coupling constants $\eps\Js$, and the corresponding Ising-dimer model on the graph $\GQ$. Suppose that the spectral curve $\C$ has genus 1 and is generic; and
consider Fock's gauge equivalent dimer model. Then necessarily, the parameters of Fock's dimer model satisfy the following: up to modular transformations, $k^2\in[0,1)$; the angle map $\mapalpha$ generically belongs to the real locus of the torus $\TT(k)$, and satisfies an additional sign condition\footnote{See Definition~\ref{def:necessary_V}.}; and the parameters $(\rho,t)$
fall in one of the classes $(\rho,t)$ of Table~\eqref{table3_intro}. They are related to the absolute value coupling constants $\Js$, resp. the frustration function $\delta$, by the third, resp. the fourth, row of Table~\eqref{table3_intro}.

{\small 
 \renewcommand\arraystretch{1.5}
 \begin{equation}\label{table3_intro}
   \begin{array}{|c||c|c|c|}
  \hline 
   (\rho,t) 
   &(\frac{1}{2},0)
   & (\frac{1}{2},\frac{1}{2})
   &(\frac{1}{2}+\frac{\tau}{2},\frac{1}{2})\\
   \hline \hline 
   \text{Hyp. $G/G^*$}
   & \text{none }
   & \text{none }
   & G^* \text{ bipartite}\\
   \hline
   \sinh(2\Js_e)
     & |\sc(\beta_e-\alpha_e)|
     & k'|\sc(\beta_e-\alpha_e)|
     & \frac{k'}{k}|\nc\bigl(\Re(\betah_e-\alphah_e)\bigr)|\\
   \hline
   \Js_e
    & \frac{1}{2}\ln\bigl(\frac{1+|\sn(\beta_e-\alpha_e)|}{|\cn(\beta_e-\alpha_e)|}\bigr)
    & \frac{1}{2}\ln\bigl(\frac{\dn(\beta_e-\alpha_e) +k'|\sn(\beta_e-\alpha_e)|}{|\cn(\beta_e-\alpha_e)|}\bigr)
    & \frac{1}{2}\ln\bigl(\frac{k'+\dn\Re(\betah_e-\alphah_e)}{k|\cn\Re(\betah_e-\alphah_e)|} \bigr)\\
   \hline
   \delta_f 
   & - \sgn (\prod_{i=1}^{|f|}\sn(\betah_{i}-\alphah_i))
   & - \sgn (\prod_{i=1}^{|f|}\sn(\betah_{i}-\alphah_i))
   & -1\\
   \hline 
   \end{array}
 \end{equation}
 }

Conversely, consider Fock's dimer model on an infinite, minimal graph $\GQ$ with parameters $k,\mapalpha,\rho,t$ satisfying the above conditions, and 
suppose that the angle map $\mapalpha$ is periodic. Consider the absolute value coupling constants $\Js$ and the frustration function $\delta=(\delta_f)_{f\in F}$ given by Table~\eqref{table3_intro}, and suppose that $\prod_{f\in F_1}\delta_f=1$. Then, there exists edge-signs $(\eps_e)_{e\in E}$ such that the Ising-dimer model on $\GQ$ arising from the coupling constants $\eps \Js$, and Fock's dimer model on $\GQ$, are gauge equivalent. Moreover, the associated spectral curve $\C$ has genus 1, is invariant by complex conjugation, is centrally symmetric, and Fock's dimer model is real. 
\end{thm}

In some respects, this result has the same flavour as the paper of~\cite{George} where George characterises spectral curves arising from Ising models with positive coupling constants. George uses the spectral data approach of~\cite{KO,GoncharovKenyon}, while we use Fock's approach~\cite{Fock} allowing for further probabilistic study of the model. Most notably, we consider \emph{real} coupling constants and not only positive ones, but we restrict to genus 1 while George's result holds for general genus.

The first significant feature of Theorem~\ref{thm:main_intro} is that this classification includes as a specific case the $Z$-invariant Ising model of Baxter~\cite{Baxter:8V,Baxter:Zinv,Baxter:exactly}, thoroughly studied in~\cite{ChelkakSmirnov2,BoutillierdeTiliere:iso_perio,BoutillierdeTiliere:iso_gen,BdTR2}; namely it consists of the first two columns of Table~\ref{table3_intro} in the case where the frustration function $\delta$ is identically equal to 1. It thus fully explains how Baxter's $Z$-invariant model actually occurs when using Fock's more general perspective, which is not restricted to the coupling constants being positive.

The second striking feature is that we identify three family of models. Family I is the known case of the $Z$-invariant Ising model of Baxter, it corresponds to the first two columns of Table~\ref{table3_intro} when the frustration function is identically equal to 1; Family II corresponds to the first two columns  when the frustration function also admits negative values; Family III corresponds to the third column of Table~\ref{table3_intro} which is fully frustrated and only occurs when the dual graph $G^*$ is bipartite. The modular parameter $\tau\in i\RR^+$ of Fock's dimer model corresponds to the elliptic modulus $k^2\in[0,1)$ of the elliptic functions of Table~\ref{table3_intro}; depending on the model, $k^2$ either plays the role of the temperature or of the inverse temperature. When $k^2$ tends to 0 and reaches it, the spectral curve $\C$ has genus 0; which we refer to as an \emph{algebraic phase transition} of the model. This is the content of Corollary~\ref{cor:algebraic_phase_transition}, in which we compute the~\emph{algebraic critical absolute value coupling constants} $\Js^{\mathrm{crit}}$ and obtain: 

{\small 
 \renewcommand\arraystretch{1.4}
 \begin{equation}\label{cor_table3_intro}
   \begin{array}{|c|c|c|c|}
  \hline 
   (\rho,t) 
   &(\frac{1}{2},0)
   & (\frac{1}{2},\frac{1}{2})
   &(\frac{1}{2}+\frac{\tau}{2},\frac{1}{2})\\
   \hline \hline 
   \text{Hyp. $G/G^*$}
   & \text{none }
   & \text{none }
   & G^* \text{ bipartite}\\
   \hline
   \sinh(2\Js_e^{\mathrm{crit}})
     & |\tan(\beta_e-\alpha_e)|
     & |\tan(\beta_e-\alpha_e)|
     & \infty\\
   \hline
   \Js_e^{\mathrm{crit}}
    & \frac{1}{2}\ln\bigl(\frac{1+|\sin(\beta_e-\alpha_e)|}{|\cos(\beta_e-\alpha_e)|}\bigr)
    & \frac{1}{2}\ln\bigl(\frac{1 +|\sin(\beta_e-\alpha_e)|}{|\cos(\beta_e-\alpha_e)|}\bigr)
    & \infty \\
   \hline
   \end{array}
 \end{equation}
 }

This sheds light on the different kinds of behavior observed in the physics literature described above, explaining clearly when one expects a phase transition with a critical temperature as in the ferromagnetic case, or when one expects a phase transition at 0 temperature. Physicists noted that the phase transition at 0 temperature could be observed on the triangular lattice but not on the hexagonal lattice; this is in accordance with Theorem~\ref{thm:main_intro} since the dual of the hexagonal lattice is not bipartite, implying that Family III cannot occur. Note however that the notion of algebraic phase transition is weaker  than that of phase transition as explained in Remark~\ref{rem:phase_transition}, and this should be explored further.

Another notable feature of Theorem~\ref{thm:main_intro} is that the models of Families II and III arise from spectral curves that go \emph{beyond simple Harnack curves}, but naturally arise in a model of  statistical mechanics. 
Indeed, one of the outstanding recent development on the dimer model on bipartite graphs (with positive edge-weights) establishes a correspondence between its spectral curves and simple Harnack curves~\cite{KO,KOS,GoncharovKenyon}. The rich underlying structure of Harnack curves provides a full description of the phase diagram of the model~\cite{KOS}; the associated Fock parameterization of the weights~\cite{Fock} has deep implications on probabilistic aspects of the model allowing to prove \emph{local formulas} for the free energy and the Gibbs measures~\cite{BCdT:genus_1,BCdT:genus_g}. In the context of the (frustrated) Ising model, the corresponding Ising-dimer model on the bipartite graph $\GQ$ possibly has negative weights. It is thus expected that we exit the realm of simple Harnack curves. As a byproduct of Theorem~\ref{thm:main_intro}, we obtain an explicit parameterization of the spectral curve $\C$, see Equation~\eqref{equ:param_spectral_curve_symm}, helping us to understand the features of the curve $\C$. 

When the spectral curve $\C$ is generic, the different behaviors are illustrated in Figure~\ref{fig:triangular_lattice_curves} below, relying on the example of the triangular lattice  with the smallest fundamental domain, see Figure~\ref{fig:dual_graph_bipartite_ter}. 

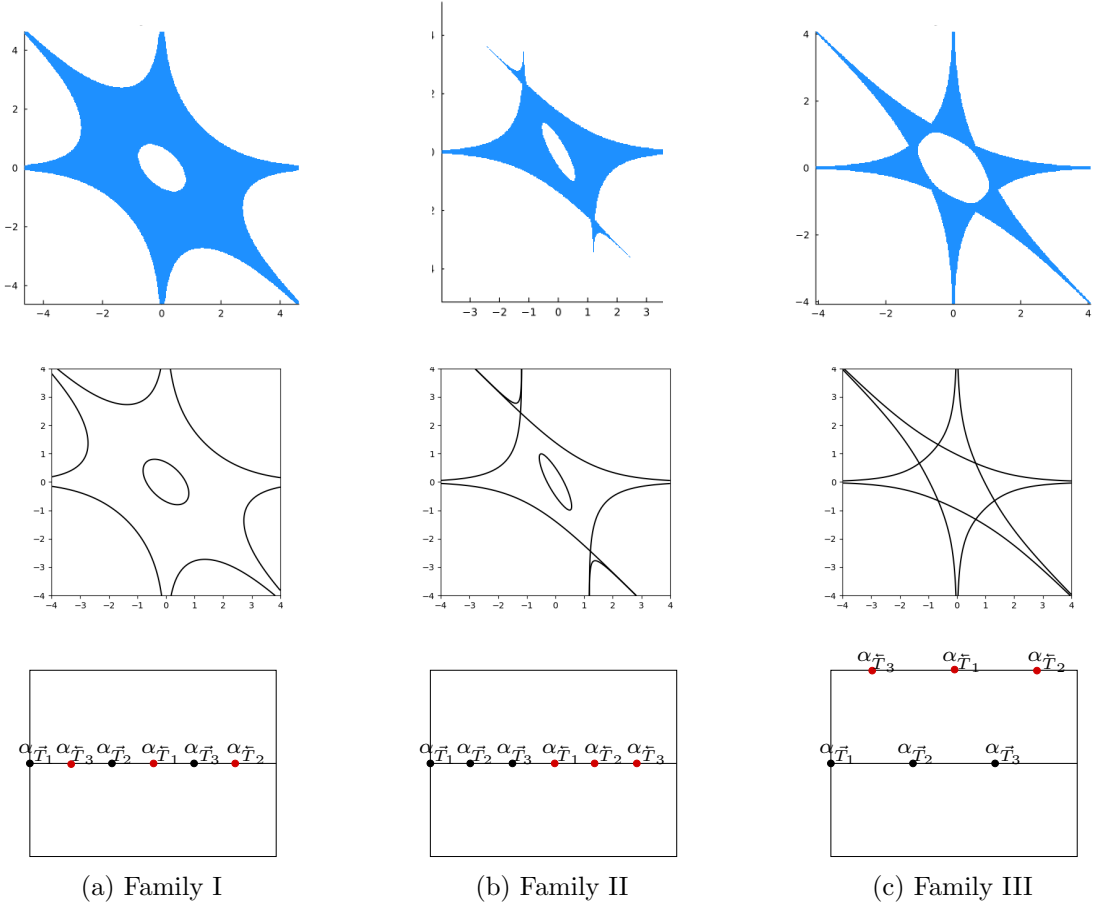
\begin{figure}[h!]
    \centering
    \begin{subfigure}{0.32\textwidth}
        \centering
        \includegraphics[scale=.3, trim = 3cm 0cm 3cm .8cm, clip]{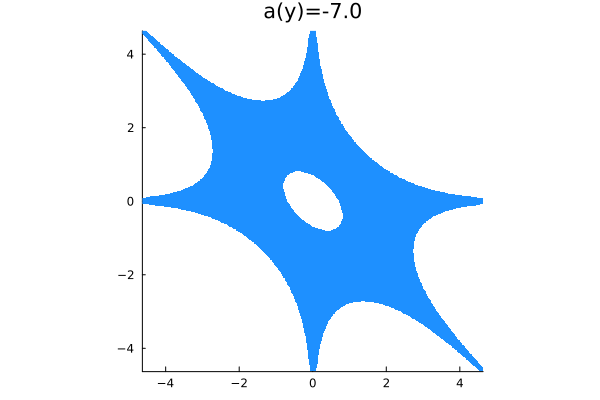}
    \end{subfigure}
    \hfill
    \begin{subfigure}{0.33\textwidth}
        \centering
        \includegraphics[scale=.33, trim = 6cm 0cm 6cm .8cm, clip]{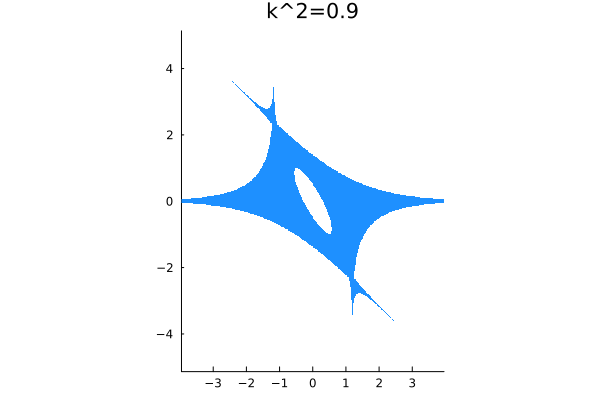}
    \end{subfigure}
    \hfill
    \begin{subfigure}{0.32\textwidth}
        \centering
        \includegraphics[scale=.3, trim = 3cm 0cm 3cm 0.8cm, clip]{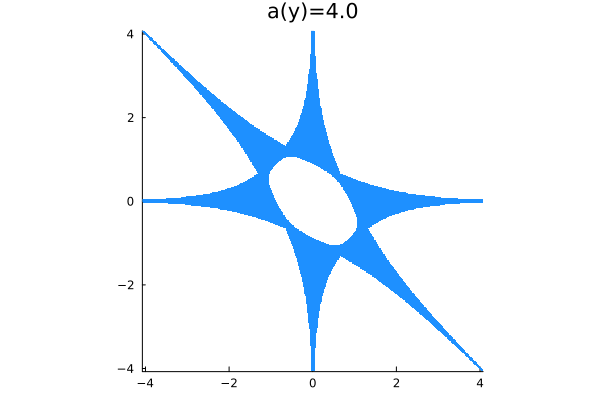}
    \end{subfigure}
\vspace{0.5cm}

\begin{subfigure}{0.32\textwidth}
        \centering
        \includegraphics[scale=.32, trim = 0cm 0cm 0cm 1.4cm, clip]{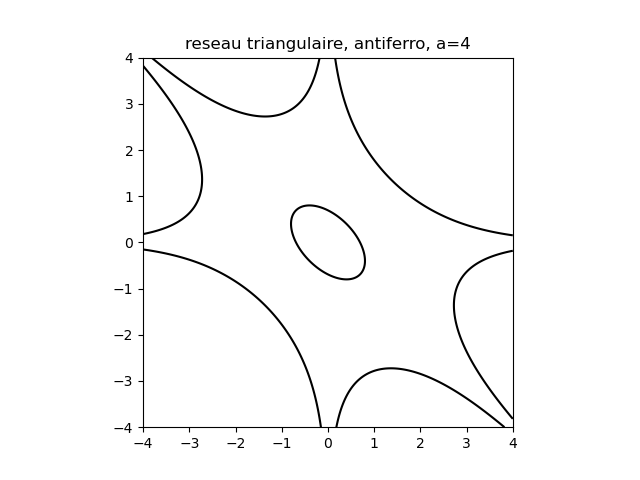}
    \end{subfigure}
    \hfill
    \begin{subfigure}{0.33\textwidth}
        \centering
        \includegraphics[scale=.32, trim = 0cm 0cm 0cm 1.4cm, clip] {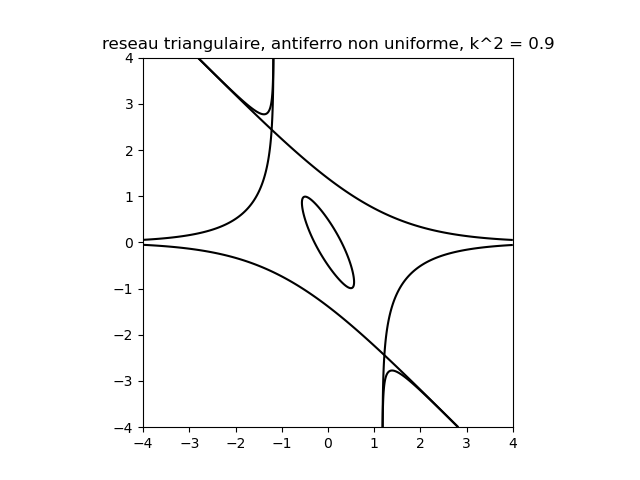}
    \end{subfigure}
    \hfill
    \begin{subfigure}{0.32\textwidth}
        \centering
        \includegraphics[scale=.32, trim = 0cm 0cm 0cm 1.4cm, clip]{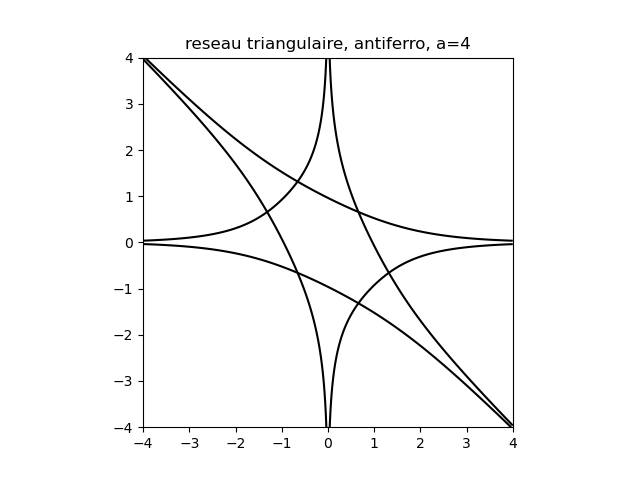}
    \end{subfigure}
\vspace{0.5cm}

      \begin{subfigure}{0.32\textwidth}
        \centering
        \begin{overpic}[width=3.3cm]{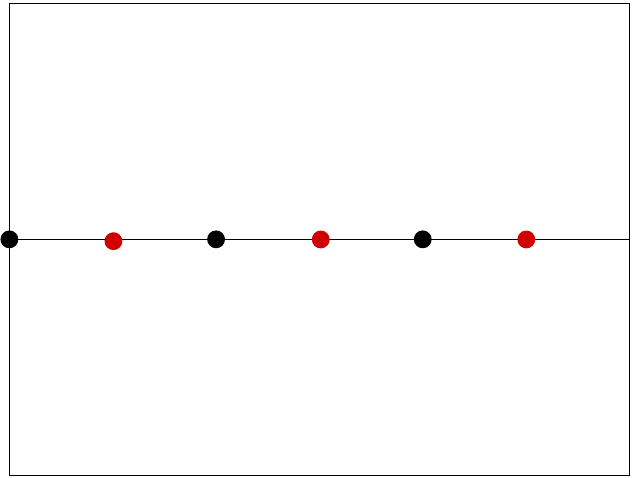}
        \put(-3,42){\scriptsize $\alpha_{\vec{T}_1}$} 
        \put(28,42){\scriptsize $\alpha_{\vec{T}_2}$} 
        \put(63,42){\scriptsize $\alpha_{\vec{T}_3}$} 
        \put(46,42){\scriptsize $\alpha_{\cev{T}_1}$} 
        \put(80,42){\scriptsize $\alpha_{\cev{T}_2}$} 
        \put(12,42){\scriptsize $\alpha_{\cev{T}_3}$} 
        \end{overpic}
  \caption{Family I}
    \end{subfigure}
    \hfill
    \begin{subfigure}{0.33\textwidth}
        \centering
        \begin{overpic}[width=3.3cm] {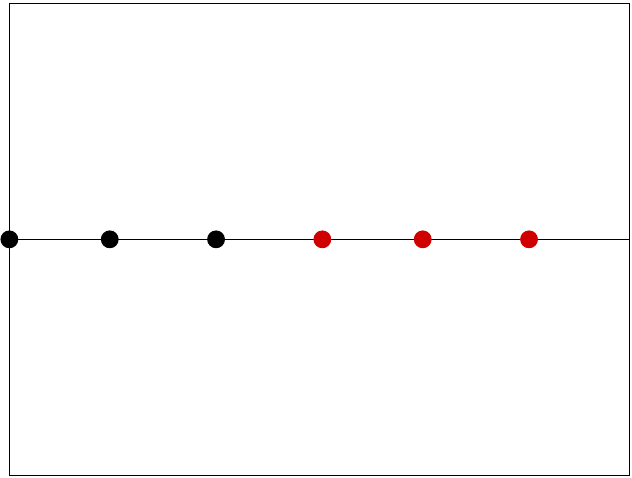}
         \put(-3,42){\scriptsize $\alpha_{\vec{T}_1}$} 
        \put(28,42){\scriptsize $\alpha_{\vec{T}_3}$} 
        \put(63,42){\scriptsize $\alpha_{\cev{T}_2}$} 
        \put(46,42){\scriptsize $\alpha_{\cev{T}_1}$} 
        \put(80,42){\scriptsize $\alpha_{\cev{T}_3}$} 
        \put(12,42){\scriptsize $\alpha_{\vec{T}_2}$} 
        \end{overpic}

        \caption{Family II}
    \end{subfigure}
    \hfill
    \begin{subfigure}{0.32\textwidth}
        \centering
        \begin{overpic}[width=3.3cm]{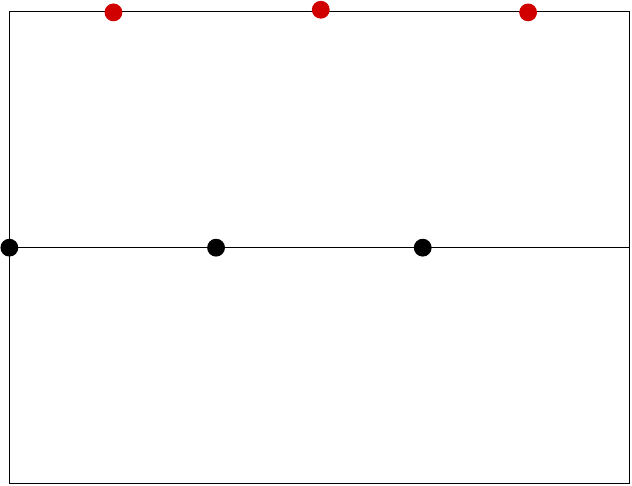}
        \put(-3,42){\scriptsize $\alpha_{\vec{T}_1}$} 
        \put(28,42){\scriptsize $\alpha_{\vec{T}_2}$} 
        \put(63,42){\scriptsize $\alpha_{\vec{T}_3}$} 
        \put(46,79){\scriptsize $\alpha_{\cev{T}_1}$} 
        \put(80,79){\scriptsize $\alpha_{\cev{T}_2}$} 
        \put(12,79){\scriptsize $\alpha_{\cev{T}_3}$} 
        \end{overpic}
        \caption{Family III}
    \end{subfigure}
    
    \caption{Spectral curves of the (frustrated) Ising model on the triangular lattice for the three possible families, represented through their amoebas (first row), the log of the modulus of their real locus (second row); the torus and angle map used to parameterize the spectral curve (third row). 
    Family I is an instance of an isotropic non-frustrated model; Family II is an instance of an anisotropic fully frustrated model; Family III is an instance of an isotropic fully frustrated model.
    Amoebas were plotted using the package PolynomialAmoebas of Sascha Timme.}
    \label{fig:triangular_lattice_curves}
\end{figure}

In order to get some grasp on the curve, we picture:\\ 
- on the top row, the amoeba of the curve $\C$, \emph{i.e.}, the image of $\C$ under the map $(z,w)\mapsto (\log|z|, \log|w|)$;\\
- on the second row, the image of the \emph{real locus} of the curve $\C$ through this same map;\\
- on the third row, the torus $\TT(\tau)$ and the angle map $\mapalpha$ used in the parameterization.\\
Indeed for non-Harnack curves the boundary of the amoeba is not the logarithm of the real locus~\cite{Mikhalkin,MikhalkinRullgard}, so this picture should help understand when the curve is Harnack.
Spectral curves of Family I are simple Harnack curves~\cite{BdTR2}. Spectral curves of Family II are close to Harnack curves: the zeros and poles of the parameterization live on a single component of the real locus of the torus $\TT(\tau)$, so that the hole in the amoeba seems to arise from the real component of $\TT(\tau)$ containing no zeros and no poles; the non-Harnacity comes from the crossing tentacles in the amoeba, and this occurs because the angle map $\mapalpha$ does not preserve the cyclic order of the train-tracks~\cite{Brugalle}. Family III is more subtle: the zeros and poles of the parameterization sit on both components of the real locus of $\TT(\tau)$, so that the log of the modulus of the real locus has two unbounded components, implying that the hole in the amoeba cannot be identified with a single real component of the torus $\TT(\tau)$. 

Another intriguing example of spectral curve $\C$ is given by that of the fully frustrated isotropic Ising model on the square lattice, see Section~\ref{subsec:square_lattice} and Figure~\ref{fig_amoeba_square}. The curve $\C$ is given by a reducible polynomial, written as a square of an irreducible polynomial. The remarkable fact is that the curve arising from this irreducible polynomial \emph{is Harnack}, so that the only feature that prevents the curve $\C$ from being Harnack is the fact that it is reducible. This is an illustration of the second part of Theorem~\ref{thm:main_intro}.


In the course of proving Theorem~\ref{thm:main_intro}, we prove three results of independent interest. In Corollary~\ref{cor:spectral_curve_max} we show that, if the spectral curve $\C$ of the (frustrated) Ising model has genus~1, is irreducible and non-singular except possibly at isolated real nodes, then it is \emph{maximal}, \emph{i.e.}, the number of connected components of the real locus is equal to $g+1=2$. This question was opened and, up to now, there was no consensus on whether this was true or not. Two other results have implications beyond the Ising model. Theorems~\ref{thm:real_spectral_curve} and~\ref{thm:Fock_real_dimer_model} characterize Fock's dimer models arising from \emph{real} dimer models thus extending~\cite[Theorem 34]{BCdT:genus_1} which is restricted to \emph{positive} weights; these results should have implications on the long standing open question of understanding the dimer model with positive weights on \emph{non-bipartite} graphs. Theorem~\ref{thm:real_symmetric_spectral_curve} characterizes Fock's dimer models arising from real dimer models on the bipartite graph $\GQ$, whose spectral curves have the additional property of being \emph{centrally symmetric}. Centrally symmetric Harnack curves arise when studying the massive Laplacian with positive conductances and masses~\cite{BdTR1}; our theorem extends this beyond the Harnack case.  More insight on these theorems is given in the ``Outline'' part of the introduction.

As said, this paper is a first step in establishing a general theory of the (frustrated) Ising model, but there remains many open questions. Here is a description of some. 
\begin{enumerate}
\item As exhibited in~\cite{BCdT:genus_1,Kenyon3}, Fock's parameterization of the dimer model gives the appropriate framework for studying the probabilistic model, and in particular prove explicit local expressions for the free energy and Gibbs measures. In the case of the $Z$-invariant Ising model, this was established in the papers~\cite{BoutillierdeTiliere:iso_gen,BdTR2}. An important open question is to give a full probabilistic description involving local formulas for the models classified in Theorem~\ref{thm:main_intro}. Most key identities used in the above mentioned papers extend to negative edge-weights. Nevertheless, these are only one of the building blocks, another being the Harnack property of the underlying spectral curve, used in many instances, in particular in the definition of the integration contour of the inverse Kasteleyn matrix~\cite{BoutillierdeTiliere:iso_gen,BdTR2}. Also, negative weights imply that some of the
probabilistic arguments collapse, and that the famous correlation inequalities~\cite{Griffiths,KellySherman,FortuinKasteleynGinibre} do not hold.
\item We conjecture that all models of the first two columns of Table~\ref{table3_intro}, consisting of Families I and II, are critical at $k=0$. 
In the case of Family II, the curve is not Harnack, but we believe that the behavior at the critical point is the same as in the Harnack case of Family I, thus raising the following questions/conjectures for Family II.
\begin{enumerate}
\item Prove that the algebraic critical absolute value coupling constants given by
\[\sinh(2\Js_e^{\mathrm{crit}})=|\tan(\beta_e-\alpha_e)|\] 
are indeed critical\footnote{The fact that the critical absolute value coupling constants for the first two columns are the same is natural and comes from a duality argument}. 
This would extend part of the results of~\cite{Li10,CimDum12} to real coupling constants, and the results of~\cite{AthUel25} beyond the case where two of the coupling constants are equal.
\item Prove that there is a second order phase transition of the free energy at $k=0$, as in~\cite{BdTR2}, see also~\cite{BdTR1}.
\item In the case of positive coupling constants, at the critical temperature on isoradial graphs, 
the Ising model has famously been proved to be conformally invariant by Smirnov~\cite{Smirnov}, and Chelkak and Smirnov~\cite{ChelkakSmirnov2,ChelkakSmirnov1}. We believe that this should still be true in the case of Family II. Proving this result would certainly be a challenge since no correlation inequalities are available, another problem being the definition of the limiting SLE process.
\end{enumerate}
\item The model of Family III is also critical at $k=0$, but has a different behavior.
In particular, it is always associated to infinite coupling constants, i.e. $0$ temperature, see Table~\ref{cor_table3}.
This should explain the phase transition observed at $T=0$ in the isotropic frustrated model \cite{Ste70a} with exponential decay of correlations in the distance $r$ as $e^{-r/\zeta}$ at $T > 0$ versus polynomial decay as $r^{-1/2}$ at $T = 0$.
Can we more generally try and identify to what model Family III corresponds ?
Can we also see the existence of the disorder temperature $T_D$ of Stephenson~\cite{Ste70b} on our parameterization ?
\item In all cases, when $k\neq 0$, the underlying spectral curve has genus 1. From simulations, it seems obvious that the amoeba has a hole; moreover by Corollary~\ref{cor:spectral_curve_max} we know that the curve is maximal. Nevertheless, when one exits the realm of Harnack curves, the boundary of the amoeba does not coincide with the log of the absolute value of the real locus~\cite{Mikhalkin,MikhalkinRullgard} and, for the moment, we are not able to prove that the amoeba indeed has a hole. If this was the case, it would prove that the free energy is analytic for all values of $k\neq 0$.
\item For the Ising model with positive coupling constants, there are equivalent ways of characterizing the phase transition: uniqueness/non-uniqueness of the Gibbs measures, spontaneous magnetization~\cite{Peierls,Griffiths}, decay of correlations~\cite{Aizenman}, change of regularity in the free energy~\cite{Baxter:exactly,Li10,CimDum12}, see also the book~\cite{Friedli_Velenik} and references therein. Using the approach of this paper, the most natural path is to use the free energy. But an interesting open question is to establish equivalences between the different points of view as in the case of positive coupling constants.
\item We believe that the coupling constants of Table~\ref{table3_intro} are all $Z$-invariant as established by Baxter in the case of Family I~\cite{Baxter:exactly}. Given the length of this paper, we decided no to perform this computation here. 
\item In Section~\ref{sec:examples}, we provide examples of angle maps $\mapalpha$ satisfying the conditions of the statement of Theorem~\ref{thm:main_intro}. Nevertheless, an open question is to characterize the latter using the immersion of the minimal graph $\GQ$. This would be an extension of the paper~\cite{BCdT:immersions} where we identified immersions of minimal graphs related to dimer models with \emph{positive} weights. 
\end{enumerate}

\paragraph{Outline of the paper.} Section~\ref{sec:background} is dedicated to background, and Section~\ref{sec:real_dimers} to characterizing Fock's dimer models arising from algebraic curves given by polynomials with real coefficients. Section~\ref{sec:classification} is then the heart of the paper; we establish five necessary conditions and then show that they are sufficient, thus proving our classification result, Theorem~\ref{thm:main_intro}. Section~\ref{sec:examples} gives examples, in particular we provide a full classification of the (frustrated) Ising model on the triangular lattice in Theorem~\ref{thm:antiferro:triangular:general}. More precisely, we have the following. 

\paragraph{Section~\ref{sec:background}.} Section~\ref{sec:dimer_model} gives the required background on the dimer model. When the graph is $\ZZ^2$-periodic, we furthermore define the characteristic polynomial, the spectral curve and the associated amoeba~\cite{KOS}. In Section~\ref{sec:frustrated_Ising_model} we define the (frustrated) Ising model on a graph $G$ with coupling constants $\eps\Js$ and frustration function $\delta$~\cite{Toulouse}, and give some useful properties. Section~\ref{sec:Ising_dimer model_representation} describes the two dimer representations of the (frustrated) Ising model: Fisher's representation on the non-bipartite Fisher graph $\GF$~\cite{Fisher}, and the bipartite dimer model representation on the graph $\GQ$~\cite{Dubedat,BoutillierdeTiliere:XORloops}. When the graph $G$ is moreover $\ZZ^2$-periodic,
the two dimer characteristic polynomials arising from the two dimer representations are equal up to a constant~\cite{Dubedat}; they thus give the same spectral curve $\C$, also referred to as the \emph{spectral curve of the (frustrated) Ising-dimer model}. This implies that the curve $\C$ arises from a polynomial with \emph{real coefficients} and is \emph{centrally symmetric}.

\paragraph{Section~\ref{sec:real_dimers}.} As said, the spectral curve of the (frustrated) Ising-dimer model arises from a polynomial with real coefficients. This section places itself in the more general framework of algebraic curves $\C$ of genus 1 arising from polynomials with real coefficients, but not necessarily from an Ising model. We further suppose that the algebraic curve has genus 1, and satisfies Assumptions $(\dagger)$, \emph{i.e.} it is irreducible, and non-singular except possibly for isolated real nodes. By Fock~\cite{Fock} we know that, up to scale change, the curve $\C$ is the spectral curve of a dimer model with Fock's weights in genus 1 on an infinite, periodic minimal graph $\Gs$. Fock's dimer model has three sets of parameters: an underlying torus $\TT(\tau)$ with modular parameter $\tau$, an angle map $\mapalpha$, and an additional parameter $t$. 

Section~\ref{sec:Fock_genus_1} is dedicated to background on Fock's dimer model in genus 1~\cite{Fock}: we define train-tracks, minimal graphs, isoradial graphs, the angle map, the discrete Abel map and genus 1 Riemann surfaces and theta functions, modular transformations. In Lemma~\ref{lem:invariance_weights}, we prove that face-weights are meromorphic functions of the parameters, and in Lemma~\ref{lem:modular_transo_dimers}, we identify the evolution of the model under modular transformations of the parameter $\tau$.

Since the polynomial defining the algebraic curve $\C$ has real coefficients, $\C$ is a \emph{real algebraic curve} with the real structure given by complex conjugation; we also know that Fock's associated dimer model is \emph{real}, see Definition~\ref{defi:dimer_extras}. These two features allow us to prove the following. We refer to the body of the paper for precise statements. 
\begin{itemize}
 \item[$\circ$] Theorem~\ref{thm:real_spectral_curve} of Section~\ref{sec:real_spectral_curve} identifies the values of the modular parameters $\tau$ and of the angle map $\mapalpha$ such that the algebraic curve $\C$ is real: for generic $\tau$, and up to modular transformations, $\tau\in i\RR^+\cup \bigl\{\frac{1}{2}+i\RR^+\,:\Im(\tau)\geq\frac{\sqrt{3}}{2}\bigr\}$, and in each case the angle map $\mapalpha$ belongs to the real locus of the torus $\TT(\tau)$. 
 \item[$\circ$] Theorem~\ref{thm:Fock_real_dimer_model} of Section~\ref{sec:Fock_real_dimer_model} 
 proves that Fock's dimer model is furthermore real if the parameter $t$ belongs to the real locus of the torus $\TT(\tau)$. 
\end{itemize}
Theorems~\ref{thm:real_spectral_curve} and~\ref{thm:Fock_real_dimer_model} extend~\cite[Proposition 13, Theorem 34]{BCdT:genus_1} where we identified parameters of Fock's dimer model on minimal graphs that are gauge equivalent to dimer models with \emph{positive edge-weights}. 

\paragraph{Section~\ref{sec:classification}.} We consider a (frustrated) Ising model defined on an infinite, periodic, isoradial graph $G$, with coupling constants $\eps\Js$, and frustration function $\delta$. We also consider the corresponding Ising-dimer model on the minimal graph $\GQ$ with spectral curve $\C$ of genus 1, and Fock's gauge equivalent dimer model on $\GQ$, with parameters $\tau,\mapalpha,t$. The goal of this section is to establish all necessary conditions satisfied by Fock's dimer model on $\GQ$ in order for it to be gauge equivalent to the Ising-dimer model on $\GQ$, and then prove that these conditions are sufficient in Theorem~\ref{thm:main}. More precisely, 

\begin{enumerate}
 \item[$\circ$] Since the Ising model spectral curve $\C$ arises from a polynomial with real coefficients, Definition~\ref{def:Focks_dimer_model} introduces \textbf{Necessary Condition I}, asking that the parameters $\tau,\mapalpha$ satisfy the conditions of Theorems~\ref{thm:real_spectral_curve} and \ref{thm:Fock_real_dimer_model}, \emph{i.e.}, that generically, and up to modular transformations $\tau\in i\RR^+\cup \bigl\{\frac{1}{2}+i\RR^+\,:\Im(\tau)\geq\frac{\sqrt{3}}{2}\bigr\}$, and that $\mapalpha,t$ belong to the real locus of the torus $\TT(\tau)$.
 \item[$\circ$] Section~\ref{sec:Fock_necessary_condition_II}. By Dubédat~\cite{Dubedat}, we know that the spectral curve $\C$ is furthermore \emph{centrally symmetric}, that is $(z,w)\in\C\leftrightarrow\bigl(\frac{1}{z},\frac{1}{w}\bigr)\in\C$. From the construction of the set of oriented train-tracks $\vec{\T}$ of $\GQ$ from the set of (unoriented) train-tracks $\T$ of $G$ we know that, to every train-track $T$ of $\T$, is assigned a pair $\vec{T},\cev{T}$ of oriented train-tracks of $\vec{\T}$. This implies that angles of the angle map~$\mapalpha$ come in pairs $(\alpha_{\vec{T}},\alpha_{\cev{T}})$. Theorem~\ref{thm:real_symmetric_spectral_curve} identifies the angles maps encoding central symmetry of the curve: for every train-track $T$ of $\T$, 
 \[
 \alpha_{\vec{T}}=\alpha_{\cev{T}}+\rho, 
 \]
 where $\rho=\frac{j}{2}+\frac{\ell}{2}\tau$, and $(j,\ell)\in\{(1,0),(0,1),(1,1)\}$ when $\tau\in i\RR^+$, and $(j,\ell)=\frac{1}{2}$ when $\bigl\{\frac{1}{2}+i\RR^+\,:\Im(\tau)\geq\frac{\sqrt{3}}{2}\bigr\}$.
This additional condition on the angle map $\mapalpha$ is the content of \textbf{Necessary Conditions II} of Definition~\ref{def:necessary_II}; the parameters of Fock's dimer model are now encoded as $\tau,\mapalpha,\rho,t$. In Lemma~\ref{lem:alternate_product_faces} we explicitly compute the corresponding face-weights of Fock's dimer model on $\GQ$, which are of three types: square faces of $\GQ$ corresponding to edges $y=\{e,e^*\}$ of $G$, faces of $\GQ$ corresponding  to faces $f$ of $G$, and those corresponding to vertices $v$ of $G$. Lemma~\ref{lem:duality} describes duality properties of Fock's dimer model on the graph $\GQ$.

The Ising-dimer model on $\GQ$ and Fock's dimer model on $\GQ$  are gauge equivalent if and only if their face-weights are equal for all faces. We then successively study implications of the equality at the three different kinds of faces. In detail we obtain. 

\item[$\circ$] Section~\ref{sec:Necessary_III}: Equality of face-weights at square faces. In Proposition~\ref{prop:maximality_spectral_curve}, we prove that the modular parameter necessarily belongs to $i\RR^+$, thus implying maximality of the spectral curve $\C$, see Corollary~\ref{cor:spectral_curve_max}. We also prove conditions on the parameter $\rho=\frac{j}{2}+\frac{\ell}{2}\tau$ encoding the central symmetry of the curve $\C$: necessarily $(j,\ell)\in\{(1,0),(1,1)\}$. Combining this with information on $\Im(t):=\ell_t/2$, taking values in $\{0,\frac{\tau}{2}\}$, we infer possible values for the angle map: $\mapalpha$ must belong to a set $\Acal_{\rho,\ell_t}$, implying in some cases \emph{bipartitedness restriction} on the graph $G$ and/or its dual graph $G^*$. This is the content of Propositions~\ref{prop:first_parameter_condition} and~\ref{prop:global_condition_anglemap}, and these additional conditions are tracked in \textbf{Necessary Conditions III}, see Definition~\ref{defi:necessary_III}. 

\item[$\circ$] Section~\ref{sec:Necessary_IV}: Equality of \emph{squared} face-weights at primal/dual faces. With these conditions, we prove additional restrictions on the parameter $t$. In Proposition~\ref{prop:t}, we show that either $\rho=\frac{1}{2}$ and $t\in\{0,\frac{1}{2}\}$, or $\rho=\frac{1}{2}+\frac{\tau}{2}$ and $t\in\{\frac{1}{2},\frac{\tau}{2}\}$. This defines \textbf{Necessary Conditions IV}, see Definition~\ref{def:necessary_IV}. At this point, face-weights of Fock's dimer model on $\GQ$ have simple expressions in terms of the Jacobi elliptic functions; this is the subject of Lemma~\ref{lem:face_weights_again}. The modular parameter $\tau\in i\RR^+$ is equivalent to considering the elliptic modulus $k^2\in[0,1)$. 

\item[$\circ$] Section~\ref{sec:Necessary_V}: Equality of face-weights at primal/dual faces. With these last conditions, we finally exclude the case $\rho=\frac{1}{2}+\frac{\tau}{2}$ and $t=\frac{\tau}{2}$, and establish an additional condition on the angle map $\mapalpha$. This is the content of \textbf{Necessary Conditions V}, see Definition~\ref{def:necessary_V}. 

\item[$\circ$] Section~\ref{sec:classification_thm}. Statement and proof of our main classification result, Theorem~\ref{thm:main}. Definition and discussion of the notion of algebraic phase transition, and algebraic critical absolute value coupling constants.
\end{enumerate}

\paragraph{Section~\ref{sec:examples}.} In Section~\ref{sec:monotone_angle_maps}, we provide examples of models satisfying the conditions of Theorem~\ref{thm:main}, using the notion of monotone angle maps studied in~\cite{BCdT:immersions}. Section~\ref{sec:triangular_lattice} is dedicated to the famous frustrated Ising model on the triangular lattice; a full classification is proved in Theorem~\ref{thm:antiferro:triangular:general}. 
The intriguing example of the frustrated isotropic model on the square lattice is the subject of Section~\ref{subsec:square_lattice}.

\paragraph{Acknowledgments.} We would like to thank Cédric Boutillier, Alexander Glazman, and Richard Kenyon for interesting discussions on the subject. We acknowledge having used AI to find relevant references on parameterization of real algebraic elliptic curves. 

\section{Background}\label{sec:background}

This section aims at giving the necessary background. Section~\ref{sec:dimer_model} is dedicated to the dimer model. After having defined the model, we introduce Kasteleyn matrices~\cite{Kast61,TemperleyFisher,Kasteleyn} and, when the underlying graph is bipartite, the notion of gauge equivalence~\cite{KOS}. When the underlying graph is $\ZZ^2$-periodic, we recall the definition of the characteristic polynomial, the spectral curve and its amoeba~\cite{KOS}. In Section~\ref{sec:frustrated_Ising_model} we define the (frustrated) Ising model, also leading to a notion of gauge equivalence related to the frustration function~\cite{Toulouse}. In Section~\ref{sec:Ising_dimer model_representation} we describe two useful dimer model representations of the (frustrated) Ising model: Fisher's correspondence~\cite{Fisher}, and the bipartite dimer representation~\cite{Dubedat,BoutillierdeTiliere:XORloops}. Note that in order to read Section~\ref{sec:real_dimers}, one only needs the content of Section~\ref{sec:dimer_model}; Sections~\ref{sec:frustrated_Ising_model} and~\ref{sec:Ising_dimer model_representation} give the necessary background for Section~\ref{sec:classification}.

\subsection{The dimer model}\label{sec:dimer_model}

\subsubsection{Definition and founding tools}\label{subsec:dimer_model_defi}

The \emph{dimer model} represents the adsorption of diatomic molecules on the surface of a crystal~\cite{FR}. It is defined as follows. 

\paragraph{Definition.} Consider a planar, simple graph $G=(V,E)$. A \emph{dimer configuration} of $G$ is a subset of edges $\Ms$ such that each vertex of $G$ is incident to exactly one edge of $\Ms$; let $\M(G)$ denote the set of dimer configurations of $G$. Assume a complex weight function $\nu$ is assigned to edges of $G$. When the graph $G$ is moreover finite, the \emph{dimer Boltzmann measure}, denoted by $\PPdimer$, is a (complex) measure on $\M(G)$ defined by:
\[
\forall\,\Ms\in\M(G),\quad \PPdimer(\Ms)=\frac{\prod_{e\in E} \nu_e}{\Zdimer(G,\nu)},
\]
where $\Zdimer(G,\nu)=\sum_{\Ms\in\M(G)}\prod_{e\in \Ms}\nu_e$ is the normalizing constant known as the \emph{partition function}. Note that this measure has total mass 1, but configurations are allowed to have complex weights for the moment. 

\paragraph{Kasteleyn matrix.}
One of the main tools used to study the dimer model is the \emph{Kasteleyn matrix}, denoted by $K$, and defined as follows~\cite{Kast61,TemperleyFisher,Kasteleyn}. It has rows and columns indexed by vertices of $V$, and non-zero coefficients corresponding to edges of $E$, given by 
\[
\forall\, vv'\in E,\quad K(v,v')=\eta_{(v,v')}\nu_{vv'},
\]
where $\eta=(\eta_{(v,v')})_{vv'\in E}$ are the \emph{Kasteleyn signs}: for every edge $vv'$ of $E$, $\eta_{(v,v')}\in\{-1,1\}$, $\eta_{(v,v')}=-\eta_{(v',v)}$; moreover, for every inner face $f$ of $G$ of degree $|f|$, if we denote by $v_1,\dots,v_n$ its boundary vertices in counterclockwise order, the Kasteleyn signs $\eta$ satisfies the \emph{Kasteleyn condition}:
\begin{equation}\label{equ:Kast_condition_I}
\prod_{j=1}^{|f|} \eta_{(v_j,v_{j+1})}=(-1)^{|f|-1}. 
\end{equation}
By~\cite{Kasteleyn}, see also~\cite{CimRes07}, Kasteleyn signs $\eta$ exist when the graph $G$ is finite or infinite, or when it is embedded in the torus and has an even number of vertices. 
Note that the Kasteleyn matrix $K$ is skew-symmetric. 

If the graph is moreover bipartite, we use the specific notation 
$\Gs=(\Vs,\Es)$. Its vertex set is naturally split into white and black: $\Vs=\Ws\sqcup\Bs$. Following~\cite{Kuperberg}, edges can more generally be assigned modulus one complex vectors $\phi=(\phi_{\ws\bs})_{\ws\bs\in\Es}$ instead of signs $\eta$, referred to as \emph{Kasteleyn-Kuperberg phases}, and the Kasteleyn matrix can be restricted to white, resp. black vertices~\cite{Percus}. More precisely, the \emph{(bipartite) Kasteleyn matrix}, denoted by $\Ks$, has rows indexed by white vertices, columns by black ones, and non-zero coefficients corresponding to edges of $\Es$, given by
\begin{equation}\label{equ:Kuperberg}
\forall\,\ws\bs\in\Es,\quad \Ks(\ws,\bs)=\phi_{\ws\bs}\nu_{\ws\bs},
\end{equation}
where, for every edge $\ws\bs$ of $\Es$, $\phi_{\ws\bs}\in\CC$, $|\phi_{\ws\bs}|=1$; moreover, for every inner face $\fs$ of $\Gs$ of degree $|\fs|$ (where $|\fs|$ is even since $\Gs$ is bipartite), if we denote by $\ws_1,\bs_1,\dots,\ws_{|\fs|/2},\bs_{|\fs|/2}$ its vertices in counterclockwise order, the Kasteleyn-Kuperberg phases $\phi$ satisfy the \emph{Kasteleyn condition}:
\begin{equation}\label{equ:Kast_condition_II}
\prod_{j=1}^{|\fs|/2} \frac{\phi_{\ws_j\bs_{j}}}{\phi_{\ws_{j}\bs_{j-1}}}=(-1)^{{|\fs|/2}+1}. 
\end{equation}

We now turn to the notion of \emph{gauge equivalent dimer models}, which holds when the graph $\Gs$ is bipartite. In the context of the dimer model, this has been introduced and studied by Kenyon, Okounkov and Sheffield~\cite{KOS}. 

\begin{defi}\label{defi:dimer_extras}\leavevmode
\begin{enumerate}
\item[$\bullet$] Let $\tilde{\Ks}$ be a complex-valued matrix whose rows are indexed by white vertices, columns by black ones, and whose non-zero coefficients correspond to edges of $\Gs$. For every face $\fs$ of degree $|\fs|$ of $\Gs$, using the same notation as above, the \emph{alternate product around the face $\fs$}, simply referred to as the \emph{face-weight} at $\fs$, is denoted by $\W(\fs)$, and defined by
\begin{equation}\label{equ:face_weight_gen}
\W(\fs)=\prod_{j=1}^{|\fs|/2} \frac{\tilde{\Ks}(\ws_j,\bs_j)}{\tilde{\Ks}(\ws_{j},\bs_{j-1})}. 
\end{equation}
\item[$\bullet$] Suppose that edges of $\Gs$ are assigned Kasteleyn-Kuperberg phases $\phi$. Then, using Equation~\eqref{equ:Kuperberg}, the matrix $\tilde{\Ks}$ defines a weight function $\tilde{\nu}$, and thus a dimer model. We speak of the \emph{dimer model arising from the Kasteleyn matrix $\tilde{\Ks}$}.
\item[$\bullet$] The dimer models on the graph $\Gs$ arising from the Kasteleyn matrices $\tilde{\Ks}$ and $\Ks$ are said to be \emph{gauge equivalent}~\cite{KOS} if, for every face $\fs$ of $\Gs$,
\begin{equation*}
\W(\fs)=(-1)^{|\fs|/2-1}\prod_{j=1}^{|\fs|/2} \frac{\nu_{\ws_j\bs_j}}{\nu_{\ws_{j}\bs_{j-1}}}.
\end{equation*}
\item[$\bullet$] The dimer model on the graph $\Gs$ arising from the Kasteleyn matrix $\tilde{\Ks}$ is said to be~\emph{real} if it is gauge equivalent to a dimer model with real edge-weights, in particular, for every face $\fs$ of $\Gs$, the face-weight $\W(\fs)$ of Equation~\eqref{equ:face_weight_gen} is real. Then, if $\W(\fs)\in (-1)^{|\fs|/2-1}\RR^+$, it is gauge equivalent to \emph{positive} edge-weights around the face $\fs$, else if $\W(\fs)\in (-1)^{|\fs|/2}\RR^+$, it is gauge equivalent to \emph{negative} edge-weights around the face $\fs$.
\end{enumerate}
\end{defi}

We refer to the paper~\cite{KOS} for the proof of the next lemma.
\begin{lem}[\cite{KOS}]
Consider a finite, planar, simple bipartite graph $\Gs$, and two gauge equivalent dimer models on $\Gs$ arising from Kasteleyn matrices $\Ks$ and $\tilde{\Ks}$. If the two models are gauge equivalent then the partition functions $\Zdimer(\Gs,\nu)$ and  $\Zdimer(\Gs,\tilde{\nu})$ are equal up to an explicit constant, and the dimer Boltzmann measures are the same. 
\end{lem}

\subsubsection{The dimer model in the $\ZZ^2$-periodic case}\label{subsec:periodic_dimer}

The graph $G$ is said to be \emph{$\ZZ^2$-periodic}, or simply \emph{periodic}, if it is invariant under the action of a 2-dimensional lattice referred to as $\ZZ^2$. The graph has a natural toroidal exhaustion $(G_{n})_{n\geq 1}$, where $G_n=G/n\ZZ^2$; the graph $G_1$ is known as the \emph{fundamental domain}. We use similar notation for the toroidal exhaustion of the dual graph $G^*$. 
From now on, we suppose that $G_1$ has an even number of vertices, this being necessary for $G_1$ to have a dimer configuration. 

Suppose that the graph $\Gs=(\Vs,\Es)$ is moreover bipartite, then it is said to be \emph{$\ZZ^2$-periodic}, or simply \emph{periodic}, if furthermore the bipartite coloring is invariant under the action of $\ZZ^2$. 

When the dimer weight function is positive, the key object to understand the model on periodic graphs is the \emph{characteristic polynomial} containing information about the free energy, Gibbs measures and the phase diagram (in the bipartite case), see for example~\cite{CKP,KOS}. This object is also central to this paper where more general weight functions are considered, and we now recall its definition. 

\paragraph{Characteristic polynomials.} Fix a face $f$ of $G$, and draw two simple paths in the plane, denoted by $\gamma_x$, resp. $\gamma_y$, joining $f$ to $f+(1,0)$, resp. $f$ to $f+(0,1)$, intersecting only at $f$. They project onto two simple closed loops in $G_1^*$ also denoted by $\gamma_x,\gamma_y$. Their homology classes $[\gamma_x],[\gamma_y]$ form a basis of the first homology group of the torus $H_1(\TT,\ZZ)$. The construction of the characteristic polynomial differs slightly depending on whether $G$ is bipartite or not. 

Suppose first that $G$ is not bipartite. Consider a dimer model on $G$ with (complex) periodic weight function $\nu$, periodic Kasteleyn signs $\eta$ assigned to edges,
and let $K_1$ be the Kasteleyn matrix of the graph $G_1$. Then, multiply by $w^{-1}$, resp. $w$ coefficients of edges crossing $\gamma_x$ from left-to-right, resp. from right-to-left. Similarly, multiply by $z$, resp. $z^{-1}$ coefficients of edges crossing $\gamma_y$ from left-to-right, resp. from right-to-left. This gives a modified weight Kasteleyn matrix $K_1(z,w)$ whose coefficients are Laurent polynomials in the variables $z,w\in \CC^*$. The \emph{characteristic polynomial}~\cite{KOS}, denoted by $P(z,w)$, is defined as
\begin{equation}\label{equ:charact_non_bipartite}
P(z,w)=\det(K_1(z,w)).
\end{equation}

\begin{rem}
Since the matrix $K_1$ is skew-symmetric, the matrix $K_1(z,w)$ has the symmetry $K_1(z,w)^t=-K_1\bigl(\frac{1}{z},\frac{1}{w}\bigr)$, implying \emph{central symmetry of the characteristic polynomial}:
\begin{equation}\label{equ:central_symmetry}
P(z,w)=(-1)^{|V_1|}P\Bigl(\frac{1}{z},\frac{1}{w}\Bigr)=P\Bigl(\frac{1}{z},\frac{1}{w}\Bigr),
\end{equation}
because the number of vertices $|V_1|$ of the fundamental domain $G_1$ is even.
\end{rem}

Suppose now that the graph $\Gs$ is bipartite. Consider a dimer model on $\Gs$ with (complex) periodic weight function $\nu$, periodic Kasteleyn-Kuperberg phases $\phi$ assigned to edges, and let $\Ks_1$ be the Kasteleyn matrix of the graph $\Gs_1$. Then, multiply by $w^{-1}$, resp. $w$ coefficients of edges crossing $\gamma_x$ having a white vertex on the left, resp. right. Similarly, multiply by $z$, resp. $z^{-1}$ coefficients of edges crossing $\gamma_y$ having a white vertex on the left, resp. right. This gives a modified weight Kasteleyn matrix $\Ks_1(z,w)$ and the \emph{characteristic polynomial}~\cite{KOS}, denoted by $\Ps(z,w)$, is defined as
\begin{equation}\label{equ:charact_bipartite}
\Ps(z,w)=\det(\Ks_1(z,w)). 
\end{equation}

Consider two gauge equivalent, periodic dimer models arising from Kasteleyn matrices $\Ks_1$ and $\tilde{\Ks}_1$ then, by~\cite{KOS}, their respective characteristic polynomials $\Ps(z,w)$, $\tilde{\Ps}(z,w)$, are equal up to a constant and up to a global scaling $(z,w)\leftrightarrow(\lambda z,\mu w)\in(\CC^*)^2$, namely there exists $(\lambda,\mu)\in(\CC^*)^2$, such that
\begin{equation}\label{equ:gauge_equivalence_polynomials}
\Ps(z,w)=c\, \tilde{\Ps}(\lambda z,\mu w).
\end{equation}

\paragraph{Spectral curves, amoebas.} Consider a Laurent polynomial $P(z,w)$ in the variables $(z,w)\in(\CC^*)^2$. Then, the \emph{spectral curve} $\C$ associated to $P$ is defined to be the zero locus of the polynomial $P$, namely 
\[
\C=\{(z,w)\in (\CC^*)^2:\ P(z,w)=0\}. 
\]
When the spectral curve arises from the characteristic polynomial of a dimer model, one speaks of the~\emph{spectral curve of this dimer model}. Recall that the definition of the polynomial $P$, resp. $\Ps$, used to define the spectral curve $\C$ differs when the underlying graph is not bipartite, resp. is bipartite.

The \emph{real locus} of the curve $\C$ consists of the real points of $\C$, that is of the set of points 
\[
\{(x,y)\in(\RR^*)^2|\ P(x,y)=0\}. 
\]
The \emph{amoeba} of the curve $\C$, a notion introduced in~\cite{Gelfand}, denoted by $\A$, is the image of the curve $\C$ trough the map $(z,w)\mapsto (\log|z|,\log|w|)$. 

\begin{rem}\label{rem:spectral_curve}\leavevmode
\begin{enumerate}
\item As a consequence of Equation~\eqref{equ:central_symmetry}, the spectral curve $\C$ of a non-bipartite dimer model is \emph{centrally symmetric}, meaning that $(z,w)\in\C$ if and only if $(\frac{1}{z},\frac{1}{w})\in\C$. 
\item As a consequence of Equation~\eqref{equ:gauge_equivalence_polynomials}, gauge equivalent bipartite dimer models have the same spectral curve up to a global scaling $(z,w)\leftrightarrow(\lambda z,\mu w)$; their amoebas are translates of each other by the vector $(\log|\lambda|,\log|\mu|)$. 
\end{enumerate}
\end{rem}

\subsection{The (frustrated) Ising model}\label{sec:frustrated_Ising_model}

Let us recall the definition of the \emph{(frustrated) Ising model} given in the introduction. 
Consider a planar, simple graph $G=(V,E)$ such that edges are assigned \emph{real coupling constants} $\eps\Js=(\eps_e\Js_e)_{e\in E}$ with $\eps_e \in \{\pm 1\}$ and $\Js_e \in \RR^+$. 
A \emph{spin configuration}, denoted by $\sigma$, is an element of $\{-1,1\}^V$. When the graph is moreover finite, the \emph{Ising Boltzmann measure with free boundary conditions}, denoted by $\PPising$, is the probability measure on the set of spin configurations defined by
\begin{equation*}
\forall\,\sigma\in\{-1,1\}^V,\quad \PPising(\sigma)=\frac{1}{\Zising(G,\eps\Js)}
e^{\sum_{vv'\in E}\eps_e\Js_e\sigma_v\sigma_{v'}},
\end{equation*}
where $\Zising(G,\eps\Js)$ is the normalizing constant known as the \emph{partition function}. From now on, we will simply refer to this model as the \emph{Ising model}, omitting the (frustrated) term. As noted by Toulouse~\cite{Toulouse}, the real effect of negative coupling constants on the model is measured by the \emph{frustration function} at faces, a terminology due to Aspen. More precisely, we have the following.

\begin{defi}\label{defi:frustrated_ising}\leavevmode
\begin{enumerate}
\item[$\bullet$] Consider an Ising model on the graph $G$ with real coupling constants $\eps\Js$. Then, for every face $f$ of $G$, the \emph{frustration function} at the face $f$ is $\delta_f := \prod_{e \in \partial f}\eps_e$. 
\item[$\bullet$] Two Ising models on the graph $G$ with respective coupling constants $\eps\Js$ and $\tilde{\eps}\tilde{\Js}$ are said to be \emph{gauge equivalent} if, for every edge $e$ of $G$, $\Js_e = \tilde{\Js}_e$ and if the frustration functions are equal at all faces. A gauge equivalence class of models is thus characterized by the set of positive coupling constants $\Js$ together with a set of signs $(\delta_f)_{f\in F}\in\{-1,1\}^{F}$, where $F$ denotes the set of faces of $G$. Note that this notion is reminiscent of the notion of gauge equivalent bipartite dimer models.
\item[$\bullet$] An (equivalence class of) Ising model is \emph{frustrated} at a face $f$ if $\delta_f = +1$; otherwise it is said to be \emph{non-frustrated} at the face $f$. We will say that the model is \emph{frustrated}, resp. \emph{non-frustrated}, if it so is at every face.
\end{enumerate}
\end{defi}
\begin{rem}\leavevmode
\begin{enumerate} 
\item Note that the notion of \emph{frustration} differs from the notion of \emph{anti-ferromagnetism}, referring to negativity of the coupling constants. For example, a model with all negative coupling constant on a bipartite graph is gauge equivalent to a non-frustrated model. 
\item In the whole of this paper, the Ising model has coupling constants $\eps\Js$ that are allowed to be negative. It then belongs to an equivalence class of models, which might be locally/globally frustrated or not.  
\end{enumerate}
\end{rem}

The definition of gauge equivalence is motivated by the following. We provide a proof of this standard argument for sake of completeness, see for example~\cite[Section 3.2.]{Cugliandolo}.

\begin{lem}\label{lem:gauge_equiv_ising}
Consider a finite, planar, simple graph $G$ and two Ising models on $G$ with coupling constants $\eps\Js$ and $\tilde{\eps}\tilde{\Js}$. If the two models are gauge equivalent, then there exists a set of fixed signs $(\xi_v)_{v \in V}$ such that the bijection
$$
  (\sigma_v)_{v \in V} \in \{-1,1\}^V \mapsto (\tilde{\sigma}_v)_{v \in V} = (\xi_v \sigma_v)_{v \in V} \in \{-1,1\}^V,
$$
is an energy preserving and probability-preserving bijection between the two Ising models. 
\end{lem}

\begin{proof}
It is always possible to find $\xi_v$ such that $\tilde{\eps}_{vv'} = \xi_v \xi_{v'}\eps_{vv'}$. Indeed, choose a vertex $v_0$ and fix $\xi_{v_0} = 1$. Then, for every path 
$\gamma = v_0,v_1, \dots, v_n$, define 
$$
  \xi_v = \prod_{i=1}^n \eps_e\tilde{\eps}_e.
$$
This is well-defined because, by definition of gauge equivalence, this product is one when $\gamma$ is a loop. Now observe that the bijection of the lemma is energy preserving:
$$
  \sum_{vv'\in E}\tilde{\eps_e}\tilde{\Js}_e \tilde{\sigma}_v \tilde{\sigma}_{v'} = \sum_{vv'\in E} \xi_v \xi_{v'} \eps_e\Js_e \tilde{\sigma}_v \tilde{\sigma}_{v'} = \sum_{vv'} \eps_e\Js_e \sigma_v \sigma_{v'}. 
$$
As a consequence, it is also probability preserving.
\end{proof} 

The following proposition proves that for each gauge equivalence class (with an additional mild assumption in the toroidal case) there indeed exists an Ising model representing this class. 

\begin{lem}\label{lem:reconstructing_Ising_model}
Suppose that the graph $G$ is simple, finite, and that it is either embedded in the plane or embedded in the torus. 
If the graph $G$ is planar then, for every gauge equivalence class, i.e. for every set of positive coupling constants $\Js$ and for every frustration function $(\delta_f)_{f \in F} \in \{-1,1\}^F$, there exists edge-signs $(\eps_e)_{e \in E}$ such that the Ising model with coupling constants $\eps\Js$ lies in this equivalence class. If the graph $G$ is toroidal, suppose further that $\prod_{f\in F} \delta_f =1$, then the same conclusion holds.
\end{lem}

\begin{proof}
The proof is an explicit algorithm constructing $\eps=(\eps_e)_{e \in E}\in\{-1,1\}^E$ such that the Ising model with coupling constants $\eps\Js$ lies in the equivalence class. It is inspired by the algorithm of~\cite[Section 3.3]{CimRes07} used to find Kasteleyn phases. Let $T$ be a spanning tree of the dual graph $G^*$, rooted at the outer face when $G$ is planar, or rooted at any face $f_0$ when $G$ is embedded in the torus. For every edge not intersecting an edge of $T$, choose an arbitrary sign, say $\eps_e = 1$. At each step, remove a degree one vertex $f_i$ of the tree, that is a leaf, and the corresponding edge $e_i^*$ which is the dual of a unique edge $e_i \in E$. Choose the sign $\eps_{e_i}$ such that $\prod_{e \in \partial f_i}\eps_e = \delta_f$. Note that all edges around $f_i$ except $e_i$ have been given a sign in an earlier step of the algorithm, so there is a unique possible choice for $\eps_{e_i}$. In the planar case, when all vertices of the spanning tree have been removed, we obtain $\eps \in \{-1,1\}^E$ having the desired property. In the torus case, we do not have any freedom left for edges around the face $f_0$. Furthermore, we have $\prod_{f\in F}\delta_f=\prod_{f\in F}\prod_{e\in\partial f}\eps_e=\prod_{e\in E}(\eps_e)^2=1$, where in the third equality we used that each edge occurs twice exactly in the product. This implies the necessary condition $\prod_{f\in F}\delta_f=1$. This condition is also sufficient because, when performing the algorithm, all edges around the face $f_0$ have been assigned a sign, and the sign of the face $f_0$ is thus given by 
\[
\delta_{f_0}=\prod_{e\in\partial f_0}\eps_e=\frac{\prod_{f\in F}\prod_{e\in f}\eps_e}{\prod_{f\in F\setminus\{f_0\}}\prod_{e\in f}\eps_e}=\frac{1}{\prod_{f\in F\setminus\{f_0\}}\delta_f},
\]
which is true by assumption. 
\end{proof}
\begin{rem}
Note that the condition $\prod_{f\in F} \delta_f=1$ occurring in the toroidal case is not really restrictive. Indeed, if it is not satisfied, one can double the graph either in the horizontal or the vertical direction, and obtain a new toroidal graph where this condition now holds.
We will see an example of this situation in the case of the square lattice in Section~\ref{subsec:square_lattice}.
\end{rem}

\subsection{Dimer model representations of the Ising model}\label{sec:Ising_dimer model_representation}

The Ising model on $G$ with coupling constants $\eps\Js$ is known to be related to the dimer model on two decorated graphs obtained from $G$: the Fisher graph $\GF$~\cite{Fisher}, and the bipartite graph $\GQ$~\cite{Dubedat,WilsonXOR,BoutillierdeTiliere:XORloops}, see Figure~\ref{fig:Fisher_Dub}.
These relations are more often used when the coupling constants are positive, \emph{i.e.} $\eps=1$, but actually hold for \emph{real} coupling constants, \emph{i.e.} $\eps \in \{\pm 1\}^E$. 

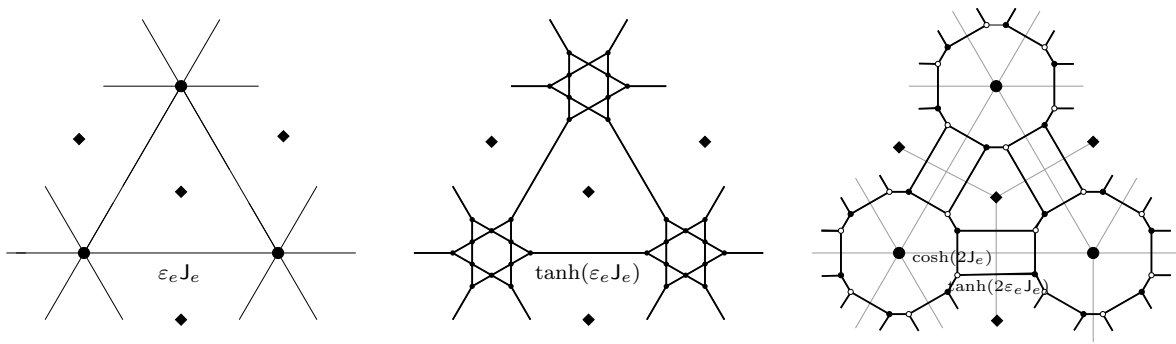
\begin{figure}[h]
  \centering
  \begin{overpic}[width=\linewidth]{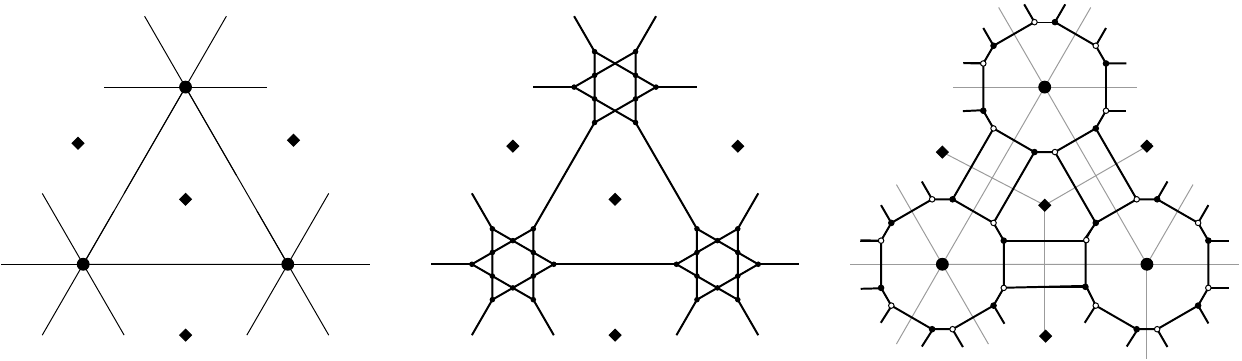} 
   \put(13,5.5){\scriptsize $\eps_e\Js_e$} 
   \put(45,5.5){\scriptsize $\tanh(\eps_e\Js_e)$}
   \put(80,4.5){\tiny $\tanh(2\eps_e\Js_e)$}
   \put(77,7){\tiny $\cosh(2\Js_e)$}
  \end{overpic}
\caption{Left: $G$ is the triangular lattice, and the dual $G^*$ is the hexagonal lattice. Center: corresponding Fisher graph $\GF$. Right: corresponding bipartite graph $\GQ$.}  \label{fig:Fisher_Dub}  
\end{figure}

\subsubsection{Fisher's correspondence} \label{sec:Fisher_correspondence}

Fisher's correspondence~\cite{Fisher} establishes an explicit mapping between the high temperature expansion of the Ising model on $G$~\cite{KramersWannier1,KramersWannier2} and the dimer model on the decorated \emph{Fisher graph} $\GF=(\VF,\EF)$. We actually use a version of the correspondence due to Dubédat~\cite[Section 4.1]{Dubedat}. For the purpose of this paper, it suffices to consider graphs with no boundary. The decorated graph $\GF$ is constructed from $G$ as follows: replace every vertex by a decoration made of triangles, where triangles are joined in a circular way, see Figure~\ref{fig:Fisher_Dub} (center). Then, polygon configurations arising from the high temperature expansion of the Ising model are in correspondence with $2^{|V|}$ dimer configurations: edges of polygon configurations correspond to \emph{long} edges of the dimer model; once the long edges are fixed, there are exactly two ways of filling each decoration with a dimer configuration~\cite{Fisher}. The dimer weight function $\mu$ corresponding to the Ising model with coupling constants $\eps\Js$ is given by, for every edge $\es$ of $\GF$,
\begin{equation}\label{equ:weight_Fisher}
\mu_{\es}=
\begin{cases}
1 & \text{ if $\es$ is a short edge, \emph{i.e.} belongs to a decoration}\\
\tanh(\eps_e\Js_e)& \text{ if the long edge $\es$ arises from an edge $e$ of $G$}.
\end{cases}
\end{equation}
Note that the weight function $\mu_\es$ is negative when the coupling constant is negative, \emph{i.e.} $\eps_e = -1$, so that the weight function $\mu$ is \emph{real-valued}. 
If one instead considers the low temperature expansion of the Ising model on the dual graph $G^*$, weights of the long edges would be equal to $e^{-2\eps_e\Js_e}$ and would all be positive. As a consequence of the correspondence, we have
\[
\Zising(G,\eps\Js)=\Bigl(\prod_{e\in E}\cosh(\eps_e\Js_e)\Bigr)\Zdimer(\GF,\mu). 
\]

\subsubsection{Bipartite dimer representation and duality}\label{sec:bipartite_representation_I}

The XOR Ising model~\cite{WilsonXOR}, constructed from two independent Ising models, is related to the dimer model on the decorated bipartite graph $\GQ=(\VQ,\EQ)$. There are actually two mappings leading to the dimer model on $\GQ$, one due to~\cite{Dubedat} based on bosonization identities and the results of~\cite{KadanoffWegner,Wu71,FanWu,WuLin}, and the other due to~\cite{BoutillierdeTiliere:XORloops} based on the results of~\cite{Nienhuis,WuLin}. Both are rather long to describe and, since we do not need the details of the construction, we refer to the original papers and only state the required results. The bipartite graph $\GQ$ is constructed from $G$ as follows, see Figure~\ref{fig:Fisher_Dub} (right): every edge $e$ of $G$ is replaced by a \emph{square}, noting the abuse of terminology since it might not be an actual geometric square; consecutive squares are then joined in a circular way using additional edges, see Figure~\ref{fig:Fisher_Dub} (right). The additional edges used to connect squares are referred to as \emph{external edges}. Note that in each square, two edges are parallel to an edge $e$ of $G$, and two are parallel to the dual edge $e^*$ of $G^*$. There are three kinds of faces in $\GQ$, those corresponding to a square, generically denoted by $y=\{e,e^*\}$ or simply $y$; those corresponding to a vertex $v$ of $G$; and those corresponding to a face $f$ of $G$/dual vertex of $G^*$. We use the same notation for the face and for the corresponding dual vertex. The bipartite coloring of vertices of $\GQ$ is chosen as in Figure~\ref{fig:Fisher_Dub} (right).  
The dimer weight function $\nu$ corresponding to the double Ising model with coupling constants $\eps\Js$ is given by, for every edge $\es$ of $\GQ$,
\begin{equation}\label{eq:Dimer_Ising_weights}
\nu_\es=
\begin{cases}
\tanh(2\eps_e\Js_e)&\text{ if $\es$ belongs to a square and is parallel to an edge $e$ of $G$}\\
\cosh^{-1}(2\eps_e\Js_e)&\text{ if $\es$ belongs to a square and is parallel to an edge $e^*$ of $G^*$}\\
1&\text{ if $\es$ is an external edge}. 
\end{cases}
\end{equation}
Note that the edge weight $\cosh^{-1}(2\eps_e\Js_e)=\cosh^{-1}(2\Js_e)$ is always positive, while the edge-weight $\tanh(2\eps_e\Js_e)$ is negative when the corresponding coupling constant is, i.e. $\eps_e = -1$, so that the weight function $\nu$ is \emph{real-valued}. 

\begin{defi}\label{defi:Ising_dimer_model}
The dimer model on $\GQ$ with edge-weights~\eqref{eq:Dimer_Ising_weights} arising from an Ising model on $G$ with real coupling constants $\eps\Js$ is referred to as the \emph{Ising-dimer model on $\GQ$}.
\end{defi}

Suppose that Kasteleyn-Kuperberg phases $\phi$ are assigned to edges of $\GQ$, and let $\KQ$ be the Kasteleyn matrix associated to $\nu$ and $\phi$. Then, face-weights of the Ising-dimer model on $\GQ$ are denoted by $\WI_{G,\eps\Js}(\,\cdot\,)$, and are given by, for every square face $y$, resp. vertex $v$ of $G$, resp. dual vertex $f$ of $G$:
{\small 
\begin{equation}\label{equ:gauge_Ising_dimers}
\begin{aligned}
\WI_{G,\eps\Js}(y) &= -\sinh^{2}(2\eps_e\Js_e)=-\sinh^{2}(2\Js_e)<0\\
\WI_{G,\eps\Js}(v) &=(-1)^{|v|-1}\prod_{e^* \in \partial v}\cosh^{-1}(2\Js_e)\in (-1)^{|v|-1}\RR^+\\
\WI_{G,\eps\Js}(f) &= (-1)^{|f|-1}\prod_{e \in \partial f}\tanh^{-1}(2\eps_e\Js_e)=(-1)^{|f|-1}\delta_f\prod_{e \in \partial f}\tanh^{-1}(2\Js_e)\in(-1)^{|f|-1}\delta_f\RR^{+}.
\end{aligned}
\end{equation}
}
where recall that $\delta_f = \prod_{e\in f} \eps_e$ is the frustration function.
\begin{rem}\label{rem:ising_model_gauge}
Observe from Equations~\eqref{equ:gauge_Ising_dimers} that face-weights of the Ising-dimer model on $\GQ$ only depend on the equivalence class of the Ising model, \emph{i.e.} on the absolute value coupling constants $\Js$, and on the frustration function $\delta=(\delta_f)_{f\in F}$. 
The Ising model is frustrated at $f$ if and only if $\delta_f = -1$ or equivalently if $\WI_{G,\eps\Js}(f) \in (-1)^{|f|}\RR^+$. Whenever there is no ambiguity on the underlying graph $G$, we simply denote the face-weights as $\WI_{\eps\Js}(\,\cdot\,)$
\end{rem}

Let us now turn to duality. Suppose that we start from an Ising model on the dual graph $G^*$ and real coupling constants $\Js^*$. Construct the decorated graph $\GQ$ starting from $G^*$ instead of $G$, while keeping the bipartite coloring fixed. Keeping the bipartite coloring fixed has the advantage of not needing an inverse when computing face-weights, which will be useful in the remainder of the paper. 

\begin{defi}\label{def:weak_duality}
The Ising model on $G$ with coupling constants $\eps\Js$, and the Ising model on $G^*$ with coupling constants $\eps^*\Js^*$ are said to be \emph{weakly dual} if face-weights at square faces of the corresponding dimer models on $\GQ$ are equal that is if, for every square face $y=\{e,e^*\}$,
\begin{align}\label{equ:weakly_dual_coupling_constants}
\WI_{G,\eps\Js}(y)=\WI_{G^*,\eps^*\Js^*}(y) \quad\Leftrightarrow\quad \sinh^{2}(2\Js_e)=\sinh^{-2}(2\Js^*_{e^*}).
\end{align}
Note that the second condition is true if and only if $\cosh^{2}(2\Js_e)=\tanh^{-2}(2\Js^*_{e^*})$, or equivalently
$\tanh^{-2}(2\Js_e)=\cosh^{2}(2\Js^*_{e^*})$. As a consequence, if the Ising models are weakly dual, then for every face corresponding to a primal or dual vertex, we have equality of the \emph{squared} face-weights:
\begin{equation}\label{equ:weak_duality_Ising}
(\WI)^2_{G,\eps\Js}(\,\cdot\,)=(\WI)^2_{G^*,\eps^*\Js^*}(\,\cdot\,).
\end{equation}
Observe that the notion of weak duality and related properties only rely on the absolute value coupling constants $\Js$.
This definition is to be paralleled with the definition of (strong) duality for the ferromagnetic Ising model asking that $\sinh(2\Js_e)\sinh(2\Js^*_{e^*})=1$~\cite{KramersWannier1,KramersWannier2}. 
\end{defi}

\begin{rem}\label{rem:duality}
 Note that if this condition is satisfied, then in the frustrated Ising model case, setting $\eps^*_{e^*}=\eps_e$, we cannot have both the condition $\WI_{G,\eps\Js}(v)\in(-1)^{|v|-1}\RR^+$ at every vertex $v$ of $G$, and $\WI_{G^*,\eps^*\Js^*}(f)\in(-1)^{|f|-1}\delta_f\RR^+$ at every dual vertex $f$ of $G^*$.
\end{rem}

\subsubsection{The periodic case}

Suppose that the graph $G$ is periodic, and consider an Ising model on the graph $G$, with periodic, real coupling constants $\eps\Js$. Consider the corresponding dimer model on the Fisher graph $\GF$, with periodic weight function $\mu$ given by~\eqref{equ:weight_Fisher}, and Kasteleyn signs $\eta$ assigned to edges; let $\Pising(z,w)$ denote the associated characteristic polynomial, see Equation~\eqref{equ:charact_non_bipartite}. Consider also the corresponding Ising-dimer model on the bipartite graph $\GQ$ with periodic weight function $\nu$ given by~\eqref{eq:Dimer_Ising_weights}, and Kasteleyn-Kuperberg phases $\phi$; let $\Pdub(z,w)$ denote the associated characteristic polynomial, see Equation~\eqref{equ:charact_bipartite}. Then, by Dubédat~\cite[Section 4.1]{Dubedat} one can choose the Kasteleyn-Kuperberg phases $\phi$ so that they are real, \emph{i.e.} take values in $\{-1,1\}$, and so that the characteristic polynomial $\Pising(z,w)$ and $\Pdub(z,w)$ are equal up to a global (real) constant. Namely, there exists a non-zero, real constant $\tilde{c}$, such that 
\begin{equation}\label{equ:Dubedat_equality}
\Pising(z,w)=\tilde{c} \Pdub(z,w).
\end{equation}

\begin{rem}\label{rem:Fisher_Dub}\leavevmode
\begin{enumerate}
\item As a consequence of Equation~\eqref{equ:Dubedat_equality}, the dimer model on the Fisher graph $\GF$ and the Ising-dimer model on the graph $\GQ$ have the same spectral curves, which we denote by~$\C$. Note that the curve $\C$ arises from a polynomial with \emph{real coefficients}.
\item As a consequence of Point 1. of Remark~\ref{rem:spectral_curve}, since the curve $\C$ arises from the characteristic polynomial of a non-bipartite dimer model (on the Fisher graph), it is \emph{centrally symmetric}. 
\end{enumerate}
\end{rem}

\section{Real dimer model on minimal bipartite graphs in genus~1}\label{sec:real_dimers}

One of the main goals of this paper is to characterize spectral curves arising from Ising-dimer models on $\GQ$, in the case where they are irreducible, and of genus 1.
From Point~1.~of Remark~\ref{rem:Fisher_Dub}, we know that such curves arise from polynomials with \emph{real coefficients}. This section places itself in the more general context of irreducible, genus 1 curves arising from polynomials with real coefficients, but not necessarily from an Ising-dimer model; it is interesting in its own respect. 

More precisely, consider first any irreducible, genus 1 algebraic curve $\C$.
Then by Fock~\cite{Fock}, up to scale change, this curve is the spectral curve of a dimer model with Fock's Kasteleyn matrix in genus~1 on an infinite, periodic, minimal graph $\Gs$; referred to as \emph{Fock's dimer model}, see Definition~\ref{def:Focks_dimer_model}. The goal of this section is to characterize Fock's dimer models arising from genus 1 curves $\C$ defined from polynomials with \emph{real coefficients}, when they satisfy Assumptions $(\dagger)$:
\begin{center}
$(\dagger)\ $ $\C$ is irreducible, and non singular except possibly at isolated real nodes\footnote{Since we consider the genus 1 version of Fock's dimer model, isolated real nodes are present when the graph $\Gs$ has a large period so that, even if we consider generic curves, such type of singularities cannot be avoided.}.  
\end{center}

This is an extension of the work~\cite[Theorem 34]{BCdT:genus_1}, where we identified parameters of Fock's dimer models arising from dimer models with \emph{positive edge-weights}. 

Fock's dimer model has three sets of parameters: an underlying torus $\TT(\tau)$ with modular parameter $\tau$, an angle map $\mapalpha$, and an additional parameter $t$. When the curve $\C$ arises from a polynomial with real coefficients, it is equipped with a natural real structure arising from complex conjugation: $(z,w)\in\C \leftrightarrow (\bar{z},\bar{w})\in\C$; such a curve is referred to as a \emph{real algebraic curve}. Our first result, Theorem~\ref{thm:real_spectral_curve}, identifies the parameters $\tau$ and $\mapalpha$ of Fock's dimer model corresponding to \emph{real algebraic curves}. When the curve $\C$ arises from a polynomial with real coefficients, Fock's associated dimer model has the further property of being a \emph{real dimer model} as in Definition~\ref{defi:dimer_extras}. Our second result, Theorem~\ref{thm:Fock_real_dimer_model}, identifies the parameters $t$ corresponding to \emph{real dimer models}.  

This section is organized as follows: in Section~\ref{sec:Fock_genus_1}, we give the definition of Fock's dimer model with all the required background; then in Section~\ref{sec:real_spectral_curve} we state and prove Theorem~\ref{thm:real_spectral_curve}, and in Section~\ref{sec:Fock_real_dimer_model} we state and prove Theorem~\ref{thm:Fock_real_dimer_model}.

\subsection{Fock's dimer model on minimal graphs in genus 1}\label{sec:Fock_genus_1}

In Section~\ref{subsec:background}, we recall the definition of all required tools: train-tracks and the angle map, the discrete Abel map, genus 1 Riemann surfaces and theta functions, modular transformations. Then, in Section~\ref{sec:def_Fock_general} we define Fock's dimer model~\cite{Fock} involving the modular parameter $\tau$, the angle map $\mapalpha$, and the parameter $t$. We also establish two useful properties: in Lemma~\ref{lem:invariance_weights}, we prove that face-weights are meromorphic functions of the parameters, and in Lemma~\ref{lem:modular_transo_dimers}, we identify the evolution of the model under modular transformations of the parameter $\tau$. 

\subsubsection{Background}\label{subsec:background}

This section on backgrounds aims at being self contained. The content with more details can be found in~\cite{KenyonSchlenker,Thurston, GoncharovKenyon,BCdT:immersions}, see also ~\cite[Section 2, Section 5]{BCdT:genus_1}.

\paragraph{Train-tracks, isoradial graphs, minimal graphs and angle map.} 

Let $G=(V,E)$ be an infinite, planar, locally finite, simple graph embedded in the plane so that faces are bounded topological discs, and let $G^*=(V^*,E^*)$ denote its dual embedded graph. 
The \emph{quad-graph} $\GGR$ has vertex set $V\sqcup V^*$, and edges are obtained by joining every primal vertex $v$ to the dual vertices on the boundary of the corresponding face.
Note that faces of $\GGR$ are quadrangles, and that the quad-graph is invariant by duality. 
A \emph{train-track} is a path in the dual graph $(\GGR)^*$ crossing opposite edges of quadrangles. 
Let $\T$ denote the set of train-tracks of the graph $G$. Following~\cite{KenyonSchlenker}, the graph $G$ is said to be \emph{isoradial} if train-tracks do not self intersect, and no pair of distinct train-tracks intersects more than once. 

When the graph $G=(V,E)$ is moreover bipartite, we instead use the notation $\Gs=(\Vs,\Es)$. In this case, train-tracks can be consistently oriented so that white vertices of $\Gs$ are on the left of the path, implying that black vertices are on the right. Denote by $\vec{\T}$ the set of oriented train-tracks. Following~\cite{Thurston,GoncharovKenyon}, the graph $\Gs$ is said to be \emph{minimal} if its oriented train-tracks do not self intersect, and no pair of train-tracks intersects twice in the same direction. As a consequence, minimal graphs are simple. More details on minimal graphs can be found in the paper~\cite{BCdT:immersions}. To each oriented train-track $\vec{T}$ is assigned an \emph{angle} $\alpha_{\vec{T}}\in\CC$ defining the \emph{angle map} $\mapalpha:\vec{\T}\rightarrow \CC$. 

Suppose further that the graph $\Gs$ is periodic. Fix a face $\fs$ of $\Gs$, and recall that 
$\gamma_x$, resp. $\gamma_y$, are two simple paths in the plane joining $\fs$ to $\fs+(1,0)$, resp. $\fs$ to $\fs+(0,1)$, intersecting only at $\fs$. They project onto the torus onto two simple closed loops in $\Gs_1^*$ also denoted by $\gamma_x,\gamma_y$.
Their homology classes $[\gamma_x],[\gamma_y]$ form a basis of the first homology group of the torus $H_1(\TT,\ZZ)$. Let us denote by $\vec{\T}_1$ the set of train-tracks of $\Gs_1$. 
Observing that each train-track $\vec{T}$ is a closed curve on the torus, its homology class can be written as $[\vec{T}]=h_T[\gamma_x]+v_T[\gamma_y]:=(h_T,v_T)$, with $(h_T,v_T)$ coprime integers. 
Homology of train-tracks satisfy the condition $\sum_{\vec{T}\in\vec{\T}_1}(h_T,v_T)=(0,0)$, see~\cite[Section 5]{BCdT:genus_1}.

\paragraph{Discrete Abel map.} Following Fock~\cite{Fock}, the discrete Abel-map $\mapd$ is a complex-valued function defined on vertices of $\GR$. Fix a reference vertex $v_0$ of $\GR$ and set $\mapd(v_0)=0$. Then, the value of $\mapd$ increases, resp. decreases, by $\alpha_{\vec{T}}$ when the oriented train-track $\vec{T}$ is crossed from right-to-left, resp. from left-to-right, see Figure~\ref{fig:minimal} (right) for an example of computation. This map is well defined. 

\paragraph{Genus 1 Riemann surface, and theta functions.} Consider the lattice $\Lambda=\ZZ+\tau\ZZ$, for some \emph{modular parameter} $\tau\in\CC$, such that $\Im(\tau)>0$. The associated genus 1 Riemann surface is the torus $\TT(\tau)=\CC/\Lambda$. Define $\rho$ to be the \emph{half-period translation vector}:
\begin{equation}\label{equ:def_rho}
\rho=\rho(j,\ell)=\frac{j}{2}+\frac{\ell}{2}\tau,\quad (j,\ell)\in\{0,1\}^2. 
\end{equation}
Then, $2\rho$ is a translation by what is known as a \emph{horizontal/vertical/diagonal} period of the lattice $\Lambda$. Note that we do not assume $\tau$ to be pure imaginary as was the case in~\cite{BCdT:genus_1}.  

The definition of the four theta functions $(\theta_{j,\ell})_{(j,\ell) \in \{0,1\}}$ can be found in~\cite[21.2.5]{DLMF}. 
Let $q=e^{i\pi\tau}$, then comparing to the alternative notation $(\theta_m)_{m \in \{1,2,3,4\}}$, we have~\cite[21.2(iii)]{DLMF},
\[
  \theta_{1,1}(z,\tau) = -\theta_1(\pi z,q); \ \ \theta_{1,0}(z,\tau)=\theta_2(\pi z,q); \ \ \theta_{0,0}(z,\tau) = \theta_3(\pi z,q); \ \ \theta_{0,1}(z,\tau) = \theta_4(\pi z,q).
\]
Whenever no confusion occurs, we remove the reference to $\tau$ in the notation. 
\begin{rem}
  Instead of the more common notation $(\theta_m)_{m \in \{1,2,3,4\}}$, we chose to use the notation $(\theta_{j,\ell})_{(j,\ell) \in \{0,1\}}$ which extends more easily to the higher genus case, see~\cite{BCdT:genus_g}, and is useful to obtain compact forms for formulas, see below. 
\end{rem}
Here are useful formulas describing parity, complex conjugation, and transformation of theta functions under (half)-periods of the torus $\TT(\tau)$. They are the genus 1 versions of formulas of the quoted references.

\begin{lem}[] \leavevmode
For $m,n\in\{0,1\}$ with indices understood mod 2, and $z\in\CC$, 
\begin{enumerate}
\item Parity\emph{~\cite[p.17]{MumfordTheta1}} 
\begin{align}\label{eq:theta:1}
\theta_{m,n}(-z) &= (-1)^{mn}\theta_{m,n}(z).
\end{align}  
\item Complex conjugation\emph{~\cite[proof of Lemma 18]{BCdT:genus_g}}
\begin{equation}\label{equ:complex_conjugation}
\overline{\theta_{m,n}(z,\tau)}=\theta_{m,n}(\bar{z},-\bar{\tau}).
\end{equation}
\item Translation by periods of $\TT(\tau)$\emph{~\cite[p.123]{MumfordTheta1}}
\begin{align}
  \theta_{m,n}(z+2\rho) &= (-1)^{jm+\ell n}(q^{-1} e^{-2i\pi z})^{\ell}\theta_{m,n}(z)  \label{eq:theta:2} \\
  \theta_{m,n}(z-2\rho) &=(-1)^{jm+\ell n} (q^{-1} e^{2i\pi z})^{\ell}\theta_{m,n}(z).   \label{eq:theta:3}
  \end{align}
\item Translation by half-periods of $\TT(\tau)$\emph{~\cite[Table 0, p.19]{MumfordTheta1}}
\begin{align}
  \theta_{m, n}(z+\rho) &= (q^{-1/4}e^{-i\pi z})^\ell i^{\ell(n+j)} (-1)^{\ell(n+j) + j\ell n +jmn} \theta_{m+\ell,n+j}(z)\label{eq:theta:4}\\
  \theta_{m,n}(z-\rho) &= (q^{-1/4}e^{i\pi z})^\ell i^{\ell(n+j)}(-1)^{jm+j\ell n+jmn}\theta_{m+\ell,n+j}(z).\label{eq:theta:5}
\end{align} 
\end{enumerate}
\end{lem}

We need the following addition formula. A proof is provided in Lemma~\ref{lem:appendix_1} of Appendix~\ref{sec:App_A} since we could not find it as such in the literature.

\begin{lem}
For $m,n\in\{0,1\}$ with indices understood mod 2, and $z,w\in\CC$,
\begin{equation} \label{eq:addition:formula}
\begin{split}
\theta_{1,1}^2(z)\theta_{0,0}^2(w) + &(-1)^{1+m+mn}\theta_{m,n}^2(z)\theta_{m+1,n+1}^2(w)=\\
&=(-1)^{n+mn} \theta_{m+1,n+1}(z+w)\theta_{m+1,n+1}(z-w)\theta_{m,n}^2(0).
\end{split}
\end{equation}
\end{lem}

We also need the following result. A proof is provided in Lemma~\ref{lem:meromorphic_theta_ratio_appendix} of Appendix~\ref{sec:App_A} since we use a slight extension which is not written as such in~\cite{MumfordTheta1}.

\begin{lem}{\emph{\cite[\textsection 6, p.24]{MumfordTheta1}}}\label{lem:meromorphic_theta_ratio}
Let $a_1,\dots,a_r, b_1,\dots,b_r\in\CC$. Then, the function
\[
f(u)=e^{-2i\pi\ell u}\prod_{i=1}^r \frac{\theta_{m,n}(u-a_i)}{\theta_{m,n}(u-b_i)},
\]
is a meromorphic function on $\TT(\tau)$ if and only if $\ell\in\ZZ$, and $\sum_{i=1}^r (a_i-b_i)-\ell\tau \in\ZZ$.
\end{lem}

\paragraph{Modular transformations of the lattice parameters.} We recall the required notions, and refer to~\cite[Ch.9]{Lawden} for more details. \emph{Modular transformations} consist of the following transformations of the modular parameter $\tau$:
\[
\tau'=\frac{c+d\tau}{a+b\tau},
\]
where $a,b,c,d\in\ZZ$ satisfy $ad-bc=1$. They can be seen as the action of the group $SL_2(\ZZ)$ on the upper half plane by Möbius transformations. They thus form a group, known as the \emph{modular group}, and we have $\TT(\tau)=\TT(\tau')$. The \emph{standard fundamental domain} is:
\begin{equation}\label{equ:fund_domain_modular}
\{\tau\in\CC, \Im(\tau)>0\,:\,\Re(\tau)\leq \frac{1}{2},|\tau|\geq 1\}.
\end{equation}
Generators of the modular group are given by:
\begin{equation*}
  \tau_1=\tau_1(\tau) = -\frac{1}{\tau} \quad ; \quad \tau_2=\tau_2(\tau) = \tau + 1.
\end{equation*}

The following lemma records the evolution of theta functions under modular transformations.

\begin{lem}\label{lem:modular_transfo}\leavevmode
\begin{enumerate}
\item Under the transformation $\tau_1=-\frac{1}{\tau}$~\emph{\cite[(20.7.30)-(20.7.32)]{DLMF}},
\begin{equation}\label{eq:tau_1}
  \begin{aligned}
(-i\tau)^{\frac{1}{2}}\theta_{1,1}(z,\tau) = -ie^{i\pi\tau_1 z^2}\theta_{1,1}(\tau_1\cdot z,\tau_1)~&,~
(-i\tau)^{\frac{1}{2}}\theta_{0,0}(z,\tau) = e^{i\pi\tau_1 z^2}\theta_{0,0}(\tau_1 \cdot z,\tau_1)\\
(-i\tau)^{\frac{1}{2}}\theta_{1,0}(z,\tau) = -ie^{i\pi\tau_1 z^2}\theta_{0,1}(\tau_1\cdot z,\tau_1)~&,~
(-i\tau)^{\frac{1}{2}}\theta_{0,1}(z,\tau) = e^{i\pi\tau_1 z^2}\theta_{1,0}(\tau_1 \cdot z,\tau_1).
  \end{aligned}
\end{equation}
\item Under the transformation $\tau_2=\tau+1$~\emph{\cite[(20.7.26)-(20.7.28)]{DLMF}},
\begin{equation}\label{eq:tau_2}
\theta_{1,1}(z,\tau) = e^{-i\frac{\pi}{4}}\theta_{1,1}(z,\tau_2),\quad 
\theta_{0,0}(z,\tau) =  \theta_{0,1}(z,\tau_2).
\end{equation} 
\end{enumerate}
\end{lem}

\subsubsection{Definition and first properties}\label{sec:def_Fock_general}

We are now ready to state the definition of Fock's dimer model~\cite{Fock}, and establish some properties.

\begin{defi}[\cite{Fock}]\label{def:Focks_dimer_model} Consider an infinite minimal graph $\Gs$, a modular parameter $\tau$, an angle map $\mapalpha:\vec{\T}\rightarrow\CC$, and an element $t$ of $\CC$. Then, \emph{Fock's dimer model on $\Gs$ in genus~1, with parameters $\tau,\mapalpha,t$}, or simply \emph{Fock's dimer model}, is defined through its Kasteleyn matrix $\Ks_{\mapalpha, t}(\,\cdot\,,\tau)=\Ks_{\mapalpha, t}(\,\cdot\,)$, referred to as \emph{Fock's Kasteleyn matrix}. The latter has rows indexed by white vertices of $\Gs$, columns by black ones, and non-zero coefficients given by, for  
every edge $\ws\bs$ of $\Gs$, 
\begin{equation}\label{equ:Focks_weight}
\Ks_{\mapalpha, t}((\ws,\bs),\tau)=\Ks_{\mapalpha, t}(\ws,\bs)=\frac{\theta_{1,1}(\beta-\alpha)}{\theta_{0,0}(t+\mapd(\fs))\theta_{0,0}(t+\mapd(\fs'))},
\end{equation}
where the edge $\ws\bs$ is crossed by train-tracks with angles $\alpha,\beta\in\CC$, see Figure~\ref{fig:minimal} (left); suppose that intersecting train-tracks have distinct angles for otherwise the coefficient is 0.
Depending on the context, we will emphasize the dependence in $\tau$ in the notation for $\Ks$, or not. 

Suppose that Kasteleyn-Kuperberg phases $\phi$ are assigned to edges of $\Gs$. Then, by Definition~\eqref{defi:dimer_extras}, Fock's Kasteleyn matrix $\Ks_{\mapalpha, t}$ defines a dimer weight function, which is \emph{complex-valued} in general. 
\end{defi}

\begin{figure}[h]
  \centering
  \begin{overpic}[width=\linewidth]{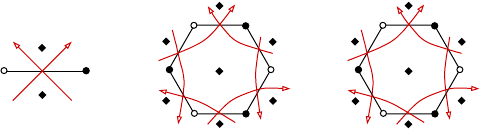}
  \put(-2,11){\scriptsize $\ws$} 
  \put(19.5,11){\scriptsize $\bs$}
  \put(2,18.5){\scriptsize $\beta$}
  \put(15,18.5){\scriptsize $\alpha$}
  \put(9,4){\scriptsize $\fs$}
  \put(9,18){\scriptsize $\fs'$}
  \put(48,26.2){\scriptsize $\alpha_1$}
  \put(41,26){\scriptsize $\beta_1$}
  \put(30.5,8){\scriptsize $\alpha_2$}
  \put(36,-1){\scriptsize $\beta_2$}
  \put(54,0){\scriptsize $\alpha_3$}
  \put(61,8){\scriptsize $\beta_3$}
  \put(39,23){\scriptsize $\ws_1$}
  \put(32,11){\scriptsize $\bs_1$}
  \put(39,0.5){\scriptsize $\ws_2 $}
  \put(51,0.5){\scriptsize $\bs_2 $}
  \put(58,11){\scriptsize $\ws_3$}
  \put(51,23){\scriptsize $\bs_3$}
  \put(45.5,9){\scriptsize $\fs$}
  \put(84,9){\scriptsize \textcolor{blue}{$\mapd(\fs)$}}
  \put(62.5,20){\scriptsize \textcolor{blue}{$\mapd(\fs)+\beta_2-\alpha_1$}}
  \put(80,27){\scriptsize \textcolor{blue}{$\mapd(\fs)+\beta_1-\alpha_1$}}
  \put(94,20){\scriptsize \textcolor{blue}{$\mapd(\fs)+\beta_1-\alpha_3$}}
  \put(94,3){\scriptsize \textcolor{blue}{$\mapd(\fs)+\beta_3-\alpha_3$}}
  \put(80,-2){\scriptsize \textcolor{blue}{$\mapd(\fs)+\beta_3-\alpha_2$}}
  \put(64,3){\scriptsize \textcolor{blue}{$\mapd(\fs)+\beta_2-\alpha_2$}}
  \end{overpic}
  \vspace{0.1cm}
  \caption{Left: train-track angles $\alpha,\beta$ at an edge $\ws\bs$. Center: train-track angles $\alpha_1,\beta_1,\dots,\alpha_{|\fs|/2},\beta_{|\fs|/2}$ around a face $\fs$ of degree $6$, with vertices $\ws_1,\bs_1,\dots,\ws_{|\fs|/2},\bs_{|\fs|/2}$ in counterclockwise order. Right (blue): computation of the discrete Abel map using $\mapd(\fs)$.}
  \label{fig:minimal}
\end{figure}
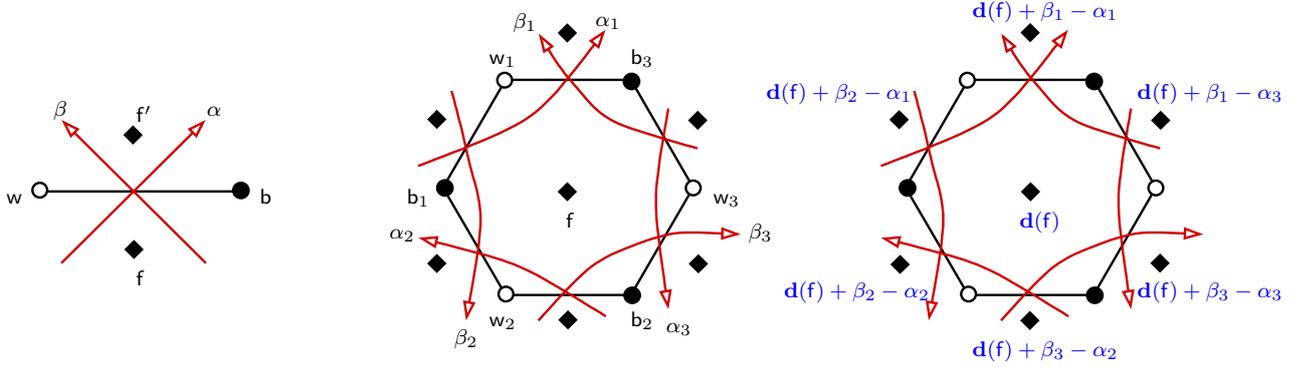

As noted in Section~\ref{subsec:dimer_model_defi}, rather than the weight function, the relevant quantities defining the dimer model are the face-weights. 
Given a face $\fs$ of degree $|\fs|$ (where $|\fs|$ is even  since $\Gs$ is bipartite) with vertices labeled $\ws_1,\bs_1,\dots,\ws_{|\fs|/2},\bs_{|\fs|/2}$ in counterclockwise order, and train-track angles denoted by $\alpha_1,\beta_1,\cdots,\alpha_{|\fs|/2},\beta_{|\fs|/2},$ as in Figure~\ref{fig:minimal} (center), the face-weight $\W_{\mapalpha,t}(\fs,\tau)=\W_{\mapalpha,t}(\fs)$ associated to Fock's dimer model is given by 
\begin{equation}\label{equ:face_weight_bip}
\W_{\mapalpha,t}(\fs,\tau)=\W_{\mapalpha,t}(\fs)=\prod_{i=1}^{|\fs|/2}
   \frac{\theta_{1,1}(\alpha_i-\beta_{i+1})}{\theta_{1,1}(\beta_i-\alpha_i)} 
 \frac{\theta_{0,0}(t+\mapd(\fs)+\beta_{i}-\alpha_i)}{\theta_{0,0}(t+\mapd(\fs)+\beta_{i+1}-\alpha_i)},
\end{equation} 
where we used the definition of the discrete Abel map to express all terms involving $\mapd$ using $\mapd(\fs)$, see Figure~\ref{fig:minimal} (right).

Theta functions have non-trivial periods along horizontal/vertical periods of the torus $\TT(\tau)$ so that, in order for Fock's Kasteleyn matrix to be well defined, the angle map $\mapalpha$ and the parameter $t$ have to be considered as living in $\CC$. The next lemma proves that actually face-weights are periodic functions on $\TT(\tau)$ of the parameters of the model. This implies that 
when considering face-weights, the angle map $\mapalpha$ and $t$ can be seen as living on $\TT(\tau)$, whereas when considering Fock's Kasteleyn matrix, lifted angles in $\CC$ should be considered.
We do not introduce different notation for the angles in $\TT(\tau)$ and the lifted angles in $\CC$, trying to make it clear from the context. Note that periodicity of face-weights under horizontal translations is already proved in~\cite[Remark 30]{BCdT:genus_g}.

\begin{lem}\label{lem:invariance_weights}
For every face $\fs$ of $\Gs$, the face-weight $\W_{\mapalpha,t}(\fs)$ is a meromorphic function on the torus $\TT(\tau)$ when considered as a function of $t$, or of any of the angles $\alpha_i,\beta_i$, $i\in\{1,\dots,|\fs|/2\}$. 
\end{lem}
\begin{proof}
Consider first the face-weight $\W_{\mapalpha,t}(\fs)$ as a function of $t$. Since the first term of $\W_{\mapalpha,t}(\fs)$ is constant when seen as a function of $t$, we need to prove that 
\[
\prod_{i=1}^{|\fs|/2}
 \frac{\theta_{0,0}(t+\mapd(\fs)+\beta_{i}-\alpha_i)}{\theta_{0,0}(t+\mapd(\fs)+\beta_{i+1}-\alpha_i)},
\]
is a meromorphic function of $t$ on $\TT(\tau)$. This is indeed the case using Lemma~\ref{lem:meromorphic_theta_ratio} and observing that 
\[
\sum_{i=1}^{|\fs|/2}\bigl(-\mapd(\fs)-\beta_{i}+\alpha_i\bigr)-\bigl(-\mapd(\fs)-\beta_{i+1}+\alpha_i\bigr)=\sum_{i=1}^{|\fs|/2}(\beta_{i+1}-\beta_{i})=0.
\]
Consider now the face-weight $\W_{\mapalpha,t}(\fs)$ as a function of $\alpha_i$. Since the graph is minimal, we know that train-tracks around a face are distinct~\cite[Lemma 8]{BCdT:immersions}, so that $\alpha_i$ corresponds to a single train-track around the face $\fs$. We thus need to show that
\[
   \frac{\theta_{1,1}(\alpha_i-\beta_{i+1})}{\theta_{1,1}(\beta_i-\alpha_i)} 
 \frac{\theta_{0,0}(t+\mapd(\fs)+\beta_{i}-\alpha_i)}{\theta_{0,0}(t+\mapd(\fs)+\beta_{i+1}-\alpha_i)},
\]
is a meromorphic function of $\alpha_i$ on $\TT(\tau)$. Using Identity~\eqref{eq:theta:4}, we have
that 
\begin{align*}
\frac{\theta_{0,0}(t+\mapd(\fs)+\beta_{i}-\alpha_i)}{\theta_{0,0}(t+\mapd(\fs)+\beta_{i+1}-\alpha_i)}
=e^{-i\pi(\beta_{i}-\beta_{i+1})}
\frac{\theta_{1,1}(t-\frac{1}{2}-\frac{\tau}{2}+\mapd(\fs)+\beta_{i}-\alpha_i)}{\theta_{1,1}(t-\frac{1}{2}-\frac{\tau}{2}+\mapd(\fs)+\beta_{i+1}-\alpha_i)}.
\end{align*}
Using furthermore that $\theta_{1,1}$ is odd, we need to prove that
\[
-e^{-i\pi(\beta_{i+1}-\beta_i)}\frac{\theta_{1,1}(\alpha_i-\beta_{i+1})}{\theta_{1,1}(\alpha_i-\beta_i)} 
\frac{\theta_{1,1}(\alpha_i-(t-\frac{1}{2}-\frac{\tau}{2}+\mapd(\fs)+\beta_{i}))}{\theta_{1,1}(\alpha_i-(t-\frac{1}{2}-\frac{\tau}{2}+\mapd(\fs)+\beta_{i+1}))},
\]
is a meromorphic function of $\alpha_i$ on $\TT(\tau)$. The exponential term is constant when seen as a function of $\alpha_i$. For the remaining product of two terms, suppose first that there is no contribution $\alpha_i$ to $\mapd(\fs)$, then this product is indeed meromorphic using Lemma~\ref{lem:meromorphic_theta_ratio} and observing that 
\[
\beta_{i+1}+\Bigl(t-\frac{1}{2}-\frac{\tau}{2}+\mapd(\fs)+\beta_{i}\Bigr)
-\beta_i-\Bigl(t-\frac{1}{2}-\frac{\tau}{2}+\mapd(\fs)+\beta_{i+1}\Bigr)=0.
\]
When $\mapd(\fs)$ contains the term $\alpha_i$, the argument is similar, but involves heavier notation so that we do not make it explicit. 
The proof is similar when considering the face-weight $\W_{\mapalpha,t}(\fs)$ as a function of $\beta_i$. \qedhere
\end{proof}

\begin{rem}\label{rem:angle_definition}
Note that the discrete Abel map $\mapd(\fs)$ is always summed with $t$, so that our proof also implies that the face-weight $\W_{\mapalpha,t}(\fs)$ is meromorphic when considered as function of $\mapd(\fs)$. 
\end{rem}

Lemma~\ref{lem:modular_transo_dimers} below proves that Fock's dimer model is invariant under modular transformations of the modular parameter $\tau$, up to an appropriate transformation of the angle map $\mapalpha$ and of the parameter $t$. In particular, this implies that we can restrict the study of Fock's dimer model on the torus $\TT(\tau)$ to the standard fundamental domain of the modular parameter $\tau$ as defined in Equation~\eqref{equ:fund_domain_modular}.

\begin{lem}\label{lem:modular_transo_dimers}
For every face $\fs$ of $\Gs$, the face-weight $\W_{\mapalpha,t}(\fs,\tau)$ transforms as follows under the two generators $\tau_1=-\frac{1}{\tau},\tau_2=\tau+1$ of the modular group:
\begin{equation*}
\W_{\mapalpha,t}(\fs,\tau)=
\W_{\tau_1 \mapalpha,\tau_1  t}(\fs,\tau_1),\quad 
\W_{\mapalpha,t}(\fs,\tau)=
\W_{\mapalpha,t+\frac{1}{2}}(\fs,\tau_2).
\end{equation*}
\end{lem}

\begin{proof}
The first equality is obtained by applying Identity~\eqref{eq:tau_1} to $\W_{\mapalpha,t}(\fs,\tau)$. Observe that the theta characteristics $(1,1)$ and $(0,0)$ are unchanged, and that arguments of the theta functions are multiplied by $\tau_1$. The terms $(-i\tau)^{\frac{1}{2}}$ and $(-i)$ cancel out, so that we are left with proving that the contribution of the terms $e^{i\pi\tau_1 z^{2}}$ cancel out. The latter is equal to:
\begin{align*}
\prod_{i=1}^{|\fs|/2}e^{i\pi[(\alpha_i-\beta_{i+1})^2+(t+\mapd(\fs)+\beta_i-\alpha_i)^2-
(\beta_i-\alpha_i)^2-(t+\mapd(\fs)+\beta_{i+1}-\alpha_i)^2]}=
e^{2i\pi\tau_1(t+\mapd(\fs))\sum_{i=1}^{|\fs|/2}(\beta_{i}-\beta_{i+1})}=1,
\end{align*}
thus proving the first equality.
The second equality is obtained by noting that, when performing the modular transformation $\tau_2$, see Identity~\eqref{eq:tau_2}, arguments of theta functions are unchanged, the characteristic $(1,1)$ is unchanged, and the characteristic $(0,0)$ gets mapped to $(0,1)$. One then uses Identity~\eqref{eq:theta:4} to bring $\theta_{0,1}$ back to $\theta_{0,0}$:
\[
\theta_{0,1}(z,\tau_2)=\theta_{0,0}\Bigl(z+\frac{1}{2},\tau_2\Bigr).
\]
The contributions of prefactors are all constant and cancel out, thus ending the proof of the second equality. \qedhere 
\end{proof}

\subsection{Parameterization of real algebraic curves using dimers}\label{sec:real_spectral_curve}

Let $\C$ be an algebraic curve of genus 1, and suppose that it satisfies Assumptions $(\dagger)$. By Fock~\cite{Fock}, there exists an infinite, periodic minimal graph $\Gs$, a torus $\TT(\tau)$, a periodic angle map $\mapalpha:\vec {\T_1}\rightarrow\CC$, such that for every $t\in\CC$, the spectral curve of Fock's dimer model on $\Gs$ is the curve $\C$ up to a global scaling $(z,w)\leftrightarrow(\lambda z,\mu w)$ for some $(\lambda,\mu)\in(\CC^*)^2$. We speak of \emph{Fock's dimer models corresponding to the curve $\C$.} Note that, since the curve is supposed to be non singular apart from isolated real singularities, the angle map $\mapalpha$ is generic in the sense that there are a priori no relations between angles of distinct train-tracks, as linear relations for example. 

Since the curve $\C$ is irreducible, we actually know that an explicit bi-rational parameterization of $\C$ is given by the meromorphic function $\Psi:\TT(\tau)\rightarrow \C$, where $\Psi(u)=(z(u),w(u))$, with $z,w$ given by 
\begin{equation}\label{equ:param_spectral_curve}
z(u)=c_1\cdot e^{2i\pi \ell_v\cdot u}\prod_{\vec{T}\in\vec{\T}_1}\theta_{1,1}(u-\alpha_{\vec{T}})^{-v_{\vec{T}}},\quad 
w(u)=c_2\cdot e^{-2i\pi \ell_h\cdot u}\prod_{\vec{T}\in\vec{\T}_1}\theta_{1,1}(u-\alpha_{\vec{T}})^{h_{\vec{T}}},
\end{equation}
for some $c_1,c_2\in\CC$, where $\sum_{\vec{T}\in\vec{\T}_1} v_{\vec{T}}=\sum_{\vec{T}\in\vec{\T}_1} h_{\vec{T}}=0$, and 
\begin{equation}\label{cond:angle_map_periodic}
\sum_{\vec{T}\in\vec{\T}_1}\alpha_{\vec{T}}v_{\vec{T}}=\ell_v\tau,\quad 
\sum_{\vec{T}\in\vec{\T}_1}\alpha_{\vec{T}}h_{\vec{T}}=\ell_h\tau,
\end{equation}
with $\ell_v,\ell_h\in\ZZ$; the form of the exponential prefactor in Equation~\eqref{equ:param_spectral_curve} is given by Lemma~\ref{lem:meromorphic_theta_ratio}. The existence of $\Psi$ follows from general parameterization theorems, while the breakthrough contribution of Fock~\cite{Fock} consists in relating parameters of $\Psi$ 
to parameters of the dimer model on a periodic minimal graph $\Gs$ with Fock's Kasteleyn matrix, see Definition~\ref{def:Focks_dimer_model}. In the case of genus 0 Harnack curves, this had been previously established in
\cite{Kenyon3,KO}, and a specific case of the genus 1 case was established in~\cite{BdTR1}. Note that the construction of the periodic minimal graph $\Gs$ arises from the paper~\cite{GoncharovKenyon}. 

The algebraic curve $\C$ is said to be \emph{real} if it is invariant by complex conjugation: $(z,w)\in\C$ if and only if $(\bar{z},\bar{w})\in\C$. 
The curve $\C$ is said to be \emph{maximal} if the number of connected components of the real locus (in the compactified toric surface associated to the Newton polygon, see for example~\cite{MikhalkinRullgard}) is equal to 2, that is equal to the genus plus 1.
The main result of this section, Theorem~\ref{thm:real_spectral_curve} below, identifies the values of the modular parameter $\tau$ and of the angle map $\mapalpha$ corresponding to \emph{real} curves. In the following section, we identify the values of the parameter $t$ associated to \emph{real dimer models}. These results generalize~\cite[Theorem 34]{BCdT:genus_1} where we characterized the values of the parameters $\tau,\mapalpha$ corresponding to simple Harnack curves and the values of $t$ associated to positive dimer edge-weights.

\begin{thm}\label{thm:real_spectral_curve}
Consider an algebraic curve $\C$ of genus 1 satisfying $(\dagger)$ and Fock's corresponding dimer model on $\Gs$
with modular parameter $\tau$, angle map $\mapalpha$ (both fixed by $\C$), and free parameter $t\in\CC$. Suppose that $\C$ contains a real point. Then, the algebraic curve $\C$ is real if and only if, for generic values, the modular parameter $\tau$ belongs to $i\RR^+\cup \bigl\{\frac{1}{2}+i\RR^+\,:\Im(\tau)\geq\frac{\sqrt{3}}{2}\bigr\}$, up to modular transformations of $\tau$. Moreover, 
\begin{enumerate}
\item when $\tau\in i\RR^+$ then, generically, the angle map $\mapalpha=(\alpha_{\vec{T}})_{\vec{T}\in\vec{\T}_1}$ takes values in $\RR + \{0,\frac{\tau}{2}\}~[\Lambda]$. Moreover, the curve $\C$ is maximal, 
\item when $\tau\in \frac{1}{2}+i\RR$ with $\Im(\tau)\geq \frac{\sqrt{3}}{2}$ then, generically, either 
the angle map $\mapalpha$ takes values in $\RR~[\Lambda]$, or the angle map $\mapalpha$ takes values in $\bigl\{0,\frac{1}{2}\bigr\}+ i\RR~[\Lambda]$. In both cases, the curve $\C$ is not maximal. 
\end{enumerate}
\end{thm}

\begin{proof}
Consider the birational parameterization $\Psi$ of the curve given by~\eqref{equ:param_spectral_curve}. Following~\cite[Chapter 11, Section 88]{DuVal}, the curve $\C$ is invariant by complex conjugation if and only if the real structure is preserved by $\Psi$, meaning that there exists an anti-holomorphic involution $\sigma(u)=a\bar{u}+b$ of the torus $\TT(\tau)$ such that, for every $u\in\TT(\tau)$, $\overline{\Psi(u)}=\Psi(\sigma(u))$. Du Val furthermore classifies these involutions, see also~\cite[Table 1]{Bates}, and obtains the following conditions on the modular parameter $\tau$ (assuming it is generic) up to modular transformations, and on the coefficients $a$ and $b$:
\begin{enumerate}
 \item[$\circ$] $\tau\in i\RR,\Im(\tau)\geq 1$, $a=1$ and $b\in\{0,\frac{1}{2}\}$, or $a=-1$ and $b\in\{0,\frac{\tau}{2}\}$.
 \item[$\circ$] $\tau\in \frac{1}{2}+i\RR,\Im(\tau)\geq \frac{\sqrt{3}}{2}$, $a\in\{-1,1\}$ and $b=0$.
\end{enumerate}

%

\underline{Proof of Point 1}. Suppose that $\tau\in i\RR,\Im(\tau)\geq 1$. Then $-\bar{\tau}=\tau$, and using Identity~\eqref{equ:complex_conjugation} we have
\[
\overline{\theta_{1,1}(z,\tau)}=\theta_{1,1}(\bar{z},-\bar{\tau})=\theta_{1,1}(\bar{z},\tau).
\]
Consider the case where $\sigma(u)=\bar{u}+b$, with $b\in\{0,\frac{1}{2}\}$. The condition 
$\overline{\Psi(u)}=\Psi(\sigma(u))$ written for the first component reads
\begin{align*}
&\forall\,u\in\TT(\tau),\quad \bar{c}_1\cdot 
e^{-2i\pi \ell_v\cdot \bar{u}}
\prod_{\vec{T}\in\vec{\T}_1}\theta_{1,1}(\bar{u}-\bar{\alpha}_{\vec{T}})^{-v_{\vec{T}}}=
c_1\cdot 
e^{2i\pi \ell_v\cdot (\bar{u}+b)}
\prod_{\vec{T}\in\vec{\T}_1}\theta_{1,1}(\bar{u}+b-\alpha_{\vec{T}})^{-v_{\vec{T}}}\\
\Leftrightarrow\ & 
\forall\,u\in\TT(\tau),\quad 
e^{-4i\pi\ell_v\cdot u}
\prod_{\vec{T}\in\vec{\T}_1}\left(\frac{\theta_{1,1}(u-\bar{\alpha}_{\vec{T}})}{\theta_{1,1}(u+b-\alpha_{\vec{T}})}\right)^{-v_{\vec{T}}}=e^{2i\pi\ell_v \cdot b}\frac{c_1}{\bar{c}_1}. 
\end{align*}
By Lemma~\ref{lem:meromorphic_theta_ratio}, the left hand side of the above equation is a meromorphic function of $\TT(\tau)$, which we denote by $m=m(u)$. The condition $\overline{\Psi(u)}=\Psi(\sigma(u))$ is equivalent to asking that this function $m$ is equal to an explicit constant, namely $e^{2i\pi\ell_v \cdot b}c_1/\bar{c}_1$. In particular, the function $m$ needs to be constant, which is true if and only if the zeros of the numerator and the denominator cancel out. Using that zeros of the function $\theta_{1,1}(\,\cdot\,)$ are the points $0\,[\Lambda]$, and using that by our genericity assumption on the angle map $\mapalpha$, there are no relations between angles of distinct train-tracks, we deduce that the meromorphic function $m$ is constant if and only if, for every train-track $\vec{T}$ such that $v_{\vec{T}}\neq 0$, 
\[
\bar{\alpha}_{\vec{T}}=(\alpha_{\vec{T}}-b)\,[\Lambda]\Leftrightarrow\ 2i\Im(\alpha_{\vec{T}})=b\,[\Lambda]\ \Leftrightarrow\ b=0 \text{ and } \Im(\alpha_{\vec{T}})\in\Bigl\{0,\frac{\tau}{2}\Bigr\} + \tau\ZZ. 
\]
The explicit value of $\frac{c_1}{\bar{c}_1}$ can be derived by the evaluation of $m$ at $u=0$. Since this is not needed in the sequel, we do note make this computation explicit. 
Repeating the argument for the second component of the function $\Psi$, we deduce that the angle map $\mapalpha$ takes values in $\RR+\bigl\{0,\frac{\tau}{2}\bigr\}~[\Lambda]$. Note that in this part of the argument we have only used that $\tau\in i\RR^+$, but not that $\Im(\tau)\geq 1$. 

Now consider the case where $\sigma(u)=-\bar{u}+b$, with $b\in\{0,\frac{\tau}{2}\}$. The condition 
$\overline{\Psi(u)}=\Psi(\sigma(u))$ written for the first component reads
\begin{align*}
&\forall\,u\in\TT(\tau),\quad \bar{c}_1\cdot 
e^{-2i\pi \ell_v\cdot \bar{u}}
\prod_{\vec{T}\in\vec{\T}_1}\theta_{1,1}(\bar{u}-\bar{\alpha}_{\vec{T}})^{-v_{\vec{T}}}=
c_1\cdot 
e^{2i\pi \ell_v\cdot (-\bar{u}+b)}
\prod_{\vec{T}\in\vec{\T}_1}\theta_{1,1}(-\bar{u}+b-\alpha_{\vec{T}})^{-v_{\vec{T}}}\\
\Leftrightarrow\ & 
\forall\,u\in\TT(\tau),\quad 
\prod_{\vec{T}\in\vec{\T}_1}\left(\frac{\theta_{1,1}(u-\bar{\alpha}_{\vec{T}})}{\theta_{1,1}(-u+b-\alpha_{\vec{T}})}\right)^{-v_{\vec{T}}}=e^{2i\pi\ell_v \cdot b}\frac{c_1}{\bar{c}_1}. 
\end{align*}
Using again that, by Lemma~\ref{lem:meromorphic_theta_ratio}, the left hand side of the above equation is a meromorphic function of $\TT(\tau)$, we deduce that, for every train-track $\vec{T}$ such that $v_{\vec{T}}\neq 0$, we must have the following condition:
\[
\bar{\alpha}_{\vec{T}}=(-\alpha_{\vec{T}}+b)\,[\Lambda]\Leftrightarrow\ 2\Re(\alpha_{\vec{T}})=b\,[\Lambda]\ \Leftrightarrow\ b=0  \text{ and } \Re(\alpha_{\vec{T}})\in\Bigl\{0,\frac{1}{2}\Bigr\} +\ZZ. 
\]
The evaluation of the left hand side at $u=0$ allows to compute the ratio $\frac{c_1}{\bar{c}_1}$ in this case. 
Repeating the argument for the second component of the function $\Psi$, 
we deduce that the angle map $\mapalpha$ takes values in $\bigl\{0,\frac{1}{2}\bigr\}+ \tau\RR~[\Lambda].$ 

Let us now consider the modular transformation $\tau_1=-\frac{1}{\tau}$, then the range $\Im(\tau)\geq 1$ gets mapped to the range $\Im(\tau)\leq 1$. By Equation~\eqref{eq:tau_1} giving the transformation of $\theta_{1,1}$ under $\tau_1$, we obtain that the first component of the map $\Psi$ can be written as:
\[
z(u,\tau)=c_1\cdot e^{2i\pi \ell_v\cdot u}\prod_{\vec{T}\in\vec{\T}_1}\theta_{1,1}(u-\alpha_{\vec{T}},\tau)^{-v_{\vec{T}}}=
\tilde{c}_1\prod_{\vec{T}\in\vec{\T}_1}\theta_{1,1}(\tau_1 u-\tau_1
\alpha_{\vec{T}},\tau_1)^{-v_{\vec{T}}},
\]
for some explicit constant $\tilde{c}_1$. 
This implies that the range $\bigl\{0,\frac{1}{2}\bigr\}+ \tau\RR~[\Lambda]$ of the angle map $\mapalpha$ gets mapped to the range $\tau_1\bigl(\bigl\{0,\frac{1}{2}\bigr\}+\tau\RR\bigr)=\RR+\bigl\{0,\frac{\tau_1}{2}\bigr\}~[\Lambda]$. We recover the case of the involution $\sigma(u)=\bar{u}$ in the range $\Im(\tau)\leq 1$. This proves that considering the two anti-holomorphic involutions $\sigma(u)=\bar{u}$, and $\sigma(u)=-\bar{u}$ in the range $\tau\in i\RR^+, \Im(\tau)\geq 1$ is equivalent to considering only the involution $\sigma(u)=\bar{u}$ in the full range $\tau\in i\RR^+$, thus ending the proof of Point 1. apart from maximality which is postponed below. 


\underline{Proof of Point 2}. Suppose that $\tau\in \frac{1}{2}+i\RR$ with $\Im(\tau)\geq \frac{\sqrt{3}}{2}$. Then $-\bar{\tau}=\tau-1$, and using Identities~\eqref{equ:complex_conjugation} and~\eqref{eq:tau_2} we have
\[
\overline{\theta_{1,1}(z,\tau)}=\theta_{1,1}(\bar{z},\tau-1)=e^{-i\frac{\pi}{4}}\theta_{1,1}(\bar{z},\tau).
\]
Consider first the case where $\sigma(u)=\bar{u}$. The Condition $\overline{\Psi(u)}=\Psi(\sigma(u))$ written for the first component is:
\begin{align*}
\forall\,u\in\TT(\tau),\quad 
e^{-4i\pi\ell_v\cdot u}
\prod_{\vec{T}\in\vec{\T}_1}\left(\frac{\theta_{1,1}(u-\bar{\alpha}_{\vec{T}})}{\theta_{1,1}(u-\alpha_{\vec{T}})}\right)^{-v_{\vec{T}}}=\frac{c_1}{\bar{c}_1}.
\end{align*}
The same argument as in the case where $\tau\in i\RR^+$ gives us that $2i\Im(\alpha_{\vec{T}})=0\,[\Lambda]$, but this time, since $\tau=\frac{1}{2}+i\RR$, the only solution is $\Im(\alpha_{\vec{T}}) \in \ZZ$. Repeating the argument for the second component, we deduce that the angle map $\mapalpha$ takes values in $\RR~[\Lambda]$. 

Lastly, consider the case where $\sigma(u)=-\bar{u}$. Then, arguing in a similar way, we obtain that $2\Re(\alpha_{\vec{T}})=0\,[\Lambda]$, so that $\Re(\alpha_{\vec{T}})\in\bigl\{0,\frac{1}{2}\bigr\}+\ZZ$, which yields $\alpha_{\vec{T}}\in \bigl\{0,\frac{1}{2}\bigr\}+ i\RR ~[\Lambda]$. 

Note that we have not used the assumption $\Im(\tau)\geq \frac{\sqrt{3}}{2}$. It is here to signify that, up to modular transformations, we can restrict to the standard fundamental domain of the modular parameter $\tau$. 

\underline{Study of maximality of the curve $\C$}. 
Given the explicit parameterization $\Psi$, the real locus of $\C$ is given by the image by $\Psi$ of the points $\{u\in\TT(\tau):\, \Psi(u)=\overline{\Psi(u)}\}$. 
Since the anti-holomorphic involution $\sigma$ commutes with complex conjugation, \emph{i.e.} $\overline{\Psi(u)}=\Psi(\sigma(u))$, the real locus is given by the image of the set of points of $\TT(\tau)$ such that $\sigma(u)=u$. 
In Point 1, this set consists of the image of the two components $\RR+\{0,\frac{\tau}{2}\}~[\Lambda]$. 
Since the map $\Psi$ is bi-rational, this image has two components and the curve $\C$ is maximal. 
In Point 2, when $\sigma(u)=\bar{u}$, then the real locus is the image of $\RR~[\Lambda]$, which consists of a single component and the curve is not maximal; 
when $\sigma(u)=-\bar{u}$, then the real locus is the image of $\{0,\frac{1}{2}\}+ i\RR~[\Lambda]$, which again consists of a single component and the curve is not maximal. 
\qedhere
\end{proof}

\begin{rem}\label{rem:thm_real_spectral_curve}\leavevmode
\begin{enumerate}
\item From the proof of Theorem~\ref{thm:real_spectral_curve}, we gather additional information on the real locus of the curve $\C$. In all cases, the real locus of $\C$ is the image by $\Psi$ of the \emph{real locus of $\TT(\tau)$} defined to be the set of points that are invariant by the anti-holomorphic involution $\sigma$. More specifically, following the case handling of Theorem~\ref{thm:real_spectral_curve}, we have
\begin{enumerate}
 \item The anti-holomorphic involution $\sigma$ of $\TT(\tau)$ is given by $\sigma(u)=\bar{u}$, and the real locus of $\TT(\tau)$ is $\RR+\{0,\frac{\tau}{2}\}~[\Lambda]$. 
 \item The anti-holomorphic involution $\sigma$ of $\TT(\tau)$ is either given by $\sigma(u)=\bar{u}$ in which case the real locus of $\TT(\tau)$ is $\RR~[\Lambda]$; or it is given by $\sigma(u)=-\bar{u}$ in which case the real locus of $\TT(\tau)$ is $\{0,\frac{1}{2}\}+ i\RR~[\Lambda]$.
\end{enumerate}
\item From the proof of Theorem~\ref{thm:real_spectral_curve}, we also know that if the modular parameter $\tau$ satisfies the conditions of either Point 1. or Point 2., then \emph{for all} angle maps $\mapalpha$ satisfying the respective conditions, the corresponding spectral curve has genus 1 and is invariant by complex conjugation; but the curve $\C$ might be singular in general. The genericity assumption on the angle map is related to the curve being non singular. 
\end{enumerate}
\end{rem}

\subsection{Fock's real dimer models}\label{sec:Fock_real_dimer_model}

Consider an infinite, periodic minimal graph $\Gs$, and Fock's dimer model on $\Gs$, with parameters $\tau,\mapalpha,t$, see Definition~\ref{def:Focks_dimer_model}. Assume that Fock's dimer model is \emph{real}, see Definition~\ref{defi:dimer_extras}, recalling that this means that it is gauge equivalent to a dimer model with \emph{real weights}. 

Then as a consequence of gauge equivalence, up to a scale change $(z,w)\leftrightarrow(\lambda z,\mu w)$, the spectral curve of Fock's dimer model arises from a polynomial with real coefficients implying that, up to a scale change, it is invariant by complex conjugation, see Section~\ref{sec:real_spectral_curve}. We thus have that the modular parameter $\tau$ and the angle map $\mapalpha$ must satisfy the conditions of Theorem~\ref{thm:real_spectral_curve}. Recall that this theorem gives no restriction on the parameter $t$; Theorem~\ref{thm:Fock_real_dimer_model} below identifies the values of the parameter $t$ corresponding to real dimer models. 

\begin{thm}\label{thm:Fock_real_dimer_model}
Consider Fock's dimer model on an infinite, periodic minimal graph $\Gs$ with
modular parameter $\tau$ and angle map $\mapalpha$ as in the statement of Theorem~\ref{thm:real_spectral_curve}. Then, up to modular transformations of $\tau$, 
\begin{enumerate}
 \item when $\tau\in i\RR^+$, Fock's dimer model is real for every generic angle map $\mapalpha$ taking values in $\RR + \{0,\frac{\tau}{2}\}~[\Lambda]$, if and only if the parameter $t$ belongs to the real locus $\RR+ \{0,\frac{\tau}{2}\}~[\Lambda]$ of the torus $\TT(\tau)$.
 \item when $\tau\in \frac{1}{2}+i\RR$ with $\Im(\tau)\geq \frac{\sqrt{3}}{2}$, 
 \begin{enumerate}
 \item[$\circ$] the condition ``Fock's dimer model is real for every generic angle map $\mapalpha$ taking values in $\RR~[\Lambda]$'' is never satisfied.
 \item[$\circ$] Fock's dimer model is real for every generic angle map $\mapalpha$ taking values in $ \bigl\{0,\frac{1}{2}\bigr\}+ i\RR~[\Lambda]$, if and only if $t\in \bigl\{\frac{1}{4},\frac{3}{4}\bigr\}+i\RR~[\Lambda]$.
 \end{enumerate}
\end{enumerate}
\end{thm}
\begin{proof}
Fock's dimer model is real if and only if, for every face $\fs$ of $\Gs$,
\begin{equation}\label{equ:real_condition}
\overline{\W_{\mapalpha,t}(\fs)}=\W_{\mapalpha,t}(\fs).
\end{equation}
\underline{Proof of Point 1}. Suppose that $\tau\in i\RR^+$, then by Equation~\eqref{equ:complex_conjugation},
$
\overline{\W_{\mapalpha,t}(\fs)}=\W_{\overline{\mapalpha},\overline{t}}(\fs).
$
Moreover for every $\mapalpha$ taking values in $\{0,\frac{\tau}{2}\}+\RR~[\Lambda]$, we have $\overline{\mapalpha}=\mapalpha\ [\Lambda]$. Recalling that by Lemma~\ref{lem:invariance_weights}, face-weights are meromorphic functions of the angle map $\mapalpha$, we deduce that Condition~\eqref{equ:real_condition} is true for every such $\mapalpha$ if and only if, 
for every $\mapalpha$ taking values in $\{0,\frac{\tau}{2}\}+\RR~[\Lambda]$,
\begin{equation}\label{equ:ratio_to_check_also}
\frac{
\W_{\mapalpha,\overline{t}}(\fs)}{\W_{\mapalpha,t}(\fs)}=1 \ 
\Leftrightarrow 
\prod_{i=1}^{|\fs|/2}
 \frac{\theta_{0,0}(\bar{t}+\mapd(\fs)+\beta_{i}-\alpha_i)}{\theta_{0,0}(\bar{t}+\mapd(\fs)+\beta_{i+1}-\alpha_i)}\frac{\theta_{0,0}(t+\mapd(\fs)+\beta_{i+1}-\alpha_i)}{\theta_{0,0}(t+\mapd(\fs)+\beta_{i}-\alpha_i)}=1.
\end{equation}
Since the graph $\Gs$ is minimal, train-tracks around a vertex are distinct~\cite[Lemma 8]{BCdT:immersions}, so that by assumption the angles $(\alpha_i)$ are generically independent. 
The above condition thus implies that, for every $i\in\{1,\dots,|\fs|/2\}$, for every $\alpha_i\in \{0,\frac{\tau}{2}\}+\RR~[\Lambda]$,
\[
\frac{\theta_{0,0}(\bar{t}+\mapd(\fs)+\beta_{i}-\alpha_i)}{\theta_{0,0}(\bar{t}+\mapd(\fs)+\beta_{i+1}-\alpha_i)}\frac{\theta_{0,0}(t+\mapd(\fs)+\beta_{i+1}-\alpha_i)}{\theta_{0,0}(t+\mapd(\fs)+\beta_{i}-\alpha_i)}=c,
\]
where the term $c$ in the right hand side is constant when seen as a function of $\alpha_i$. 

By Lemma~\ref{lem:meromorphic_theta_ratio}, the left hand side is a meromorphic function of $\alpha_i$ on $\TT(\tau)$ which needs to be equal to a constant on the real component of $\TT(\tau)$, consisting  of two circles, hence it must be equal to a constant on the whole of $\TT(\tau)$. This is true if and only if the zeros and the poles cancel. Using that zeros of the function $\theta_{0,0}$ are located at $\frac{1}{2}+\frac{\tau}{2}\ [\Lambda]$, and using that the angle map $\mapalpha$ is generic, the above equation is true if and only if 
\begin{align*}
&\bar{t}+\mapd(\fs)+\beta_i=t+\mapd(\fs)+\beta_i\ [\Lambda], \text{ and }
\bar{t}+\mapd(\fs)+\beta_{i+1}=t+\mapd(\fs)+\beta_{i+1}\ [\Lambda]\\
\Leftrightarrow \quad & \bar{t}=t\ [\Lambda]\quad \Leftrightarrow\quad  2i\Im(t)=0 \ [\Lambda]\quad \Leftrightarrow\quad t\in \RR+\Bigl\{0,\frac{\tau}{2}\Bigr\}~[\Lambda].
\end{align*}
When $t\in\RR+\{0,\frac{\tau}{2}\}~[\Lambda]$, we now know that the above ratio is constant, but we still need to check that Condition~\eqref{equ:ratio_to_check_also} holds, that is that the constant is indeed equal to 1. This can be checked explicitly using that $\bar{t}=t+r+s\tau$, with $r,s\in\{0,1\}$, using Identity~\eqref{eq:theta:2} and the fact that the angles are fixed. This concludes the proof of Point 1. 

\underline{Proof of Point 2}. Suppose that $\tau\in \frac{1}{2}+i\RR$. Since the modular parameter will change in the forthcoming computations, we write its dependence in the face-weights. By
Equation~\eqref{equ:complex_conjugation}, 
\[
\overline{\W_{\mapalpha,t}(\fs,\tau)}=\W_{\overline{\mapalpha},\overline{t}}(\fs,-\bar{\tau})=
\W_{\overline{\mapalpha},\overline{t}}(\fs,\tau-1)=\W_{\overline{\mapalpha},\overline{t}+\frac{1}{2}}(\fs,\tau),
\]
where in the last equality we used the second equality of Lemma~\ref{lem:modular_transo_dimers}. 

$\circ$ Suppose first that the angle map $\mapalpha$ takes values in $\RR~[\Lambda]$, then we have $\overline{\mapalpha}=\mapalpha\ [\Lambda]$, so that Condition~\eqref{equ:real_condition} is true for every such $\mapalpha$ if and only if, for every $\mapalpha$ taking values in $\RR~[\Lambda]$,
\[
\frac{\W_{\mapalpha,\overline{t}+\frac{1}{2}}(\fs,\tau)}{\W_{\mapalpha,t}(\fs,\tau)}=1.
\]
Arguing as in the case $\tau\in i\RR^+$, this is true if and only if 
\[
\bar{t}+\frac{1}{2}=t\ [\Lambda] \quad \Leftrightarrow\quad 2i\Im(t)=\frac{1}{2}\ [\Lambda] 
\quad \Leftrightarrow\quad i\Im(t)=\frac{1}{4}\ [\Lambda/2]. 
\]
This condition is never satisfied on the lattice $\Lambda/2$ thus proving the first part of Point 2.

$\circ$ Suppose that the angle map $\mapalpha$ takes values in $\bigl\{0,\frac{1}{2}\bigr\}+i\RR~[\Lambda]$, then we have $\overline{\mapalpha}=-\mapalpha\ [\Lambda]$,  so that Condition~\eqref{equ:real_condition} is true for every such $\mapalpha$ if and only if, for every $\mapalpha$ taking values in $\bigl\{0,\frac{1}{2}\bigr\}+i\RR~[\Lambda]$,
\begin{equation}\label{equ:ratio_to_check_once_more}
\frac{\W_{-\mapalpha,\overline{t}+\frac{1}{2}}(\fs,\tau)}{\W_{\mapalpha,t}(\fs,\tau)}=1.
\end{equation}
Since the computation is a bit more delicate, let us detail the next step. The above implies the following condition
\[
\frac{\theta_{0,0}(\bar{t}+\frac{1}{2}-\mapd(\fs)-\beta_{i}+\alpha_i)}{\theta_{0,0}(\bar{t}+\frac{1}{2}-\mapd(\fs)-\beta_{i+1}+\alpha_i)}\frac{\theta_{0,0}(t+\mapd(\fs)+\beta_{i+1}-\alpha_i)}{\theta_{0,0}(t+\mapd(\fs)+\beta_{i}-\alpha_i)}=c,
\]
which is equivalent to the following condition
\begin{align*}
&-\bar{t}-\frac{1}{2}+\mapd(\fs)+\beta_i=t+\mapd(\fs)+\beta_i\ [\Lambda], \text{ and }
-\bar{t}-\frac{1}{2}+\mapd(\fs)+\beta_{i+1}=t+\mapd(\fs)+\beta_{i+1}\ [\Lambda]\\
\quad \Leftrightarrow & -\bar{t}-\frac{1}{2}=t \quad \Leftrightarrow \quad 2\Re(t)=-\frac{1}{2} \ [\Lambda]
\quad \Leftrightarrow \quad \Re(t)=-\frac{1}{4} \ [\Lambda/2]
\quad \Leftrightarrow\quad t\in \Bigl\{\frac{1}{4},\frac{3}{4}\Bigr\}+ i\RR~[\Lambda].
\end{align*}
The fact that the ratio~\eqref{equ:ratio_to_check_once_more} is indeed equal to 1 when $2\Re(t)=-\frac{1}{2} \ [\Lambda]$ is checked similarly to Point 1. This ends the proof of the second part of Point 2. 
\end{proof}

\begin{rem}\label{rem:thm_real_dimer_model}
Similarly to Point 2. of Remark~\ref{rem:thm_real_spectral_curve}, we gather from the proof of Theorem~\ref{thm:Fock_real_dimer_model} that if the modular parameter $\tau$ satisfies the conditions of either Point 1. or Point 2., then \emph{for all} angle maps $\mapalpha$ and parameters $t$ satisfying the respective conditions, Fock's corresponding dimer model is real.
\end{rem}


\section{Classification of (frustrated) Ising models in genus 1} \label{sec:classification}

Consider an Ising model on an infinite, periodic, isoradial graph $G$ with real coupling constants $\eps\Js$, where the definition of isoradial graph is recalled in Section~\ref{subsec:background}. 
Consider the corresponding Ising-dimer model on the infinite, periodic graph $\GQ$, see Definition~\ref{defi:Ising_dimer_model}. Then, by Lemma~\ref{lem:isoradial_minimal} of the forthcoming Section~\ref{subsec:train_tracks}, the graph $\GQ$ is minimal. Let $\C$ be the spectral curve of this dimer model; assume that it is irreducible and has genus 1.
Suppose that a vertex of $\GQ$ is distinguished, allowing to define the spectral data of this Ising-dimer model on $\GQ$~\cite{KOS,GoncharovKenyon,GGK}. Then, by Fock~\cite{Fock}, there exists a choice of modular parameter $\tau$, a periodic angle map $\mapalpha$, and a parameter $t\in\TT(\tau)$, such that the dimer model with Fock's weights on $\GQ$, see Definition~\ref{def:Focks_dimer_model}, is gauge equivalent to the Ising-dimer model on $\GQ$ with the given spectral data. The goal of this section is to characterize, and classify the Ising-dimer models on $\GQ$ through the set of parameters $\tau,\mapalpha,t$ of Fock's gauge equivalent dimer model. This is in the same spirit as the work of George~\cite{George} with the following difference: George uses the spectral data approach of~\cite{KO,GoncharovKenyon}, whereas we use Fock's approach~\cite{Fock}, allowing to have a very concrete model in the end, amenable to further study. Also, we allow for \emph{real} coupling constants, and not only positive ones as in~\cite{George}, thus exiting the well behaved world of Harnack curves; but we restrict to genus 1 while~\cite{George} considers general genus spectral curves. 

Before outlining the content of this section, let us summarize the information that we can gather on Fock's gauge equivalent dimer model on $\GQ$ with parameters $\tau,\mapalpha,t$, using Section~\ref{sec:real_dimers}. Since the curve $\C$ arises from a polynomial with real coefficients, the curve is invariant by complex conjugation implying that the parameters $\tau,\mapalpha$ satisfy the conditions of Theorem~\ref{thm:real_spectral_curve}. The fact that $\C$ arises from a polynomial with real coefficients also implies that Fock's gauge equivalent dimer model is real, see Definition~\ref{defi:dimer_extras}, implying that the parameter $t$ satisfies the conditions of Theorem~\ref{thm:Fock_real_dimer_model}. This motivates the following.

\begin{defi}\label{def:Necessary_I}
We say that the parameters $\tau,\mapalpha,t$ of Fock's dimer model on $\GQ$ satisfy
\emph{Necessary Conditions I} if, up to modular transformations of $\tau$:
\begin{enumerate}
\item[1. ] either, the modular parameter $\tau\in i\RR^+$, in which case the angle map $\mapalpha$ takes values in $\RR+\{0,\frac{\tau}{2}\}~[\Lambda]$, and the parameter $t$ belongs to $\RR+ \{0,\frac{\tau}{2}\}~[\Lambda]$,
\item[2. ] or, the modular parameter $\tau\in \frac{1}{2}+i\RR$ with $\Im(\tau)>\frac{\sqrt{3}}{2}$, in which case the angle map $\mapalpha$ takes values in $ \bigl\{0,\frac{1}{2}\bigr\}+ i\RR~[\Lambda]$, and the parameter 
$t\in \bigl\{\frac{1}{4},\frac{3}{4}\bigr\}+ i\RR~[\Lambda]$.
\end{enumerate}
Recall that, by Points 2.~of Remarks~\ref{rem:thm_real_spectral_curve} and~\ref{rem:thm_real_dimer_model}, we do not need to assume that the angle map $\mapalpha$ is generic in order to have a spectral curve $\C$ that it is invariant by complex conjugation, and such that Fock's dimer model is real. This justifies that we take \emph{all} angle maps taking values in the respective real loci of the torus $\TT(\tau)$.
\end{defi}

\begin{rem}\leavevmode
\begin{enumerate}
 \item Suppose that the Ising model is defined on an infinite, periodic, planar graph $G$ (not necessarily isoradial), and consider its corresponding Ising-dimer model on $\GQ$ with spectral curve $\C$. Then, there exists a dimer model with Fock's weights on an infinite, periodic, \emph{minimal} graph $\Gs$ which has the same spectral curve $\C$, up to a global scaling; but this time the graph $\Gs$ cannot easily be related to the graph $\GQ$.  
 \item Recall that by Lemma~\ref{lem:modular_transo_dimers}, we know that Fock's dimer model is invariant under modular transformations of the modular parameter $\tau$, up to an appropriate transformation of the angle map $\mapalpha$ and of the parameter $t$. As a consequence, we only need to understand the model when the modular parameter is restricted to the appropriate subset of the standard fundamental domain, see Equation~\eqref{equ:fund_domain_modular}.
\end{enumerate}
\end{rem}

With Definition~\ref{def:Necessary_I} at hand, we now turn to the content of this section. 
By Point 2. of Remark~\ref{rem:Fisher_Dub}, we moreover know that the spectral curve $\C$ of the Ising-dimer model on $\GQ$ is \emph{centrally symmetric}, \emph{i.e.}, $(z,w)\in\C$ if and only if $\bigl(\frac{1}{z},\frac{1}{w}\bigr)\in\C$. This property has further implications on the angle map $\mapalpha$, which is the subject of Theorem~\ref{thm:real_symmetric_spectral_curve} of Section~\ref{sec:Fock_necessary_condition_II}, leading to Necessary Conditions II. Note that Necessary Conditions I and II are motivated by models defined on infinite, periodic, minimal graphs, but they still make sense when the graph $\GQ$ is not periodic. Hence, when not needed, we will not make the periodicity assumption. In Lemma~\ref{lem:alternate_product_faces} we explicitly compute face-weights of Fock's dimer model on infinite, minimal graphs $\GQ$ for the three kinds of faces. In Lemma~\ref{lem:duality}, we prove a duality property for Fock's dimer model. 

Fock's dimer model and the Ising-dimer model on an infinite, minimal graph $\GQ$ are gauge equivalent if and only if their face-weights are equal at all faces. Assuming that the parameters $\tau,\mapalpha,t$ of Fock's dimer model satisfy Necessary Conditions II, Sections~\ref{sec:Necessary_III} to~\ref{sec:Necessary_V} aim at characterizing the parameters $\tau,\mapalpha,t$ that arise from gauge equivalent Ising-dimer models on $\GQ$. More precisely, we have the following. 
Recall that faces of $\GQ$ are of three kinds: faces corresponding to squares $y=\{e,e^*\}$, where $e,e^*$ is a pair of primal/dual edge of $G/G^*$, those corresponding to vertices $v$ of $G$, and those corresponding to dual vertices $f$ of $G^*$.
The subject of Section~\ref{sec:Necessary_III} is to see the implications of equality of face-weights at square faces $y=\{e,e^*\}$:
\begin{equation}\label{equ:condition_square_face_yet_once_more}
\WI_{\eps\Js}(y)=\W_{\mapalpha,t}(y).
\end{equation}
In particular $\W_{\mapalpha,t}(y)$ must be negative since $\WI_{\eps\Js}(y)$ is, see Equation~\eqref{equ:gauge_Ising_dimers}. 
The first result is Proposition~\ref{prop:maximality_spectral_curve} proving that necessarily $\tau\in i\RR^+$. An important implication of this, see Corollary~\ref{cor:spectral_curve_max}, is that when the graph $\GQ$ is periodic, the spectral curve $\C$ is \emph{maximal}.
Then, in Proposition~\ref{prop:first_parameter_condition}, we derive a local condition necessarily satisfied by the angle map $\mapalpha$ and the parameter $t$, and turn it into a global condition in Proposition~\ref{prop:global_condition_anglemap}. 
This leads to the definition of Necessary Conditions III, and of the absolute value coupling constants $\Js_{(\mapalpha,\rho,t)}$ satisfying~\eqref{equ:condition_square_face_yet_once_more}. 

In Section~\ref{sec:Necessary_IV}, we study implications of equality of \emph{squared} face-weights at faces corresponding to vertices $v$ of $G$ and dual vertices $f$ of $G^*$:
\begin{equation*}
\forall\,f\in V^*,\ v\in V,\quad 
( \WI_{\eps\Js_{(\mapalpha,\rho,t)}} )^2 (f)=( \W_{\mapalpha,t} )^2(f),\quad
( \WI_{\eps\Js_{(\mapalpha,\rho,t)}} )^2 (v)=( \W_{\mapalpha,t} )^2(v).
\end{equation*}
This leads to further conditions on the angle map $\mapalpha$ and the parameter $t$, which is the content of Proposition~\ref{prop:t} and the subject of Necessary Conditions IV. 
Face-weights satisfying these conditions are computed in Lemma~\ref{lem:face_weights_again} using Jacobi elliptic functions. 

In Proposition~\ref{prop:Necessary_V} of Section~\ref{sec:Necessary_V}, we explore a last condition on the angle map $\mapalpha$ needed to have equality of \emph{face-weights}, and not only of \emph{squared face-weights}, leading to the definition of Necessary Conditions V. This characterization of parameters $\tau,\mapalpha,t$ of Fock's dimer model arising from gauge equivalent Ising-dimer models on $\GQ$ allows us to prove the main result of this paper, namely 
Theorem~\ref{thm:main} of Section~\ref{sec:classification_thm}, classifying Ising models; we then define the notion of algebraic phase transition, infer algebraic critical temperatures, and explain the relations to the corresponding notions arising from physics.

\subsection{Centrally symmetric dimer models on $\GQ$ with real edge-weights}\label{sec:Fock_necessary_condition_II}

Consider the Ising-dimer model on an infinite, periodic, minimal graph $\GQ$, with periodic coupling constants $\eps\Js$. Let $\C$ be the associated spectral curve, assume that it has genus 1 and satisfies Assumptions $(\dagger)$. Then by Point 2. of Remark~\ref{rem:Fisher_Dub}, we know that this curve is centrally symmetric. Consider also Fock's gauge equivalent dimer model on $\GQ$ with parameters $\tau,\mapalpha,t$, thus satisfying Necessary Conditions~I. The main result of this section is Theorem~\ref{thm:real_symmetric_spectral_curve} of Section~\ref{subsec:central_symmetry} identifying the values of the angle map $\mapalpha$ corresponding to curves that have the property of being centrally symmetric, leading to Necessary Conditions II. We start with Section~\ref{subsec:train_tracks}, highlighting useful properties of train-tracks of the graph $\GQ$. Finally, in Lemma~\ref{lem:alternate_product_faces} of Section~\ref{subsec:face_weights_Fock}, we explicitly compute face-weights of Fock's dimer model on $\GQ$ when it satisfies Necessary Conditions I and II. 

\subsubsection{Train-tracks of the graph $G$ and of the graph $\GQ$}\label{subsec:train_tracks}

\paragraph{Train-tracks of $G$ and of $\GQ$, and angle map.} Let $G$ be an infinite, planar graph as in Section~\ref{subsec:background}. Recall that $\GGR$ is the associated quad-graph, and that $\T$ is the set of train-tracks, consisting of paths in the dual graph $(\GGR)^*$ crossing opposite sides of quadrangles. A priori, there is no natural consistent orientation of the train-tracks of $\T$, see Figure~\ref{fig:tt_G_GQ} (left). 

\begin{figure}[h]
  \centering
  \begin{overpic}[width=\linewidth]{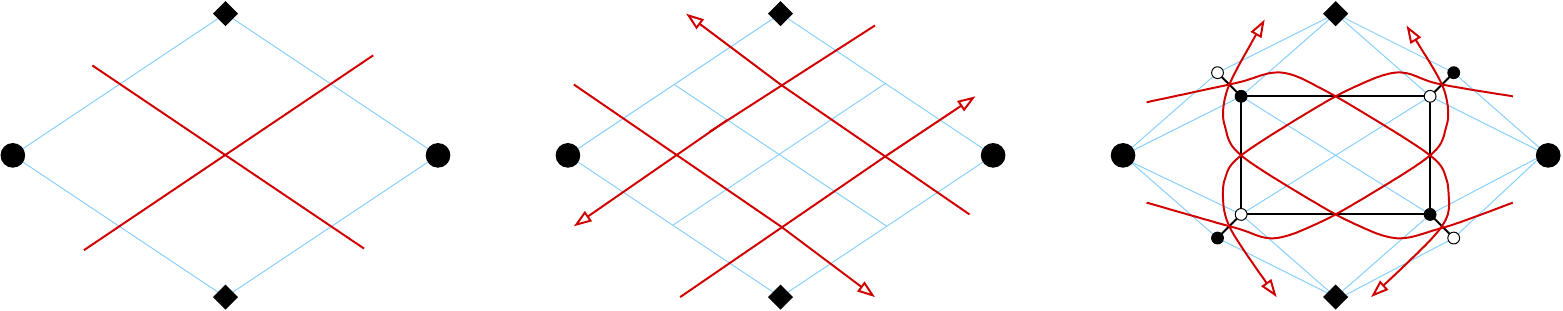}
   \end{overpic}
\caption{Comparison of the set of train-tracks $\T$ (left), $\vec{\T}_0$ (center) and $\vec{\T}$ (right).}
\label{fig:tt_G_GQ}
\end{figure}

Consider the graph $\GRR_0$ obtained from $\GR$ by cutting all quadrangles into four. The corresponding set of train-tracks can be consistently oriented so that vertices of $G$ and $G^*$ are on the right, see Figure~\ref{fig:tt_G_GQ} (center), and we denote this set by $\vec{\T}_0$. To each train-track $T$ of $\T$, there corresponds a pair $\{\vec{T},\cev{T}\}$ of non-intersecting train-tracks of $\vec{\T}_0$, oriented in opposite directions. As a consequence, two oriented train-tracks of $\vec{\T}_0$ intersect if and only of they arise from distinct intersecting train-tracks of $\T$. 

Now consider the graph $\GQ$, and the associated quad-graph $\GRR$; there are two kinds of quadrangles: those arising from quadrangles of $\GRR_0$, and ``flat'' quadrangles corresponding to external edges of $\GQ$, see Figure~\ref{fig:tt_G_GQ} (right). Since the graph $\GQ$ is bipartite, train-tracks are consistently oriented so that white vertices of $\GQ$ are on the left, see Section~\ref{subsec:background}, and this set is denoted by $\vec{\T}$. The set of train-tracks $\vec{\T}$ is obtained from $\vec{\T}_0$ by adding a crossing between train-tracks of a pair at each ``flat'' quadrangle, see Figure~\ref{fig:tt_G_GQ} (center and right). As a consequence, no train-track is deleted or created, the only new intersections are between train-tracks of a pair, and they arise in opposite direction. We thus have that, to every (unoriented) train-track $T$ of $\T$, there corresponds a pair $\{\vec{T}$, $\cev{T}\}$ of oriented train-tracks of $\vec{\T}$. Note that, without further assumption on $G$, there is no consistent way of assigning arrows to the pair $\{\vec{T},\cev{T}\}$ just by looking at $T$; we nevertheless use the notation $\vec{T},\cev{T}$ implying that a choice of labeling has been made, this is possible but this labeling does not come from an information we have on $T$. The relation between $G$ and $\GQ$ implies the following. 
\begin{lem}\label{lem:isoradial_minimal}
The graph $G$ is isoradial if and only if the graph $\GQ$ is minimal.  
\end{lem}
\begin{proof}
Returning to the construction of the set of oriented train-tracks $\vec{\T}$ from $\T$, we have that train-tracks self intersect in $\T$ if and only of they do so in $\vec{\T}$. We also have that two distinct train-tracks of $\T$ intersect more than once, if and only if two distinct oriented train-tracks of $\vec{\T}$ arising from a distinct pair of train-tracks of $\T$ intersect in the same direction, which is forbidden in minimal graphs. The only additional crossings that we have in $\vec{\T}$ and not in $\T$ come from oriented train-tracks of a pair, which intersect in the opposite direction, and this is allowed in minimal graphs. 
\end{proof}

Consider a train-track $T$ of $\T$, and a choice of labeling $\vec{T},\cev{T}$ for the two corresponding oriented train-tracks of $\vec{\T}$. Then, the angle map $\mapalpha$ assigns angles $\alpha_{\vec{T}}$ and $\alpha_{\cev{T}}$ to the train-tracks $\vec{T},\cev{T}$. 
When the graph $G$ is furthermore periodic, let us denote by $\T_1$ the set of train-tracks of $G_1$, and by $\vec{\T}_1$ the set of oriented train-tracks of $\GQ_1$. Observe that the homology classes of train-tracks of 
$\vec{\T}_1$ come in opposite pairs: for every $T\in\T_1$, $(h_{\vec{T}},v_{\vec{T}})=-(h_{\cev{T}},v_{\cev{T}})$.

\subsubsection{Centrally symmetric dimer models on $\GQ$}\label{subsec:central_symmetry}

Consider the Ising-dimer model on an infinite, periodic, minimal graph $\GQ$ with associated 
spectral curve $\C$ of genus 1 satisfying Assumptions $(\dagger)$. Consider Fock's gauge equivalent dimer model with parameters $\tau,\mapalpha,t$ satisfying Necessary Conditions I. Then the spectral curve $\C$ and that of Fock's dimer model are related by a global scaling $(z,w)\leftrightarrow (\lambda z,\mu w)$, with $(\lambda,\mu)\in(\CC^*)^2$, see Point 2. of Remark~\ref{rem:spectral_curve}. A bi-rational parameterization of $\C$ is given by the map $\Psi$ of Equation~\eqref{equ:param_spectral_curve} on the graph $\GQ$, with $c_1,c_2\in\CC$. Using the specific structure of the train-tracks exhibited in Section~\ref{subsec:train_tracks}, we suppose that a labeling $\vec{T},\cev{T}$ of the pairs of train-tracks is chosen. The map $\Psi:\TT(\tau)\rightarrow \C$ can then further be written as $\Psi(u)=(z(u),w(u))$, with $z,w$ given by 
{\small 
\begin{equation}\label{equ:param_spectral_curve_symm}
z(u)=c_1\cdot e^{2i\pi \ell_v\cdot u}
\prod_{T\in \T_1}\left(\frac{\theta_{1,1}(u-\alpha_{\vec{T}})}{\theta_{1,1}(u-\alpha_{\cev{T}})}\right)^{-v_{\vec{T}}},\quad 
w(u)=c_2\cdot e^{-2i\pi \ell_h\cdot u}
\prod_{T\in \T_1}\left(\frac{\theta_{1,1}(u-\alpha_{\vec{T}})}{\theta_{1,1}(u-\alpha_{\cev{T}})}\right)^{h_{\vec{T}}},
\end{equation}
}
for some $c_1,c_2\in\CC$, where 
\begin{equation}\label{cond:angle_map_periodic_symm}
\sum_{\vec{T}\in\vec{\T}_1}v_{\vec{T}}(\alpha_{\vec{T}}-\alpha_{\cev{T}})=\ell_v\tau ,\quad 
\sum_{\vec{T}\in\vec{\T}_1}h_{\vec{T}}(\alpha_{\vec{T}}-\alpha_{\cev{T}})=\ell_h\tau ,
\end{equation}
with $\ell_v,\ell_h\in\ZZ$. 

Recall from Point 2.~of Remark~\ref{rem:Fisher_Dub} that the curve $\C$ has the additional property of being centrally symmetric. Since the curve $\C$ is non-singular (apart for isolated real nodes), the angle map is generic in the sense that there are a priori no relations between angles of oriented train-tracks of $\vec{\T}$ arising from distinct train-tracks of $\T$, as linear relations for example. Theorem~\ref{thm:real_symmetric_spectral_curve} below identifies the values of the angle map $\mapalpha$ of Fock's gauge equivalent dimer model when the curve $\C$ has this additional property.

\begin{thm}\label{thm:real_symmetric_spectral_curve}
Consider the Ising-dimer model on an infinite, periodic minimal graph $\GQ$ with periodic real edge-weights and genus 1 spectral curve $\C$ satisfying $(\dagger)$. Consider Fock's gauge equivalent dimer model with parameters $\tau$, generic angle map $\mapalpha$, $t$ satisfying Necessary Conditions I. Then, the spectral curve $\C$ is centrally symmetric if and only if, up to modular transformations of the modular parameter $\tau$, for every train-track $T$ of $\T_1$, 
\begin{enumerate}
\item when $\tau\in i\RR^+$, then 
$\alpha_{\cev{T}}=\alpha_{\vec{T}}+\rho\ [\Lambda]$,
for some half-period vector $\rho=\frac{j}{2}+\frac{\ell}{2}\tau,\ (j,\ell)\in\{0,1\}^2$, with $(j,\ell)\neq (0,0)$,
\item when $\tau\in \frac{1}{2}+i\RR$, with $\Im(\tau)>\frac{\sqrt{3}}{2}$, then $\alpha_{\cev{T}}=\alpha_{\vec{T}}+\rho\ [\Lambda]$, with $\rho=\frac{1}{2}$.
\end{enumerate}
\end{thm}

\begin{proof}
Consider the birational parameterization $\Psi$ of the spectral curve $\C$ given by~\eqref{equ:param_spectral_curve_symm}. By central symmetry of the curve $\C$, we have that 
$(z,w)\in\C$ if and only if $\bigl(\frac{1}{z},\frac{1}{w}\bigr)\in\C$, implying that $\frac{1}{\Psi(u)}:=\bigl(\frac{1}{z(u)},\frac{1}{w(u)}\bigr)$ yields another parameterization of the curve $\C$. By~\cite[Corollary, Section 86]{DuVal}, this means 
that there exists a holomorphic automorphism $\sigma$ of the torus $\TT(\tau)$ such that $1/\Psi(u)=\Psi(\sigma(u))$. The further condition $1/(1/\Psi(u))=\Psi(u)$ implies that $\sigma$ must be an involution. The condition of being a holomorphic involution on $\sigma$ implies that $\sigma=au+b$ with $a=1$ and $b\in\{0,\frac{1}{2},\frac{\tau}{2},\frac{1}{2}+\frac{\tau}{2}\}$, or $a=-1$ and $b\in\CC$. 

$\circ$ Consider first the case where $a=1$ and $b\in\{0,\frac{1}{2},\frac{\tau}{2},\frac{1}{2}+\frac{\tau}{2}\}$. Then, looking at the first component of the map $\Psi$, we must have $\frac{1}{z(u)}= z(u+b)$, that is 
\begin{align*}
&c_1^{-1}\cdot e^{-2i\pi \ell_v\cdot u} 
\prod_{T\in \T_1}\left(\frac{\theta_{1,1}(u-\alpha_{\cev{T}})}{\theta_{1,1}(u-\alpha_{\vec{T}})}\right)^{-v_{\vec{T}}}=
c_1 e^{2i\pi \ell_v\cdot (u+b)} 
\prod_{T\in \T_1}\left(\frac{\theta_{1,1}(u+b-\alpha_{\vec{T}})}{\theta_{1,1}(u+b-\alpha_{\cev{T}})}\right)^{-v_{\vec{T}}}\\
\Leftrightarrow\, & e^{-4i\pi \ell_v\cdot u} \prod_{T\in \T_1}
\left(\frac{\theta_{1,1}(u-\alpha_{\cev{T}})}{\theta_{1,1}(u-\alpha_{\vec{T}})}
\frac{\theta_{1,1}(u+b-\alpha_{\cev{T}})}{\theta_{1,1}(u+b-\alpha_{\vec{T}})}\right)^{-v_{\vec{T}}}=(c_1)^2\cdot e^{2i\pi \ell_v b}. 
\end{align*}
The argument then resembles that of the proof of Theorem~\ref{thm:real_spectral_curve}. By Lemma~\ref{lem:meromorphic_theta_ratio}, the left hand side is a meromorphic function on $\TT(\tau)$ which must be equal to a constant, namely $(c_1)^2\cdot e^{2i\pi \ell_v b}$. This meromorphic function is equal to a constant if and only if the zeros of the numerator and denominator cancel out. Using that the zeros of the function $\theta_{1,1}(\,\cdot\,)$ are the points $0\,[\Lambda]$, and using that by our genericity assumption on the angle map $\mapalpha$ there are no relations between pairs of oriented train-tracks $\vec{T},\cev{T}$ corresponding to distinct train-tracks $T$, we deduce that, for every train-track $T$ such that $v_{\vec{T}}\neq 0$, we have
\begin{itemize}
 \item[$\circ$] either $\alpha_{\vec{T}}=\alpha_{\cev{T}}$, but then this cancels a zero and a pole of the function $z$, thus changing the number of zeros and poles of $z$, which is not possible,
 \item[$\circ$] or $\alpha_{\cev{T}}=\alpha_{\vec{T}}-b\,[\Lambda],\text{ and } \alpha_{\vec{T}}=\alpha_{\cev{T}}-b\,[\Lambda]$. When $b\in\{0,\frac{1}{2},\frac{\tau}{2},\frac{1}{2}+\frac{\tau}{2}\}$, the two conditions are compatible, and give 
$\alpha_{\cev{T}}=\alpha_{\vec{T}}-b\,[\Lambda]$. For the same reason as above, we must exclude the case $b=0$.
\end{itemize}
We thus have $\alpha_{\cev{T}}=\alpha_{\vec{T}}+\rho\ [\Lambda]$ for some half-period vector $\rho$ as in the statement, with $(j,\ell)\neq (0,0)$. The explicit value of the constant $(c_1)^2\cdot e^{2i\pi \ell_v b}$ can be computed by evaluating the meromorphic function at $u=0$. Since this is not needed in the sequel, we do not perform this computation. Repeating the argument for the second component of the map $\Psi$ yields that this holds for all train-tracks. 

We must now take into account the fact that the parameters $\tau,\mapalpha$ satisfy Necessary Condition~I. When $\tau\in i\RR^+$, 
$\mapalpha$ takes values in $\RR+ \{0,\frac{\tau}{2}\}~[\Lambda]$, which is compatible with the condition $\alpha_{\cev{T}}=\alpha_{\vec{T}}+\rho$, ending the proof of Point 1. When $\tau\in \frac{1}{2}+i\RR$, $\mapalpha$ takes values in $\bigl\{0,\frac{1}{2}\bigr\}+i\RR~[\Lambda]$. 
This condition is compatible with $\alpha_{\cev{T}}=\alpha_{\vec{T}}+\rho\ [\Lambda]$ if and only if $\ell=0$, thus ending the proof of Point 2. 

$\circ$ Consider now the case where $a=-1$, and $b\in\CC$. Since the map $(z,w)\mapsto(\frac{1}{z},\frac{1}{w})$ preserves the real locus of the curve, and since the orientation is preserved on each of the real components, 
the holomorphic involution $\sigma(u)=-u+b$, which is orientation reversing, cannot occur.  
\end{proof}
\begin{rem}\label{rem:thm_real_symmetric_spectral_curve}
Form the proof of Theorem~\ref{thm:real_symmetric_spectral_curve}, we also gather that we do not need the angle map $\mapalpha$ to be generic in order for the spectral curve to have central symmetry; this holds for all angle maps satisfying the conditions of Point 1. or Point 2. 
\end{rem}
Theorem~\ref{thm:real_symmetric_spectral_curve} and Remark~\ref{rem:thm_real_symmetric_spectral_curve} motivate the following. 

\begin{defi}\label{def:necessary_II}
We say that the parameters $\tau,\mapalpha,t$ of Fock's dimer model on $\GQ$ satisfy \emph{Necessary Conditions II}, if they satisfy Necessary Conditions I and if, up to modular transformations of $\tau$,
\begin{enumerate}
\item either $\tau\in i\RR^+$, then $
\alpha_{\cev{T}}=\alpha_{\vec{T}}+\rho\ [\Lambda]$,
for some half-period vector $\rho=\frac{j}{2}+\frac{\ell}{2}\tau,\ (j,\ell)\in\{0,1\}^2$ with $(j,\ell)\neq (0,0)$,
\item or $\tau\in \frac{1}{2}+i\RR$, with $\Im(\tau)>\frac{\sqrt{3}}{2}$, then $\alpha_{\cev{T}}=\alpha_{\vec{T}}+\rho\ [\Lambda]$, with $\rho=\frac{1}{2}$. 
\end{enumerate}
To emphasize the dependence in $\rho$, we now write the parameters of Fock's dimer model as $\tau,\mapalpha,\rho,t$. Note that, as in Definition~\ref{def:Necessary_I}, periodicity of the graph $\GQ$ motivates this definition but is not required for it to make sense. 
\end{defi}

\subsubsection{Face-weights of Fock's dimer model on $\GQ$ and duality}\label{subsec:face_weights_Fock}

Consider Fock's dimer model on an infinite minimal graph $\GQ$ with parameters $\tau,\mapalpha;\rho,t$ satisfying Necessary Conditions II. In this section, we do not need to assume that the graph $\GQ$ is periodic; we explicitly compute face-weights. Instead of denoting them by $\W_{\mapalpha,t}(\, \cdot\, ,\tau)=\W_{\mapalpha,t}(\, \cdot\, )$ as in Section~\ref{sec:def_Fock_general}, let us denote them by $\W_{\mapalpha,\rho,t}(\, \cdot\, ,\tau)=\W_{\mapalpha,\rho,t}(\, \cdot\, )$ emphasizing the dependence in $\rho$ stemming from Necessary Conditions II. 

\paragraph{Notation.} Recall that there is no natural way of globally assigning the arrows to the pair of oriented train-tracks $\vec{T},\cev{T}$ of $\vec{\T}$ corresponding to a train-track $T$ of $\T$, but locally around each face of $\GQ$, a consistent choice can be made, justifying the upcoming notation. For the angles $\alpha_{\vec{T}}$, $\alpha_{\cev{T}}$, when no reference to the train-tracks is needed, we will simply denote the pair of angles as $\alpha,\alpha'$.

%

Around a face corresponding to a dual vertex $f$ of degree $|f|$ of $G$, see Figure~\ref{fig:primal_dual} (left), the angles $(\alpha_i)$ are associated to train-tracks turning clockwise around $f$, and then angles $(\beta_i)$ are associated to train-tracks turning counterclockwise, with $i\in\{1,\dots,|f|\}$, and the respective numbering being such that $\beta_{i+1}=\alpha_i'$. Around a face corresponding to a vertex $v$ of degree $|v|$ of $G$, see Figure~\ref{fig:primal_dual} (right), the angles $(\alphat_i)$ correspond to train-tracks turning counterclockwise, and the angles $(\betat_i)$ correspond to train-tracks turning clockwise, with $i\in\{1,\dots,|v|\}$, and respective labeling being such that $\betat_{i+1}=\alphat_i'$. 

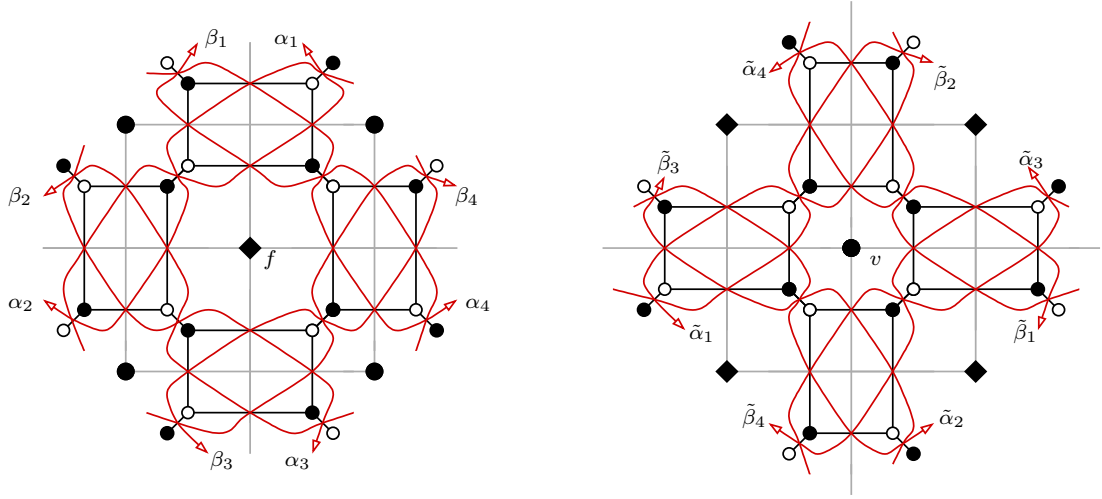
\begin{figure}[h]
  \centering
  \begin{overpic}[width=14cm]{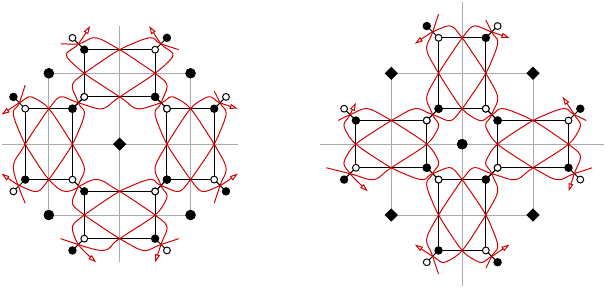}
   \put(22,43){\scriptsize $\alpha_1$}
   \put(15.5,43){\scriptsize $\beta_1$}
   \put(-3,18){\scriptsize $\alpha_2$}
   \put(-3,28){\scriptsize $\beta_2$}
   \put(23,3){\scriptsize $\alpha_3$}
   \put(16,3){\scriptsize $\beta_3$}
   \put(40,18){\scriptsize $\alpha_4$}
   \put(39,28){\scriptsize $\beta_4$}
   \put(21,22){\scriptsize $f$}
   \put(61,15){\scriptsize $\alphat_1$}
   \put(91.5,15){\scriptsize $\betat_1$}
   \put(84.5,7){\scriptsize $\alphat_2$}
   \put(84,39){\scriptsize $\betat_2$}
   \put(92,31.5){\scriptsize $\alphat_3$}
   \put(58,31){\scriptsize $\betat_3$}
   \put(66,40){\scriptsize $\alphat_4$}
   \put(66,7){\scriptsize $\betat_4$}
   \put(78,22){\scriptsize $v$}
   \end{overpic}
\caption{Left: notation around a face corresponding to a dual vertex $f$ of $G$. Right: notation around a face corresponding to a vertex $v$ of $G$. }
\label{fig:primal_dual}
\end{figure}

Around a face corresponding to a square $y=\{e,e^*\}$ of $\GQ$, we use $\alpha_e,\beta_e,\alpha_e',\beta_e'$ as in Figure~\ref{fig:square_face} (left), where the angles $\alpha_e,\beta_e$, and $\alpha_e',\beta_e'$, are the angles of the two edges of the square corresponding to an edge $e$ of $G$; we also use the notation $\alphat_e:=\alpha_{e^*}\,\betat_e=\beta_{e^*},\alphat'_e=\alpha'_{e^*},\betat'_e=\beta'_{e^*}$, where $\alphat_e,\betat_e$, and $\alphat'_e,\betat'_e$, are the angles of the two edges of the square corresponding to a dual edge $e^*$ of $G^*$, see Figure~\ref{fig:square_face} (right).

\begin{figure}[h]
  \centering
  \begin{overpic}[width=8cm]{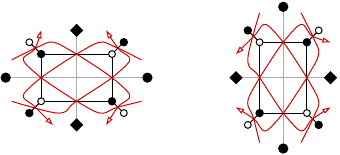} 
   \put(5,20){\scriptsize $e$}
   \put(30,38){\scriptsize $\alpha_e$}
   \put(12,38){\scriptsize $\beta_e$}
   \put(14,6){\scriptsize $\alpha_e'$}
   \put(31,6){\scriptsize $\beta_e'$}
   \put(23,24){\scriptsize $y$}
   \put(-5,22){\scriptsize $v_1$}
   \put(46,22){\scriptsize $v_2$}
   \put(21,2){\scriptsize $f_1$}
   \put(21,40){\scriptsize $f_2$}
   \put(71,19){\scriptsize $e^*$}
   \put(54,13){\scriptsize $\alphat_e=\alpha_{e^*}$}
   \put(98,13){\scriptsize $\betat_e=\beta_{e^*}$}
   \put(98,33){\scriptsize $\alphat'_e$}
   \put(65,33){\scriptsize $\beta_e'$}
   \put(84,24){\scriptsize $y$}
  \end{overpic}
\caption{Notation around a square face $y=\{e,e^*\}$, seen from the primal edge $e$ point of view (left), or from the dual edge $e^*$ point of view (right). Observe that we have $\alphat_e=\beta_e,\betat_e=\alpha_e'=\alpha_e+\rho$.}
\label{fig:square_face}
\end{figure}
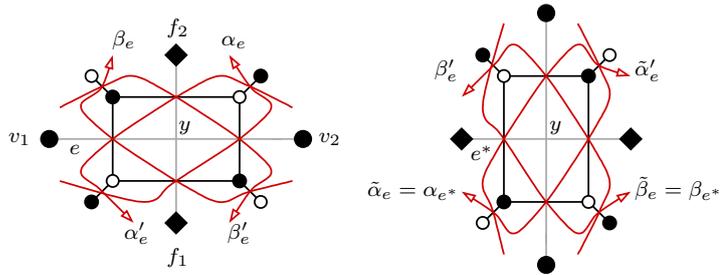

By Lemma~\ref{lem:invariance_weights}, face-weights are meromorphic functions of the parameters. Therefore we can choose all the non-prime angles $\alpha$ to have fixed values in $\CC$, and all prime angles to satisfy $\alpha'=\alpha+\rho$, where $\alpha$ can be any of the angles $\alpha,\alpha_i,\alphat_i$ above. Note that this choice \emph{fixes the values in $\CC$} of the angles $(\beta_i),(\tilde{\beta}_i)$, and we have 
\begin{equation}\label{equ:angle_relation}
\beta_{i+1}=\alpha_i'=\alpha_{i}+\rho,\quad \betat_{i+1}=\alphat_i'=\alphat_i+\rho.
\end{equation}
Using further Remark~\ref{rem:angle_definition}, we know that face-weights are also meromorphic functions of the value of the discrete Abel map at the vertex corresponding to the face. It is thus useful to compute values of the discrete Abel map $\mapd$ in $\TT(\tau)$. Let us choose a vertex $v_0$ of $G$ and set $\mapd(v_0)=0$. Note that this is consistent with Necessary Conditions I for the possible range of the angle map $\mapalpha$, and this is not restrictive since we have not yet fixed the value of the parameter $t$. 
\begin{lem}\label{lem:discrete_Abel_map}
Values of the discrete Abel map in $\TT(\tau)$ are given by:
{\small 
\begin{align}\label{eq:mapd:mod:2rho}
\mapd(x)=
\begin{cases}
0&\text{ if $x$ is a primal vertex $v$},\\
\rho &\text{ if $x$ is a dual vertex $f$},\\
(\beta_e-\alpha_e)+\rho & \text{ if $x$ is a square face $y=\{e,e^*\}$ of $\GQ$, with the notation of Figure~\ref{fig:square_face}}.
\end{cases}
\end{align}}
\end{lem}
\begin{proof}
Using the notation of Figure~\ref{fig:square_face} (left), we have:
\begin{align*}
\mapd(v_2)&=\mapd(v_1)+\alpha_e-\alpha'_e+\beta_e-\beta'_e=\mapd(v_1)-2\rho \ [\Lambda]=\mapd(v_1)\ [\Lambda]\\
\mapd(f_2)&=\mapd(f_1)+\beta_e-\alpha_e-(\beta'_e-\alpha'_e)=\mapd(f_1)\ [\Lambda]\\
\mapd(f_1)&=\mapd(v_1)-\alpha_e'+\alpha_e=\mapd(v_1)-\rho\ [\Lambda]\\
\mapd(y)&=\mapd(v_1)+\beta_e-\alpha'_e=\mapd(v_1)+(\beta_e-\alpha_e)-\rho\ [\Lambda].
\end{align*}
The proof is concluded by using that $\mapd(v_0)=0$. 
\end{proof}

We now turn to the computation of face-weights. 
\begin{lem}\label{lem:alternate_product_faces}
Consider Fock's dimer model on an infinite, minimal graph $\GQ$ with parameters $\tau,\mapalpha,\rho,t$ satisfying Necessary Conditions II. Then, the face-weights are explicitly given by, for every square face $y$, for every face corresponding to a dual vertex $f$, resp. primal vertex $v$:
\begin{align}
\W_{\mapalpha,\rho,t}(y)
&=(-1)^{j + j \ell}\frac{\theta_{1,1}^2(\beta_e-\alpha_e)}{
      \theta_{1+\ell,1+j}^2(\beta_e-\alpha_e)}
      \frac{\theta_{0,0}^2(t)}{\theta_{\ell,j}^2(t)}
=(-1)^{j + j \ell}\frac{\theta_{1+\ell,1+j}^2(\betat_e-\alphat_e)}{
      \theta_{1,1}^2(\betat_e-\alphat_e)}
      \frac{\theta_{0,0}^2(t)}{\theta_{\ell,j}^2(t)}\label{eq:square:face-weight}\\
\W_{\mapalpha, \rho, t}(f)&=(-1)^{|f|}\prod_{i=1}^{|f|}\frac{\theta_{1+\ell,1+j}(0)}{\theta_{1+\ell,1+j}(\alpha_{i-1}-\alpha_i)}\frac{\theta_{0,0}(t+\alpha_{i-1}-\alpha_i)}{\theta_{0,0}(t)}\label{eq:dual:face-weight}\\
\W_{\mapalpha,\rho, t}(v)
&=(-1)^{|v|}\prod_{i=1}^{|v|}
\frac{\theta_{1+\ell,1+j}(\alphat_{i}-\alphat_{i-1})}{\theta_{1+\ell,1+j}(0)}\frac{\theta_{\ell,j}(t)}{\theta_{\ell,j}(t+\alphat_i-\alphat_{i-1})}\label{eq:primal:face-weight}.
\end{align}
In particular the model is non-degenerate if $t \notin \{\frac{1}{2}+\frac{\tau}{2},\frac{1}{2}+\frac{\tau}{2}+\rho\}\ [\Lambda]$. 
\end{lem}

\begin{proof}
The fact that the model is non degenerate if $t\notin \{\frac{1}{2}+\frac{\tau}{2},\frac{1}{2}+\frac{\tau}{2}+\rho\}\ [\Lambda]$ comes from looking at the face-weight at a square face, recalling that the zeros of $\theta_{(0,0)}(\,\cdot\,)$ are located at $\bigl(\frac{1}{2}+\frac{\tau}{2}\bigr)\ [\Lambda]$, using Identity~\eqref{eq:theta:4}, and observing that if $t\in \{\frac{1}{2}+\frac{\tau}{2},\frac{1}{2}+\frac{\tau}{2}+\rho\}\ [\Lambda]$, the face-weight is either equal to $0$ or to infinity, which gives a degenerate model. 

In the three following computations, we use Equation~\eqref{equ:face_weight_bip} to compute the face-weight with the relevant choice of labeling of the angles, as given in Figures~\ref{fig:primal_dual} and~\ref{fig:square_face}. We choose the value of the angle map $\mapd$ at the center vertex according to Lemma~\ref{lem:discrete_Abel_map}, with a lift that suits our computations, and we use Equation~\eqref{equ:angle_relation}. 
We also need the following two identities, which can readily obtained from Identities~\eqref{eq:theta:4} and~\eqref{eq:theta:5}. For all $x,y\in\CC$, we have
\begin{equation}\label{equ:theta_ratios}
\frac{\theta_{m,n}(x+\rho)}{\theta_{m,n}(y+\rho)}=
e^{-i\pi\ell(x-y)}\frac{\theta_{m+\ell,n+j}(x)}{\theta_{m+\ell,n+j}(y)},\quad 
\frac{\theta_{m,n}(x-\rho)}{\theta_{m,n}(y-\rho)}=
e^{i\pi\ell(x-y)}\frac{\theta_{m+\ell,n+j}(x)}{\theta_{m+\ell,n+j}(y)}.
\end{equation}

\emph{Square face $y=\{e,e^*\}$.} We choose $\mapd(y)$ to be equal to $(\beta_e-\alpha_e)+\rho$, then we have
\begin{align*}
\W_{\mapalpha, \rho,t}(y)&=\frac{\theta_{1,1}(\beta_e-\alpha_e)\theta_{1,1}(\beta'_e-\alpha'_e)}{
\theta_{1,1}(\alpha'_e-\beta_e)\theta_{1,1}(\alpha_e-\beta'_e)}
\frac{\theta_{0,0}(t+\mapd(y)+\alpha'_e-\beta_e)\theta_{0,0}(t+\mapd(y)+\alpha_e-\beta'_e)}{\theta_{0,0}(t+\mapd(y)+\alpha_e-\beta_e)\theta_{0,0}(t+\mapd(y)+\alpha'_e-\beta'_e)}\\
&=\frac{\theta_{1,1}^2(\beta_e-\alpha_e)}{
\theta_{1,1}(\alpha_e-\beta_e+\rho)\theta_{1,1}(\alpha_e-\beta_e-\rho)}
\frac{\theta_{0,0}(t+2\rho)\theta_{0,0}(t)}{\theta_{0,0}^2(t+\rho)}.
\end{align*}
Using Identities~\eqref{eq:theta:3} and~\eqref{eq:theta:5}, we obtain 
\begin{equation*}
    \begin{aligned}
  \W_{\mapalpha, \rho,t}(y) &= \frac{\theta_{1,1}^2(\beta_e-\alpha_e)}{
    q^{-\ell/2}i^{2(\ell + j \ell)}(-1)^{j\ell + \ell + j}\theta_{1+\ell,1+j}^2(\alpha_e-\beta_e)}
    \frac{\theta_{0,0}^2(t)q^{-\ell}e^{-2i\pi \ell t}}{q^{-\ell/2}e^{-2i\pi \ell t}i^{2j\ell}\theta_{\ell,j}^2(t)}\\
  &= (-1)^{j + j \ell}\frac{\theta_{1,1}^2(\beta_e-\alpha_e)}{
      \theta_{1+\ell,1+j}^2(\alpha_e-\beta_e)}
      \frac{\theta_{0,0}^2(t)}{\theta_{\ell,j}^2(t)},
  \end{aligned}
\end{equation*}
and the proof of the first identity is concluded by using that $ \theta_{1+\ell,1+j}(\,\cdot\,)$
is odd. Observing that $\alphat_e=\beta_e,\betat_e=\alpha_e+\rho$, we obtain 
\begin{align*}
 \W_{\mapalpha, \rho,t}(y)=(-1)^{j + j \ell}
\frac{\theta_{1,1}^2(\alphat_e-\betat_e+\rho)}{\theta_{1+\ell,1+j}^2(\alphat_e-\betat_e+\rho)}
\frac{\theta_{0,0}^2(t)}{\theta_{\ell,j}^2(t)}=(-1)^{j + j \ell}
\frac{\theta_{1+\ell,1+j}^2(\alphat_e-\betat_e)}{\theta_{1,1}^2(\alphat_e-\betat_e)}
\frac{\theta_{0,0}^2(t)}{\theta_{\ell,j}^2(t)},
\end{align*}
where in the last equality, we used Identity~\eqref{equ:theta_ratios}.

\emph{Face corresponding to a dual vertex $f$.} We choose $\mapd(f)=-\rho$, then we have
{\small 
\begin{align}
\W_{\mapalpha, \rho, t}(f)
&= \prod_{i=1}^{|f|} 
 \frac{\theta_{1,1}(\alpha_i-\beta_{i+1})}{\theta_{1,1}(\beta_i-\alpha_i)}
 \frac{\theta_{0,0}(t+\mapd(f)+\beta_i-\alpha_i)}{\theta_{0,0}(t+\mapd(f)+\beta_{i+1}-\alpha_i)}=
  \prod_{i=1}^{|f|} \frac{\theta_{1,1}(-\rho)}{\theta_{1,1}(\beta_i-\alpha_i)}\frac{\theta_{0,0}(t-\rho+\beta_i-\alpha_i)}{\theta_{0,0}(t)}\label{equ:face_weight_dual_bis}\\
&=\prod_{i=1}^{|f|}(-1)\frac{\theta_{1,1}(\rho)}{\theta_{1,1}(\alpha_{i-1}-\alpha_i+\rho)}\frac{\theta_{0,0}(t+\alpha_{i-1}-\alpha_i)}{\theta_{0,0}(t)}.\nonumber 
\end{align}}
Then using Identity~\eqref{equ:theta_ratios}, we obtain 
\begin{align*}
\W_{\mapalpha, \rho, t}(f)&=(-1)^{|f|}\prod_{i=1}^{|f|}e^{i\pi\ell(\alpha_{i-1}-\alpha_i)}\frac{\theta_{1+\ell,1+j}(0)}{\theta_{1+\ell,1+j}(\alpha_{i-1}-\alpha_i)}\frac{\theta_{0,0}(t+\alpha_{i-1}-\alpha_i)}{\theta_{0,0}(t)}. 
\end{align*}
The proof is concluded by using that, since the angles $(\alpha_i)$ are fixed, 
$\prod_{i=1}^{|f|}e^{-i\pi\ell(\alpha_{i-1}-\alpha_i)}=~1$. 

\emph{Face corresponding to a primal vertex $v$.} We choose $\mapd(v)=0$, and obtain
{\small 
\begin{align}
\W_{\mapalpha, \rho, t}(v) 
&=\prod_{i=1}^{|v|} \frac{\theta_{1,1}(\betat_i-\alphat_i)}{\theta_{1,1}(\alphat_i-\betat_{i+1})}\frac{\theta_{0,0}(t+\mapd(v)+\alphat_i-\betat_{i+1})}{\theta_{0,0}(t+\mapd(v)+\alphat_i-\betat_i)}
=\prod_{i=1}^{|v|} \frac{\theta_{1,1}(\betat_i-\alphat_i)}{\theta_{1,1}(-\rho)}\frac{\theta_{0,0}(t-\rho)}{\theta_{0,0}(t+\alphat_i-\betat_i)}\label{equ:face_weight_primal_bis}\\
&=\prod_{i=1}^{|v|}(-1) \frac{\theta_{1,1}(\alphat_{i-1}-\alphat_i+\rho)}{\theta_{1,1}(\rho)}\frac{\theta_{0,0}(t-\rho)}{\theta_{0,0}(t+\alphat_i-\alphat_{i-1}-\rho)}.\nonumber 
\end{align}}
Using Identity~\eqref{equ:theta_ratios} for both factors we obtain: 
\begin{align*}
\W_{\mapalpha, \rho, t}(v)=(-1)^{|v|}\prod_{i=1}^{|v|}e^{-i\pi\ell(\alphat_{i-1}-\alphat_i)}e^{i\pi\ell(\alphat_i-\alphat_{i-1})}
\frac{\theta_{1+\ell,1+j}(\alphat_{i-1}-\alphat_i)}{\theta_{1+\ell,1+j}(0)}\frac{\theta_{\ell,j}(t)}{\theta_{\ell,j}(t+\alphat_i-\alphat_{i-1})}.
\end{align*}
The exponential terms cancel because the $(\alphat_i)$'s are fixed, and we have $\theta_{1+\ell,1+j}(\alphat_{i-1}-\alphat_i)=\theta_{1+\ell,1+j}(\alphat_{i}-\alphat_{i-1})$ because $(j,\ell)\neq (0,0)$ so that the theta function is even.  
\end{proof}

\paragraph{Duality.}  As in Section~\ref{sec:bipartite_representation_I}, suppose that the graph $\GQ$ is constructed from the dual graph $G^*$ instead of $G$, while keeping the bipartite coloring fixed. Then, the graph $\GRR$ is unchanged, so that train-tracks and the angle map are the same as well. The only thing that gets modified in the construction of Fock's dimer model is the base point of the discrete Abel map, which is now a primal vertex of the dual graph $G^*$, that is a dual vertex of the graph $G$. Denoting by $\W_{G,\mapalpha,\rho,t}(\,\cdot\,)$, resp. $\W_{G^*,\mapalpha,\rho,t}(\,\cdot\,)$, the face-weights of Fock's dimer model on $\GQ$ arising from $G$, resp. from $G^*$, we prove the following. 
\begin{lem}\label{lem:duality}
Face-weights of Fock's dimer model on the graph $\GQ$ arising from $G$ and from $G^*$ are related by the following:
\[
\W_{G,\mapalpha,\rho,t}(\,\cdot\,)=\W_{G^*,\mapalpha,\rho,t+\rho}(\,\cdot\,).
\] 
\end{lem}
\begin{proof}
Changing the base point of the discrete Abel map from a primal vertex of $G$ to a primal vertex of $G^*$ has the effect of translating all values by $\rho$; indeed the discrete Abel map is defined additively, and by Lemma~\ref{lem:discrete_Abel_map}, we know that $\mapd(v)-\mapd(f)=\rho$ in $\TT(\tau)$. The proof is concluded by recalling that the discrete Abel map is always summed with $t$ in the face-weights, see Equation~\eqref{equ:Focks_weight} and also Remark~\ref{rem:angle_definition}.
\end{proof}

\begin{rem}\label{rem:duality_in_face_weights}
Since when constructing the graph $\GQ$ from the dual graph $G^*$, the bipartite coloring is not exchanged, our choice of notation for angles around faces does not go through directly through duality; also the primal and dual faces have opposite bipartite coloring. As a consequence, one has to be careful when reading the above duality relation on the right-hand-sides of the identities of Lemma~\ref{lem:alternate_product_faces}, which rely on our choice of notation around faces. We nevertheless have the following relations which can be derived from understanding the duality relation with our choice of notation, or by a direct computation using also Identity~\eqref{equ:theta_ratios}. More precisely, define the following which correspond to the right-hand-sides of Lemma~\ref{lem:alternate_product_faces}: 
\begin{align*}
\Ws^{\, \square}_{(\alpha_e,\beta_e),t}&=(-1)^{j + j \ell}\frac{\theta_{1,1}^2(\beta_e-\alpha_e)}{
      \theta_{1+\ell,1+j}^2(\beta_e-\alpha_e)}
      \frac{\theta_{0,0}^2(t)}{\theta_{\ell,j}^2(t)}\\
\Ws_{(\alpha_1,\dots,\alpha_m),t}&=(-1)^{m}\prod_{i=1}^{m}\frac{\theta_{1+\ell,1+j}(0)}{\theta_{1+\ell,1+j}(\alpha_{i-1}-\alpha_i)}\frac{\theta_{0,0}(t+\alpha_{i-1}-\alpha_i)}{\theta_{0,0}(t)}\\
\Ws^*_{(\alpha_1,\dots,\alpha_m),t}&=(-1)^{m}\prod_{i=1}^{m}
\frac{\theta_{1+\ell,1+j}(\alpha_{i}-\alpha_{i-1})}{\theta_{1+\ell,1+j}(0)}\frac{\theta_{\ell,j}(t)}{\theta_{\ell,j}(t+\alpha_i-\alpha_{i-1})}, 
\end{align*}
then, we have the following duality relations: 
\begin{equation}\label{equ:duality_relations}
\begin{split}
&\Ws^{\, \square}_{(\alpha_e,\beta_e),t}=\frac{1}{\Ws^{\, \square}_{(\beta_e,\alpha_e+\rho),t+\rho}}=
\frac{1}{\Ws^{\, \square}_{(\alphat_e,\betat_e),t+\rho}},\\
&\Ws_{(\alpha_1,\dots,\alpha_m),t}=\frac{1}{\Ws^*_{(\alpha_m,\dots,\alpha_1),t+\rho}},\  \Ws^*_{(\alpha_1,\dots,\alpha_m),t}=\frac{1}{\Ws_{(\alpha_m,\dots,\alpha_1),t+\rho}}.
\end{split}
\end{equation}
\end{rem}

\subsection{Equality of face-weights at square faces: Necessary Conditions III}\label{sec:Necessary_III}

Consider an Ising model on an infinite, isoradial graph $G$ with real coupling constants $\eps\Js$, and the associated Ising-dimer model on the minimal graph $\GQ$. Consider Fock's dimer model on $\GQ$ with parameters $\tau,\mapalpha,\rho,t$ satisfying Necessary Conditions II. The two models are gauge equivalent if and only if face-weights are equal at all faces. As a consequence, we must have equality of face-weights for every square face $y=\{e,e^*\}$, implying in particular that the face-weight at $y$ must be negative. Using Equations~\eqref{equ:gauge_Ising_dimers} and~\eqref{eq:square:face-weight}, we thus have the following condition at every square face $y$: 
\begin{equation}\label{equ:condition_square_face}
-\sinh^2(2\Js_e)=\WI_{\eps\Js}(y)=\W_{\mapalpha,\rho,t}(y,\tau)=
(-1)^{j + j \ell}\frac{\theta_{1,1}^2(\beta_e-\alpha_e)}{\theta_{1+\ell,1+j}^2(\beta_e-\alpha_e)}
\frac{\theta_{0,0}^2(t)}{\theta_{\ell,j}^2(t)}<0,
\end{equation}
where we use the notation of Figure~\ref{fig:square_face}. In what follows, we identify parameters  $\tau,\mapalpha,\rho,t$ satisfying the negativity Condition~\eqref{equ:condition_square_face}.

\begin{prop}\label{prop:maximality_spectral_curve}
Consider Fock's dimer model on an infinite, minimal graph $\GQ$ with parameters $\tau,\mapalpha,\rho,t$ satisfying Necessary Conditions II. Then, if $\tau\in \frac{1}{2}+i\RR^+$, the negativity  Condition~\eqref{equ:condition_square_face} is never satisfied.
\end{prop}
\begin{proof}
The proof consists in showing that the face-weight $\W_{\mapalpha,\rho,t}(y,\tau)$ is always positive. 
Using Equation~\eqref{eq:theta:4}, we obtain
\begin{equation*}
\W_{\mapalpha,\rho,t}(y,\tau) 
= -\frac
{\theta_{1,1}^2(\beta-\alpha)\theta_{0,0}^2(t)}
{\theta_{1,0}^2(\beta-\alpha)\theta_{0,1}^2(t)}
= \frac
{\theta_{1,1}^2(\beta-\alpha)\theta_{1,0}^2(t+\tau/2)}
{\theta_{1,0}^2(\beta-\alpha)\theta_{1,1}^2(t+\tau/2)}
\end{equation*}
Using the algebraic identity on theta functions from Lemma~\ref{lem:ascending:Landen:theta} of Appendix~\ref{sec:App_A}, and denoting $\gamma = \beta - \alpha$, $s = t+\tau/2$ and $\tau_2 = 2\tau-1$, this can be rewritten as
\begin{equation*}
  \W_{\mapalpha,\rho,t}(y,\tau)
  =
  \frac
  {\theta_{1,1}^2(\gamma,\tau_2)\theta_{0,0}^2(\gamma,\tau_2)\theta_{1,0}^2(s,\tau_2)\theta_{0,1}^2(s,\tau_2)}
  {\theta_{1,0}^2(\gamma,\tau_2)\theta_{0,1}^2(\gamma,\tau_2)\theta_{1,1}^2(s,\tau_2)\theta_{0,0}^2(s,\tau_2)}
\end{equation*}
Note that since Necessary Conditions I are satisfied and $\tau \in 1/2 + i\RR^+$, $t \in \{\frac14,\frac34\} + i\RR~[\Lambda]$ so $s = t+\tau/2 \in \{0,\frac12\} + i\RR~[\Lambda]$.
Even though $\tau_2 \in i\RR^+$, the sign of $\W_{\mapalpha,\rho,t}(y,\tau)$ is not yet obvious from this equation since $\gamma, s \in \{0,\frac12\} +i\RR ~[\Lambda]$, and we only know that theta functions are real on $\RR [\Lambda]$ when the modulus is imaginary, so we need to use modular transformations. 
If we let $j_1, j_2 \in \{0,1\}$ be such that $\gamma + j_1/2, s+j_2/2 \in i\RR [\Lambda]$, by Equation~\eqref{eq:theta:4}
{\small 
\begin{equation*}
  \W_{\mapalpha,\rho,t}(y,\tau)
  =
  \frac
  {\theta_{1,1+j_1}^2(\gamma + j_1/2,\tau_2)\theta_{0,j_1}^2(\gamma + j_1/2,\tau_2)\theta_{1,j_2}^2(s+j_2/2,\tau_2)\theta_{0,1+j_2}^2(s+j_2/2,\tau_2)}
  {\theta_{1,j_1}^2(\gamma + j_1/2,\tau_2)\theta_{0,1+j_1}^2(\gamma + j_1/2,\tau_2)\theta_{1,1+j_2}^2(s+j_2/2,\tau_2)\theta_{0,j_2}^2(s+j_2/2,\tau_2)}.
\end{equation*}
}
Using Equation~\eqref{eq:tau_1} to change from imaginary arguments to real arguments modulo $[\Lambda]$, we see that this is positive real: all the constant factors from Equation~\eqref{eq:tau_1} cancel and there remains only a product of squared theta functions of real arguments with pure imaginary modulus.
\end{proof}

\begin{cor}\label{cor:spectral_curve_max}
Consider an Ising-dimer model on an infinite, periodic, minimal graph $\GQ$, and let $\C$ be its spectral curve. Suppose that $\C$ has genus 1 and satisfies assumptions $(\dagger)$. Then, the curve $\C$ is maximal. 
\end{cor}
\begin{proof}
Consider Fock's gauge equivalent dimer model on $\GQ$~\cite{Fock}, with parameters $\tau,\mapalpha,\rho,t$. Then, these parameters must satisfy Necessary Conditions II. Moreover, face-weights must be equal at all square faces which, by Proposition~\ref{prop:maximality_spectral_curve} implies that $\tau\neq \frac{1}{2}+i\RR^+$. As a consequence, we have $\tau\in i\RR^+$ which, by Point 1. of Theorem~\ref{thm:real_spectral_curve}, implies maximality of the spectral curve $\C$.
\end{proof}

Suppose now that $\tau\in i\RR^+$. Then, by Necessary Conditions I, the parameter $t$ belongs to $\RR+\bigl\{0,\frac{\tau}{2}\bigr\}~[\Lambda]$ and, for every square face $y=\{e,e^*\}$, the parameters $\alpha_e,\beta_e$ belong to $\RR+\bigl\{0,\frac{\tau}{2}\bigr\}~[\Lambda]$, so that we have: 
\begin{equation}\label{equ:square_condition}
t=\Re(t)+\frac{\ell_t}{2}\tau,\quad \beta_e-\alpha_e=\Re(\beta_e-\alpha_e)+\frac{\ell_e}{2}\tau,
\end{equation}
for some $\ell_t\in\{0,1\}$, and $(\ell_e)_{\scriptsize{\{y=\{e,e^*\}:\text{ square face}\}}}$ taking values in $\{0,1\}$. By Necessary Condition~II, $\rho=\frac{j}{2}+\frac{\ell}{2}\tau$, for some $(j,\ell)\neq (0,0)$.

\begin{prop}\label{prop:first_parameter_condition}
Suppose that $\tau\in i\RR^+$, and consider Fock's dimer model on an infinite, minimal graph $\GQ$ satisfying Necessary Conditions II, with parameters $\tau,\mapalpha,\rho=\frac{j}{2}+\frac{\ell}{2}\tau,t=\Re(t)+\frac{\ell_t}{2}\tau\in\RR+\{0,\frac{\tau}{2}\}~[\Lambda]$. Then, 
\begin{enumerate}
\item if $(j,\ell)=(0,1)$ the negativity Condition~\eqref{equ:condition_square_face} is never satisfied, 
\item if $(j,\ell)\in\{(1,0),(1,1)\}$, the negativity Condition~\eqref{equ:condition_square_face} is satisfied if and only if, for every square face $y$, $\ell_e=(\ell_t+\ell)\ [2]$. 
Moreover, 
\begin{enumerate}
 \item when $(j,\ell)=(1,0)$ and $\ell_t=1$, the graph $G$ and the dual graph $G^*$ need to be bipartite, 
 \item when $(j,\ell)=(1,1)$ and $\ell_t=0$, the dual graph $G^*$ needs to be bipartite; and if $\ell_t=1$, the graph $G$ needs to be bipartite.
\end{enumerate}
\end{enumerate}
\end{prop}
\begin{proof}
The proof consists in analyzing the condition $\W_{\mapalpha,\rho,t}(y)<0$ for every square face $y$. From Identity~\eqref{eq:theta:4} we obtain, for every $z\in\CC$, $s\in\{0,1\}$,
\[
\theta_{m,n}\Bigl(z+\frac{s}{2}\tau\Bigr)=(q^{-\frac{1}{4}}e^{-i\pi z})^{s}(-i)^{ns}\theta_{m+s,n}(z)\quad \Rightarrow\quad 
\frac{\theta_{m,n}^2\bigl(z+\frac{s}{2}\tau\bigr)}{\theta_{m+\ell,n+j}^2(z+\frac{s}{2}\tau)}=(-1)^{js}\frac{\theta_{m+s,n}^2(z)}{\theta_{m+\ell+s,n+j}^2(z)}.
\]
Therefore, setting $\beta_e-\alpha_e=\Re(\beta_e-\alpha_e)+\frac{\ell_e}{2}\tau$, and $t=\Re(t)+\frac{\ell_t}{2}\tau$ in Equation~\eqref{eq:square:face-weight} and using Identity~\eqref{eq:theta:4}, we obtain, for every square face $y$,
\begin{align}
\W_{\mapalpha,\rho,t}(y)
&=
(-1)^{j+j\ell}(-1)^{j\ell_e}(-1)^{j\ell_t}
\frac{\theta_{1+\ell_e,1}^2(\Re(\beta_e-\alpha_e))}{\theta_{1+\ell+\ell_e,1+j}^2(\Re(\beta_e-\alpha_e))}
\frac{\theta_{\ell_t,0}^2(\Re(t))}{\theta_{\ell+\ell_t,j}^2(\Re(t))}\nonumber\\
&=
(-1)^{j(\ell+1+\ell_e+\ell_t)}
\frac{\theta_{1+\ell_e,1}^2(\Re(\beta_e-\alpha_e))}{\theta_{1+\ell+\ell_e,1+j}^2(\Re(\beta_e-\alpha_e))}
\frac{\theta_{\ell_t,0}^2(\Re(t))}{\theta_{\ell+\ell_t,j}^2(\Re(t))}.\label{equ:weight_square_face_new}
\end{align}
Since $\tau\in i\RR^+$, theta functions are real for real values of the parameters, so that the terms involving theta functions are positive, and we are left with considering the sign in front. 

When $(j,\ell)=(0,1)$, the sign is always positive; when $(j,\ell)=(1,0)$, it is negative if and only if $\ell_e=\ell_t\ [2]$; when $(j,\ell)=(1,1)$, it is negative if and only if $\ell_e=(\ell_t+1)\ [2]$, thus proving the first part of Proposition~\ref{prop:first_parameter_condition}.

Now, consider a face of $\GQ$ corresponding to a dual vertex $f$ of $G$. Using the notation of Figure~\ref{fig:primal_dual} (left) we have, for every $i\in\{1,\dots,|f|\}$, $\beta_i-\alpha_i=\Re(\beta_i-\alpha_i)+\frac{\ell_i}{2}\tau$, and using the first part of the statement: 
\begin{align}
\sum_{i=1}^{|f|}(\beta_i-\alpha_i)&= \sum_{i=1}^{|f|}(\Re(\beta_i-\alpha_i))+\frac{\tau}{2}\sum_{i=1}^{|f|}\ell_i 
=\sum_{i=1}^{|f|}(\Re(\beta_i-\alpha_i))+ \frac{\tau}{2}\bigl( |f|(\ell_t+\ell)\ [2]\bigr)\nonumber\\
&=\sum_{i=1}^{|f|}(\Re(\beta_i-\alpha_i))+ \tau\Bigl( \frac{|f|}{2}(\ell_t+\ell)\ [1]\Bigr).\label{equ:cond_1}
\end{align}
Moreover, by Equation~\eqref{equ:angle_relation}, we have 
\begin{align}\label{equ:cond_2}
\sum_{i=1}^{|f|}(\beta_i-\alpha_i)&=\sum_{i=1}^{|f|}(\alpha_{i-1}+\rho-\alpha_i)=
|f|\rho=\frac{|f|}{2}j+\tau\frac{|f|}{2}\ell.
\end{align}
Right-hand-sides of~\eqref{equ:cond_1} and~\eqref{equ:cond_2} need to be equal in $\TT(\tau)$, implying that their imaginary parts need to be equal in $\TT(\tau)$. Therefore, we must have 
$\frac{|f|}{2}(\ell_t+\ell)\,[1]=\frac{|f|}{2}\ell \ [1]$, \emph{i.e.}, 
$\frac{|f|}{2}\ell_t \ [1]=0$. 
This is true if $\ell_t=0$, but if $\ell_t=1$, this is true if and only 
if $|f|$ is even. Since the argument holds for every dual vertex $f$, we must have that all dual vertices have even degree, which is equivalent to saying that the graph $G$ is bipartite. 

Consider a face of $\GQ$ corresponding to a vertex $v$ of $G$. Using the notation of Figure~\ref{fig:primal_dual} (right), for every $i\in\{1,\dots,|v|\}$, we have the following relation between the angles using the primal or dual perspective: $\betat_i-\alphat_i=\alpha_i-\beta_i+\rho$. Using the first part of the statement, we thus have:
\begin{align*}
\sum_{i=1}^{|v|}(\betat_i-\alphat_i)&=\sum_{i=1}^{|v|}(\alpha_i-\beta_i +\rho)=
\sum_{i=1}^{|v|}(\Re(\alpha_i-\beta_i))+\frac{\tau}{2}(|v|(\ell_t+\ell)\ [2])+|v|\rho\\
&=\sum_{i=1}^{|v|}(\Re(\alpha_i-\beta_i))+\frac{|v|}{2}j+\tau\Bigl( \frac{|v|}{2}(\ell_t+\ell)\ [1]+\frac{|v|}{2}\ell \Bigr).
\end{align*}
Moreover, by Equation~\eqref{equ:angle_relation}, we have
\begin{align*}
\sum_{i=1}^{|v|}(\betat_i-\alphat_i)&=\sum_{i=1}^{|v|}(\alphat_{i-1}+\rho-\alphat_i)=
|v|\rho=\frac{|v|}{2} j+\tau\frac{|v|}{2}\ell.
\end{align*}
Again, we need the imaginary parts of both right-hand-sides to be equal in $\TT(\tau)$, so that we must have that $\frac{|v|}{2}(\ell_t+\ell)=0\ [1]$. When $\ell_t=0$ and $\ell=0$, or $\ell_t=1$ and $\ell=1$, this is always true. When $\ell_t=1$ and $\ell=0$, or $\ell_t=0$ and $\ell=1$, this is true if and only if $|v|$ is even. Since this holds at every primal vertex $v$, we deduce that the dual graph $G^*$ is bipartite. Gathering these conditions proves the second part of the statement.
\end{proof}

Proposition~\ref{prop:first_parameter_condition} identifies the set of parameters $\mapalpha,\rho,t$ of Fock's dimer models satisfying the negativity Condition~\eqref{equ:square_condition}, but the condition on the angle map is given locally, that is for every square face $y=\{e,e^*\}$. This raises the question of the existence and characterization of (global) angle maps $\mapalpha$ satisfying these local constraints. 

In order to handle this question we need to add additional information on the set of train-tracks $\vec{\T}$ when either $G$ or $G^*$ is bipartite. Recall from Section~\ref{subsec:train_tracks} that, in general, there is no natural way of assigning labels $\vec{T}$ and $\cev{T}$ to the two oriented train-tracks of $\vec{\T}$ corresponding to a train-track $T$ of $\T$ just by looking at $T$. The situation is nevertheless different when the graph $G$ or its dual graph $G^*$ is bipartite. Suppose first that $G^*$ is bipartite. Then the train-tracks $\T$ of $G^*$ can themselves be oriented so that white vertices of $G^*$ are on the left. The oriented train-tracks $\vec{\T}$ of $\GQ$ can be partitioned as $\vec{\T} =  \vec{\T}_{*,\circ} \sqcup \vec{\T}_{*,\bullet}$, where $\vec{\T}_{*,\circ}$, resp. $\vec{\T}_{*,\bullet}$, consists of the train-tracks turning counterclockwise around every white, resp. black vertex of $G^*$; moreover the set $\vec{\T}_{*,\circ}$ is in orientation preserving bijection with $\T$, see Figure~\ref{fig:dual_graph_bipartite}.

\begin{figure}[h]
  \centering
  \begin{overpic}[width=12cm]{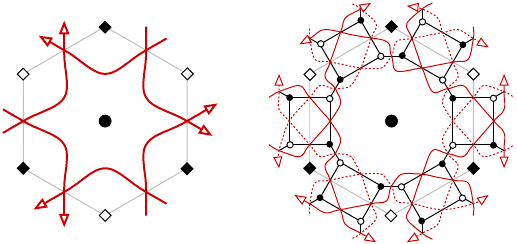} 
  \end{overpic}
\caption{Left: the oriented train-tracks of $\T$ when $G$ is the triangular lattice and $G^*$ is the hexagonal lattice, which is bipartite. Right: the oriented train-tracks of $\vec{\T}$ split into $\vec{\T}_{\circ,*}$ (full lines, in orientation preserving bijection with $\T$), and $\vec{\T}_{\bullet,*}$ (dotted lines).}  \label{fig:dual_graph_bipartite}  
\end{figure}

When the graph $G$ is bipartite, a similar construction can be done, with the word ``face'' replaced by ``vertex''. This yields a partition of the set of oriented train-tracks $\vec{\T}$ as $\vec{\T} =  \vec{\T}_\circ \sqcup \vec{\T}_\bullet$, where $\vec{\T}_\circ$, resp. 
$\vec{\T}_\bullet$, consists of the train-tracks turning counterclockwise around every white, resp. black vertex of $G$. For every train-track $T$ of $\T$, we then set the labeling $\vec{T}$, $\cev{T}$ as follows:
\begin{equation}\label{eq:bipartite:orientation}
 \text{If $G^*$ is bipartite, then } \vec{T} \in \vec{\T}_{*,\circ}, \cev{T} \in \vec{\T}_{*,\bullet} ; \quad 
  \text{if $G$ is bipartite, then } \vec{T} \in \vec{\T}_\circ, \cev{T} \in \vec{\T}_\bullet.
\end{equation}
Note that if both graphs $G$ and $G^*$ are bipartite, the above condition fixes the bipartite coloring of $G^*$ once the coloring of $G$ is set. Note also that only differences of angles occur in face-weights, so that we need to fix a base point for the angle map $\mapalpha$. If $G$ is not bipartite, we fix any $\vec{T}_0$, and set $\alpha_{\vec{T}_0}=0$ (observing that this is compatible with Necessary Conditions I). If $G$, resp. $G^*$, is bipartite, then $\vec{T}_0$ is fixed to belong to $\vec{\T}_\circ$, resp. $\vec{\T}_{*,\circ}$, and we also set $\alpha_{\vec{T}_0}=0$. 

\begin{defi}\label{defi:angle_map}
Consider an infinite, isoradial graph $G$, and the corresponding minimal graph $\GQ$. Suppose that, up to modular transformations, $\tau\in i\RR^+$, and consider parameters $\mapalpha,\rho=\frac{j}{2}+\frac{\ell}{2}\tau,t=\Re(t)+\frac{\ell_t}{2}\tau$ satisfying Necessary Conditions II, recalling that Necessary Conditions II implies that, for every train-track $T$ of $\T$, $\alpha_{\vec{T}}=\alpha_{\cev{T}}+\rho\ [\Lambda]$.
Let us define the following sets $\Acal_{\rho,\,\ell_t/2}$ for the angle map $\mapalpha$:
\begin{enumerate}
 \item Suppose that $(j,\ell)=(1,0),\ell_t=0$. For every train-track $T$ of $\T$, choose any labeling $\vec{T}$ and $\cev{T}$ for the two corresponding oriented train-tracks, and define:
\[
\Acal_{\rho,\,\ell_t/2}=\{\mapalpha: \ \forall\,T \in\T,\ \alpha_{\vec{T}}\in\RR~[\Lambda]\}.
\] 
Since $\rho=\frac{1}{2}$ is real we have, for every $T \in\T$, $\alpha_{\cev{T}}\in\RR~[\Lambda]$, so that the angle map $\mapalpha$ takes values in $\RR~[\Lambda]$.
\item Suppose that $(j,\ell)=(1,0),\ell_t=1$. By Proposition~\ref{prop:first_parameter_condition}, both the graph $G$ and the dual graph $G^*$ are bipartite. We choose the labeling of the train-tracks as in Equation~\eqref{eq:bipartite:orientation} (right), and define:
\begin{align*}
\Acal_{\rho,\,\ell_t/2}=\left\{\mapalpha:\ \right. &\text{ around every primal and dual face, } \alpha_{\vec{T}} \text{ alternatively belongs to }  \\
&\left. \RR~[\Lambda] \text{ and }\RR+\Bigl\{\frac{\tau}{2}\Bigr\}~[\Lambda]\right\}.
\end{align*}
This angle map can consistently be constructed, because both $G$ and $G^*$ are bipartite. Since $\rho=\frac{1}{2}$, $\alpha_{\vec{T}}$ and $\alpha_{\cev{T}}$ both belong to the same component of the torus $\TT(\tau)$. 
 \item Suppose that $(j,\ell)=(1,1)$. By Proposition~\ref{prop:first_parameter_condition}, if $\ell_t=0$ the dual graph $G^*$ is bipartite, and if $\ell_t=1$ the graph $G$ is bipartite. 
 We choose the labeling of the pairs of train-tracks as in Equation~\eqref{eq:bipartite:orientation} and define, 
 \[
\Acal_{\rho,\,\ell_t/2}=\{\mapalpha:\ \forall\,T\in\T,\ \alpha_{\vec{T}}\in \RR~[\Lambda]\}. 
 \]
where recall that, for every $T\in\T$, in the first case we have $\vec{T}\in\vec{\T}_{*,\circ}$, while in the second $\vec{T}\in\vec{\T}_{\circ}$. Since $\rho$ has imaginary part $\frac{\tau}{2}$ we have, for every $T \in\T$, $\alpha_{\cev{T}}\in\RR+\{\frac{\tau}{2}\}~[\Lambda]$. 
\end{enumerate}
Note that the set of angles $\Acal_{\rho,\,\ell_t/2}$ only depends on the imaginary part $\frac{\ell_t}{2}\tau$ of $t$, and not on its real part. 
Note also that the angle map $\mapalpha$ is generic since, for every train-track $T$ of $\T$, there is no additional constraint on the angle $\alpha_{\vec{T}}$ except from that of belonging to one of the real components of the torus $\TT(\tau)$.
\end{defi}

\begin{prop}\label{prop:global_condition_anglemap}
Under the assumptions of Definition~\ref{defi:angle_map}, for $(j,\ell)=(1,0)$ and $\ell_t\in\{0,1\}$, or $(j,\ell)=(1,1)$ and $\ell_t\in\{0,1\}$, the angle map $\mapalpha$ satisfies the local condition:
\begin{equation}\label{Cond:local}
\forall\, \text{square face $y=\{e,e^*\}$},\quad \beta_e-\alpha_e=\Re(\beta_e-\alpha_e)+\frac{\ell_e}{2}\tau, \ \text{ with $\ell_e=(\ell_t+\ell) \ [2]$},
\end{equation}
if and only if $\mapalpha$ belongs to $\Acal_{\rho,\,\ell_t/2}$.
\end{prop}
\begin{proof}\leavevmode

$\bullet$ Case $(j,\ell)=(1,0)$ and $\ell_t=0$. Suppose that $\mapalpha\in \Acal_{\rho,\,\ell_t/2}$, then 
since $\rho=\frac{1}{2}$ is real, the angle map $\mapalpha$ takes values in $\RR~[\Lambda]$.
As a consequence, for every square face $y=\{e,e^*\}$, $\beta_e-\alpha_e\in\RR~[\Lambda]$; it is thus of the form $\beta_e-\alpha_e=\Re(\beta_e-\alpha_e)+\frac{\ell_e}{2}\tau$, with $\ell_e=\ell_t+\ell\ [2]=0\ [2]$. Suppose now that $\mapalpha\notin \Acal_{\rho,\,\ell_t/2}$. Then, since the graph of train-tracks $\GR$ is connected, since $\mapalpha$ satisfies Necessary Conditions I, since $\rho$ is real, there exists two intersecting train-tracks $\vec{T}_1$, $\vec{T}_2$, such that $\alpha_{\vec{T}_1}\in\RR~[\Lambda]$, 
and $\alpha_{\vec{T}_2}\in\RR+\{\frac{\tau}{2}\}~[\Lambda]$. Since the train-tracks $\vec{T}_1$, $\vec{T}_2$ intersect, they must do so at a square face $y$, implying that we either have $\beta_e-\alpha_e$ or $\betat_e-\alphat_e$ that belongs to 
$\RR+ \{\frac{\tau}{2}\}~[\Lambda]$. But, by Necessary Conditions II, $\betat_e-\alphat_e=\alpha_e-\beta_e+\rho\ [\Lambda]=\alpha_e-\beta_e+\frac{1}{2}\ [\Lambda]$, so that both $\beta_e-\alpha_e$ and $\betat_e-\alphat_e$ must belong to 
$\RR+ \{\frac{\tau}{2}\}~[\Lambda]$, and in particular $\beta_e-\alpha_e$. This means that at the square face $y$, Condition~\eqref{Cond:local} is not satisfied, ending the proof of this case. The case $(i,j)=(1,0)$ and $\ell_t1$ is treated below. 

$\bullet$ Case $(j,\ell)=(1,1)$ and $\ell_t=0$. Note that because of the choice of labeling for pairs of train-tracks when the graph $G^*$ is bipartite, for every square $y=\{e,e^*\}$, we either have $\alpha_e\in \vec{\T}_{*,\circ}$ and $\beta_e\in\vec{\T}_{*,\bullet}$, or $\alpha_e\in \vec{\T}_{*,\bullet}$ and $\beta_e\in\vec{\T}_{*,\circ}$; we also have that $\alphat_e,\betat_e$ either both belong to $\vec{\T}_{*,\circ}$, or they both belong to $\vec{\T}_{*,\bullet}$.
Suppose that $\mapalpha\in \Acal_{\rho,\,\ell_t/2}$, then because of the labeling of pairs of train-tracks, we have that,
 for every square face $y$, $\beta_e-\alpha_e\in\RR+ \{\frac{\tau}{2}\}~[\Lambda]$; 
 it is thus of the form $\beta_e-\alpha_e=\Re(\beta-\alpha)+\frac{\ell_e}{2}\tau$, with $\ell_e=\ell_t+\ell\ [2]=1\ [2]$. 
 Suppose now that $\mapalpha\notin \Acal_{\rho,\,\ell_t/2}$. Then, since the graph of train-tracks $\GRR$ is connected, since $\mapalpha$ satisfies Necessary Conditions I, there exists two train-tracks $\vec{T}_1,\vec{T}_2\in \vec{\T}_\circ$ such that 
 $\alpha_{\vec{T}_1}\in\RR~[\Lambda]$, and $\alpha_{\vec{T}_2}\in\RR+ \{\frac{\tau}{2}\}~[\Lambda]$, intersecting at a square face $y$. Because of our choice of labeling for pairs of train-tracks, they must correspond to and edge crossed by train-tracks with angles $\alphat_e,\betat_e$. 
 We thus have $\betat_e-\alphat_e\in\RR+\{\frac{\tau}{2}\}~[\Lambda]$ which, with the condition $\betat_e-\alphat_e=\alpha_e-\beta_e+\rho\ [\Lambda]=\alpha_e-\beta_e+\frac{1}{2}+\frac{1}{2}\tau\ [\Lambda]$,
  implies that $\beta_e-\alpha_e\in\RR/\ZZ$. This means that at the square face $y$, Condition~\eqref{Cond:local} is not satisfied, ending the proof of this case. 

$\bullet$ Case $(j,\ell)=(1,1)$ and $\ell_t=1$. The proof is similar to the case $(j,\ell)=(1,1)$ and $\ell_t=0$, with the role of $\alpha_e,\beta_e$ taken by $\alphat_e,\betat_e$. Because of the choice of labeling for pairs of train-tracks when $G$ is bipartite, for every square $y$, we have that either $\alpha_e,\beta_e$ both belong to $\vec{\T}_{\circ}$, or they both belong to $\vec{\T}_{\bullet}$; we also have that 
$\alphat_e\in \vec{\T}_{\circ}$ and $\betat_e\in\vec{\T}_{\bullet}$, or $\alphat_e\in \vec{\T}_{\bullet}$ and $\betat_e\in\vec{\T}_{\circ}$. The remainder of the argument is so similar that we do not repeat it. 

$\bullet$ Case $(j,\ell)=(1,0)$ and $\ell_t=1$. Suppose that $\mapalpha\in\Acal_{\rho,\,\ell_t/2}$. Then, by definition of the angle map, and since $\alpha_{\vec{T}}$ and $\alpha_{\cev{T}}$ belong to the same real component we have that, for every square face $y=\{e,e^*\}$,
 $\beta_e-\alpha_e\in\RR+\{\frac{\tau}{2}\}~[\Lambda]$; it is thus of the form $\beta_e-\alpha_e=\R(\beta_e-\alpha_e)+\frac{\ell_e}{2}\tau$, with $\ell_e=\ell_t+\ell\ [2]=1\ [2]$. The remainder of the argument is a combination of the above, and we do not repeat it here. 
\end{proof}

Let us now wrap up. Consider an Ising model on an infinite, isoradial graph $G$ with real coupling constants $\eps\Js$, and the corresponding Ising-dimer model on $\GQ$. Consider Fock's dimer model on $\GQ$ with parameters $\tau,\mapalpha,\rho=\frac{j}{2}+\frac{\ell}{2}\tau,t=\Re(t)+\frac{\ell_t}{2}\tau$ satisfying Necessary Conditions II. Then, if the two models are gauge equivalent, by Proposition~\ref{prop:first_parameter_condition}, the modular parameter $\tau$ must belong to $i\RR^+$, the case $(j,\ell)=(0,1)$ cannot occur, and angles of the angle map $\mapalpha$ need to satisfy a local condition at every square face. The local condition on the angle map $\mapalpha$ is proved to be equivalent to a global condition in Proposition~\ref{prop:global_condition_anglemap}, it also requires some additional bipartitedness assumptions on the underlying graphs. This motivates the following. 

\begin{defi}\label{defi:necessary_III}
We say that the parameters $\tau,\mapalpha,\rho=\frac{j}{2}+\frac{\ell}{2}\tau,t=\Re(t)+\frac{\ell_t}{2}\tau$ of Fock's dimer model on the graph $\GQ$ satisfy \emph{Necessary Conditions III} if they satisfy Necessary Conditions II and if, up to modular transformations of $\tau$,
\begin{enumerate}
\item[$\bullet$] $\tau\in i\RR^+$; $(j,\ell)\neq (0,1)$; and the angle map $\mapalpha$ belongs to the set $\Acal_{\rho,\,\ell_t/2}$ of Definition~\ref{defi:angle_map};
\end{enumerate}
Furthermore, 
\begin{enumerate}
\item[$\bullet$] when $\rho=\frac{1}{2}$ and $\ell_t=1$, the graph $G$ and the dual graph $G^*$ are bipartite; 
\item[$\bullet$] when $\rho=\frac{1}{2}+\frac{\tau}{2}$ and $\ell_t=0$, the dual graph $G^*$ is bipartite, and when $\ell_t=1$, the graph $G$ is bipartite. 
\end{enumerate}
\end{defi}

When the parameters $\tau,\mapalpha,\rho,t$ satisfy Necessary Conditions III, using Equation~\eqref{equ:weight_square_face_new} and Proposition~\ref{prop:first_parameter_condition} (in particular the fact that $j=1$), we can rewrite Condition~\eqref{equ:condition_square_face} as:
{\small 
\begin{equation}\label{equ:cond_square_face_bis}
-\sinh^2(2\Js_e)=(-1)^{j + j \ell}\frac{\theta_{1,1}^2(\beta_e-\alpha_e)}{
      \theta_{1+\ell,1+j}^2(\beta_e-\alpha_e)}
      \frac{\theta_{0,0}^2(t)}{\theta_{\ell,j}^2(t)}=
-\frac{\theta_{1+\ell_t+\ell,1}^2(\Re(\beta_e-\alpha_e))}{\theta_{1+\ell_t,0}^2(\Re(\beta_e-\alpha_e))}
\frac{\theta_{\ell_t,0}^2(\Re(t))}{\theta_{\ell+\ell_t,1}^2(\Re(t))}<0.
\end{equation}}
This allows us to determine the absolute value coupling constants $\Js$ for every square face $y=\{e,e^*\}$; let us denote by $\Js_e=\Js_{e,(\mapalpha,\rho,t)}$ the solution to Equation~\eqref{equ:cond_square_face_bis}. 

\begin{rem}\label{rem:duality_again}
On the dual graph $G^*$, define the absolute value $\Js^*_{e^*,(\mapalpha,\rho,t)}$ to be the solution of 
\[
\sinh^{2}(2\Js^*_{e^*,(\mapalpha,\rho,t)})=-\frac{1}{\W_{G,\mapalpha,\rho,t}}(y).
\] 
Let $\eps^*$ denote any sign function on the edges of the dual graph.
Then using subsequently: the definition of Ising-dimer weights, the definition of $\Js^*_{e^*,(\mapalpha,\rho,t)}$, and Condition~\eqref{equ:condition_square_face}, we obtain
\[
\WI_{G^*,\eps^*\Js^*_{(\mapalpha,\rho,t)}}(y)=-\sinh^{-2}(2\Js^*_{e^*,(\mapalpha,\rho,t)})=\W_{G,\mapalpha,\rho,t}(y)=\WI_{G,\eps\Js_{(\mapalpha,\rho,t)}}(y),
\]
so that, recalling Definition~\ref{def:weak_duality}, the coupling constants $\eps^*\Js^*_{(\mapalpha,\rho,t)}$ are the weak dual of the coupling constant $\eps\Js_{(\mapalpha,\rho,t)}$. 
Note that the signs $\eps^*$ do not matter here since weak duality only involves the absolute values of the coupling constants.
Note also that by Lemma~\ref{lem:duality} we have:
\[
\WI_{G^*,\eps^*\Js^*_{(\mapalpha,\rho,t)}}(y)=\W_{G^*,\mapalpha,\rho,t+\rho}(y).
\]
As a consequence, we say that Fock's dimer model on $\GQ$ with parameters $\tau,\mapalpha,\rho,t$ and parameters $\tau,\mapalpha,\rho,t+\rho$ are \emph{weakly dual}.
\end{rem}

\subsection{Equality of squared face-weights: Necessary Conditions IV}\label{sec:Necessary_IV}

Consider an Ising model on an infinite, periodic, isoradial graph $G$ with periodic, real coupling constants $\eps\Js$, and the corresponding Ising-dimer model on the minimal graph $\GQ$. Consider Fock's dimer model with parameters $\tau,\mapalpha,\rho,t$. Suppose that the two models are gauge equivalent, implying that the parameters $\tau,\mapalpha,\rho,t$ satisfy Necessary Conditions III, and recall that Necessary Conditions III allows to define the absolute value coupling constants $\Js_{(\mapalpha,\rho,t)}$, see Equation~\eqref{equ:cond_square_face_bis}. If the two models are gauge equivalent, we must also have equality of the squared face-weights at all faces of $\GQ$ corresponding to dual vertices, resp. primal vertices, namely:
\begin{equation}\label{equ:condition_other_faces}
\forall\,f\in V^*,\quad 
( \WI_{\eps\Js_{(\mapalpha,\rho,t)}} )^2 (f)=( \W_{\mapalpha,\rho,t} )^2(f),\forall\,v\in V,\quad
( \WI_{\eps\Js_{(\mapalpha,\rho,t)}} )^2 (v)=( \W_{\mapalpha,\rho,t} )^2(v).
\end{equation}

\subsubsection{Necessary Conditions IV}

In the following, we identify conditions on the parameters $\rho$ and $t$ so that  Equation~\eqref{equ:condition_other_faces} is satisfied whenever Equation~\eqref{equ:cond_square_face_bis} is. 

\begin{prop}\label{prop:t}
Consider Fock's dimer model on an infinite, minimal graph $\GQ$ with parameters $\tau,\mapalpha,\rho,t$ satisfying Necessary Conditions III, \emph{i.e.}, $\rho=\frac{j}{2}+\frac{\ell}{2}\tau$, $t=\Re(t)+\frac{\ell_t}{2}\tau$, with $(j,\ell)\in\{(1,0),(1,1)\}$, and $\ell_t\in\{0,1\}$.
For every $\mapalpha\in\Acal_{\rho,\ell_t/2}$ consider the absolute value coupling constants $\Js_{(\mapalpha,\rho,t)}$ 
of Equation~\eqref{equ:cond_square_face_bis}. Then,
Condition~\eqref{equ:condition_other_faces} is satisfied for every generic angle map $\mapalpha\in\Acal_{\rho,\,\ell_t/2}$ if and only if, 
\begin{enumerate}
 \item either $\rho=\frac{1}{2}$ and $t\in\{0,\frac{1}{2}\}$, 
 \item or $\rho=\frac{1}{2}+\frac{\tau}{2}$ and $t\in\{\frac{1}{2},\frac{\tau}{2}\}$. 
\end{enumerate}
In particular, if the model is non degenerate, Condition~\eqref{equ:condition_other_faces} is never satisfied when $\rho=\frac{1}{2}$ and $\ell_t=1$. 
\end{prop}
\begin{proof}
Fix $\rho,t$ as above. Consider $\mapalpha\in\Acal_{\rho,\,\ell_t/2}$ and the absolute value coupling constants $\Js_{(\mapalpha,\rho,t)}$ given by Equation~\eqref{equ:cond_square_face_bis}. Condition~\eqref{equ:cond_square_face_bis} is satisfied at every square face, and we have,
\begin{align*}
\frac{1}{\tanh^{2}(2\Js_{e,(\mapalpha,\rho,t)})} &=1+\frac{1}{\sinh^{2}(2\Js_{e,(\mapalpha,\rho,t)})}
= 1- (-1)^{j+j\ell}\frac{
\theta_{1+\ell,1+j}^2(\beta_e-\alpha_e)\theta_{\ell,j}^2(t)}
{\theta_{1,1}^2(\beta_e-\alpha_e)\theta_{0,0}^2(t)}\\
    &= \frac{\theta_{1,1}^2(\beta_e-\alpha_e)\theta_{0,0}^2(t)- (-1)^{j+j\ell}
      \theta_{1+\ell,1+j}^2(\beta_e-\alpha_e)\theta_{\ell,j}^2(t)}
      {\theta_{1,1}^2(\beta_e-\alpha_e)\theta_{0,0}^2(t)}\\
    &= (-1)^{\ell + j \ell}\frac{\theta_{\ell,j}(\beta_e-\alpha_e+t)\theta_{\ell,j}(\beta_e-\alpha_e-t)\theta_{1+\ell,1+j}^2(0)}{\theta_{1,1}^2(\beta_e-\alpha_e)\theta_{0,0}^2(t)}, 
\end{align*}
where in the last line, we used Identity~\eqref{eq:addition:formula} at $(m,n)=(1+\ell,1+j)$. 


Let $f$ be a dual vertex of $G$ with $|f|$ boundary faces, and let $(e_i)_{i=1,\dots,|f|}$ denote its boundary edges, ordered counterclockwise. Then, using the notation of Figure~\ref{fig:primal_dual} (left), Equation~\eqref{equ:gauge_Ising_dimers} at the face corresponding to $f$ and the above identity for $\tanh^{-2}(2\Js_{e_i,(\mapalpha,\rho,t)})$, we have 
\begin{align*}\label{eq:face-weight:dual:I}
(\WI_{\eps\Js_{(\mapalpha,\rho,t)}})^2(f)&= \prod_{i=1}^{|f|}\frac{1}{\tanh^2(2\Js_{e_i,(\mapalpha,\rho,t)})}\\
&=\prod_{i=1}^{|f|}(-1)^{\ell+j\ell}\frac{\theta_{\ell,j}(\beta_i-\alpha_i+t)\theta_{\ell,j}(\beta_i-\alpha_i-t)\theta_{1+\ell,1+j}^2(0)}{\theta_{1,1}^2(\beta_i-\alpha_i)\theta_{0,0}^2(t)}\\
&=\prod_{i=1}^{|f|}\frac{\theta_{0,0}(\alpha_{i-1}-\alpha_i+t)\theta_{0,0}(\alpha_{i-1}-\alpha_i-t)\theta_{1+\ell,1+j}^2(0)}{\theta_{1+\ell,1+j}^2(\alpha_{i-1}-\alpha_i)\theta_{0,0}^2(t)},
\end{align*}
where in the last equality we used that, for every $i\in\{1,\dots,|f|\}$, $\beta_i-\alpha_i=\alpha_{i-1}-\alpha_i+\rho$, and Identity~\eqref{eq:theta:4}.

On the other hand, by Equation~\eqref{eq:dual:face-weight} of Lemma~\ref{lem:alternate_product_faces}, we have
\begin{equation*}
(\W_{\mapalpha, \rho, t})^2(f) 
=\prod_{i=1}^{|f|}\frac{\theta_{1+\ell,1+j}^2(0)}{\theta_{1+\ell,1+j}^2(\alpha_{i-1}-\alpha_i)}\frac{\theta_{0,0}^2(t+\alpha_{i-1}-\alpha_i)}{\theta_{0,0}^2(t)}.
\end{equation*}
Therefore we have, for every generic $\mapalpha\in\Acal_{\rho,\,\ell_t/2}$, for every face corresponding to a dual face $f$, $(\WI_{\eps\Js_{(\mapalpha,\rho,t)}} )^2 (f)=( \W_{\mapalpha,\rho,t} )^2(f)$ if and only if, for every generic $\mapalpha\in\Acal_{\rho,\,\ell_t/2}$, for every $f$,
\begin{equation*}\label{eq:CNS}
\prod_{i=1}^{|f|}\theta_{0,0}(\alpha_{i-1}-\alpha_i+t) = \prod_{i=1}^{|f|} \theta_{0,0}(\alpha_{i-1}-\alpha_i-t)\quad \Leftrightarrow\quad 
\prod_{i=1}^{|f|}\frac{\theta_{0,0}(\alpha_{i-1}-\alpha_i+t)}{\theta_{0,0}(\alpha_{i-1}-\alpha_i-t)}=1.
\end{equation*}
Note that because of our choice of labeling, see Figure~\ref{fig:primal_dual}, train-tracks associated to $\alpha_i$'s for distinct $i$'s do not belong to pairs of train-tracks around $f$. Moreover, since the graph $\GQ$ is minimal, they correspond to distinct train-tracks of $\T$, see~\cite[Lemma 8]{BCdT:immersions}.
By definition of $\Acal_{\rho,\,\ell_t/2}$, this means that the $\alpha_i$'s for distinct values of $i$ are generically independent and live on a real component of $\TT(\tau)$.
 The above condition thus implies that, for every $i\in\{1,\dots,|f|\}$, for every $\alpha_i$ living on a real component of $\TT(\tau)$,
\begin{equation}\label{equ:ratio_to_check}
\frac{\theta_{0,0}(\alpha_{i-1}-\alpha_i+t)}{\theta_{0,0}(\alpha_{i-1}-\alpha_i-t)}\frac{\theta_{0,0}(\alpha_{i}-\alpha_{i+1}+t)}{\theta_{0,0}(\alpha_{i}-\alpha_{i+1}-t)}=c,
\end{equation}
where the term $c$ of the right-hand-side is constant when seen as a function of $\alpha_i$. By Lemma~\ref{lem:meromorphic_theta_ratio}, the left-hand-side is a meromorphic function of $\alpha_i$ on $\TT(\tau)$ which needs to be equal to a constant on a real component of $\TT(\tau)$, hence it must be equal to a constant on $\TT(\tau)$. This is true if and only if the poles and the zeros cancel, and using that the $\alpha_i$'s for distinct $i$'s are generic, this is true if and only if 
\begin{align*}
\alpha_{i-1}+t=\alpha_{i-1}-t\ [\Lambda],\text{ and } \alpha_{i+1}+t=\alpha_{i+1}-t\ [\Lambda]\quad \Leftrightarrow\quad t=-t\ [\Lambda] \quad \Leftrightarrow\quad & 2t=0\ [\Lambda].
\end{align*}

When $\ell_t=0$, $t$ is real, so that we must have $t\in\{0,\frac{1}{2}\}$. Recall further that by Lemma~\ref{lem:alternate_product_faces}, if the model is non degenerate, $t\notin \{\frac{1}{2}+\frac{\tau}{2},\frac{1}{2}+\frac{\tau}{2}+\rho\}$. This implies that if $\rho=\frac{1}{2}$, then $t\in\{0,\frac{1}{2}\}$, and if $\rho=\frac{1}{2}+\frac{\tau}{2}$, then $t=\frac{1}{2}$. 

When $\ell_t=1$, $t$ has imaginary part equal to $\frac{\tau}{2}$, so that we must have $t\in\{\frac{\tau}{2},\frac{1}{2}+\frac{\tau}{2}\}$. When $\rho=\frac{1}{2}+\frac{\tau}{2}$, the non-degeneration condition implies furthermore that $t=\frac{\tau}{2}$. When $\rho=\frac{1}{2}$, the non-degeneration condition implies that $t\notin\{\frac{1}{2}+\frac{\tau}{2},\frac{\tau}{2}\}$, so that in this case there exists no $t$ such that the model is non-degenerate and satisfies Equation~\eqref{equ:ratio_to_check}. This yields the first part of the result. 

We thus assume that either $\rho=\frac{1}{2}$ and $t\in\{0,\frac{1}{2}\}$, or $\rho=\frac{1}{2}+\frac{\tau}{2}$ and $t\in\{\frac{1}{2},\frac{\tau}{2}\}$. 
We know that the ratio~\eqref{equ:ratio_to_check} is constant, but we still need to check that $\prod_{i=1}^{|f|}\frac{\theta_{0,0}(\alpha_{i-1}-\alpha_i+t)}{\theta_{0,0}(\alpha_{i-1}-\alpha_i-t)}=1$. 
The parameter $t$ is of the form $t=\frac{j_t}{2}+\frac{\ell_t}{2}\tau$ for some $j_t\in\{0,1\}$, so that $t=-t+j_t+\ell_t\tau$; note also that we have the condition $j_t\ell_t=0$. Using Identity~\eqref{eq:theta:2} we have
\begin{align*}
\prod_{i=1}^{|f|}\theta_{0,0}(\alpha_{i-1}-\alpha_i+t)&=\prod_{i=1}^{|f|}\theta_{0,0}(\alpha_{i-1}-\alpha_i-t+j_t+\ell_t\tau)\\
&=\prod_{i=1}^{|f|}(q^{-1}e^{-2i\pi(\alpha_{i-1}-\alpha_i-t)})^{\ell_t}\theta_{0,0}(\alpha_{i-1}-\alpha_i-t).
\end{align*}
The exponential term involving the angles $(\alpha_i)$ cancels because the latter are fixed. Moreover, recalling that $q=e^{i\pi\tau}$, we have $q^{-\ell_t}e^{-2i\pi \ell_t t}=e^{-i\pi \ell_t \tau} e^{i\pi (j_t \ell_t+\ell_t\tau)}=e^{i\pi j_t\ell_t}$. The latter is equal to 1 because $j_t \ell_t=0$, thus concluding the first part of the statement and the condition on equality of squared face-weights at faces corresponding to dual vertices $f$.

Fix $\rho,t$ as above. Consider $\mapalpha\in\Acal_{\rho,\,\ell_t/2}$ and the absolute value $\eps^*\Js^*_{(\mapalpha,\rho,t)}$ of the dual coupling constants of Remark~\ref{rem:duality_again}. Then, if Equation~\eqref{equ:condition_square_face} is satisfied at every square face $y$, by Remark~\ref{rem:duality_again} we also have, for every $y$,
\[
\WI_{G^*,\eps^*\Js^*_{(\mapalpha,\rho,t)}}(y)=\W_{G^*,\mapalpha,\rho,t+\rho}(y). 
\]
Otherwise stated Condition~\eqref{equ:condition_square_face} is also satisfied from the dual graph $G^*$ perspective. Repeating the above argument from the dual perspective yields that, for every face corresponding to a primal vertex $v$ (which now corresponds to a dual vertex of $G^*$), 
\[
(\WI)^2_{G^*,\eps^*\Js^*_{(\mapalpha,\rho,t)}}(v)=(\W)^2_{G^*,\mapalpha,\rho,t+\rho}(v),
\]
if and only if $2(t+\rho) = 0\ [\Lambda]$, \emph{i.e.}  $2t=0\ [\Lambda]$, since $2\rho=0\ [\Lambda]$. Now, using Equation~\eqref{equ:weak_duality_Ising} and Lemma~\ref{lem:duality} we obtain:
\[
(\WI)^2_{G,\eps\Js_{(\mapalpha,\rho,t)}}(v)=(\W)^2_{G,\mapalpha,\rho,t}(v),
\]
if and only if $t=0\ [\Lambda]$, thus proving equality of the squared face-weights at faces corresponding to primal vertices of $G$, and ending the proof. 
\end{proof}
\begin{rem}\label{rem:prop_t}
Note that the allowed values of $t$ are derived from assuming genericity of the angle map $\mapalpha$, but that if $t$ is as given by Point~1. or Point~2., then Condition~\eqref{equ:cond_square_face_bis} is satisfied.
\end{rem}
Proposition~\ref{prop:t}, and Remark~\ref{rem:prop_t} motivate the following. 
\begin{defi}\label{def:necessary_IV}
We say that the parameters $\tau,\mapalpha,\rho,t$ of Fock's dimer model on the infinite, minimal graph $\GQ$ satisfy \emph{Necessary Conditions IV} if they satisfy Necessary Conditions III (implying in particular that $\tau\in i\RR^+$ up to modular transformations, and that the angle map $\mapalpha$ belongs to $\Acal_{\rho,\,\ell_t/2}$), and if moreover:
\begin{enumerate}
 \item either $\rho=\frac{1}{2}$ and $t\in\{0,\frac{1}{2}\}$, 
 \item or $\rho=\frac{1}{2}+\frac{\tau}{2}$, $t=\frac{1}{2}$ and the dual graph $G^*$ is bipartite; or $\rho=\frac{1}{2}+\frac{\tau}{2}$, $t=\frac{\tau}{2}$ and the graph $G$ is bipartite. 
\end{enumerate}
\end{defi}

\subsubsection{Computation of face-weights using Jacobi elliptic functions}\label{sec:jacobi:elliptic}

We now compute face-weights of Fock's dimer model on the infinite, minimal graph $\GQ$ with parameters $\tau,\mapalpha,\rho,t$ satisfying Necessary Conditions IV. It turns out that they have natural expressions using Jacobi elliptic functions, whose definition we now recall. Consider a modular parameter $\tau\in i\RR^+$. The \emph{elliptic modulus} and \emph{complementary elliptic modulus} are defined as
\begin{equation}\label{equ:def_k_kp}
  k = k(\tau) = \frac{\theta_{1,0}^2(0,\tau)}{\theta_{0,0}^2(0,\tau)} \quad ; \quad k' = k'(k)=k'(\tau) = \frac{\theta_{0,1}^2(0,\tau)}{\theta_{0,0}^2(0,\tau)};
\end{equation}
they satisfy $k^2+(k')^2=1.$
When $\tau\in i\RR^+$, then $k^2$ belongs to $[0,1)$ and $(k')^2$ to $(0,1]$ (we allow $\tau = i\infty$ as a degenerate case but not $\tau = 0$).

From now on, we thus write parameters of Fock's dimer model as $k,\mapalpha,\rho,t$, with $k^2\in[0,1)$; the parameter $k$ is considered up to modular transformations. 

The \emph{complete elliptic integral of the first kind}, denoted by $K(k)$, and its complementary, denoted by $K'(k)$, are defined by:
\begin{equation*}
  K = K(k) = \int_{0}^{\pi/2} \frac{1}{\sqrt{1-k^2 \sin^2 x}}dx \quad ; \quad K'(k) := K(k'),
\end{equation*}
see Equation~22.2.2 of \cite{DLMF}. They are related to the modular parameter by $\tau=i\frac{K'(k)}{K(k)}$.

\emph{Jacobi elliptic functions} are periodic, meromorphic functions on $\CC$ characterized, up to normalization, by the position of their poles and zeros. There are twelve of them, and we only introduce those used in this paper. Introducing the notation $\alphah=2K(k)\alpha$, they are expressed using the following ratios of theta functions, see~\cite[Section 2.1]{Lawden} for example. 
\begin{equation}\label{equ:elliptic_def}
\begin{split}
\sn(\alphah|k)&=-
\frac{\theta_{0,0}(0,\tau)}{\theta_{1,0}(0,\tau)}\frac{\theta_{1,1}(\alpha,\tau)}{\theta_{0,1}(\alpha,\tau)},\quad 
\cn(\alphah|k)=\frac{\theta_{0,1}(0,\tau)}{\theta_{1,0}(0,\tau)}\frac{\theta_{1,0}(\alpha,\tau)}{\theta_{0,1}(\alpha,\tau)},\\
\dn(\alphah|k)&=\frac{\theta_{0,1}(0,\tau)}{\theta_{0,0}(0,\tau)}\frac{\theta_{0,0}(\alpha,\tau)}{\theta_{0,1}(\alpha,\tau)}.
\end{split}
\end{equation}
Whenever no confusion occurs, we omit the arguments $\tau$ and $k$ in the notation.
The other Jacobi elliptic functions are defined in terms of these three basic functions, for example 
{\small 
\begin{equation}\label{equ:elliptic_def_2}
\begin{split}
\sc(\alphah)&=-\frac{\theta_{0,0}(0)}{\theta_{0,1}(0)}\frac{\theta_{1,1}(\alpha)}{\theta_{1,0}(\alpha)},\quad
\dc(\alphah)=\frac{\theta_{1,0}(0)}{\theta_{0,0}(0)}\frac{\theta_{0,0}(\alpha)}{\theta_{1,0}(\alpha)},\quad 
\sd(\alphah)= -\frac{\theta_{0,0}^2(0)}{\theta_{0,1}(0)\theta_{1,0}(0)}\frac{\theta_{1,1}(\alpha)}{\theta_{0,0}(\alpha)}.
\end{split}
\end{equation}
}

Recall that our final goal is to control the sign of $\W_{\mapalpha,\rho,t}(\,\cdot\,)$ when computed at faces corresponding to primal and dual vertices. When $k^2\in[0,1)$, the sign of elliptic functions is well known when their arguments are real. Hence, in the following we express face-weights using real arguments. When no real part is taken, it means that the argument is real.
We use the notation introduced in Section~\ref{subsec:face_weights_Fock}, see Figures~\ref{fig:primal_dual} and~\ref{fig:square_face}.

\begin{lem}\label{lem:face_weights_again}
Consider Fock's dimer model on an infinite, minimal graph $\GQ$ with parameters 
$k,\mapalpha,\rho,t$ satisfying Necessary Conditions IV, implying in particular that $k^2\in[0,1)$ up to modular transformations, and that $\mapalpha$ belongs to $\Acal_{\rho,\,\ell_t/2}$. Then,
\begin{enumerate}
\item[$\bullet$] When $\rho=\frac{1}{2}$, and $t\in\{0,\frac{1}{2}\}$, every angle map $\mapalpha$ belonging to $\Acal_{\frac{1}{2},0}$ takes values in $\RR~[\Lambda]$, and face-weights are explicitly given by:
\renewcommand\arraystretch{2}
\begin{equation}\label{table1}
  \begin{array}{|c|c|c|}
  \hline 
  \W_{\mapalpha,\rho,t}(\,\cdot\,)& t=0 & t=\frac{1}{2}  \\
  \hline \hline 
  y&-\sc^2(\betah_e-\alphah_e)&-(k')^2\sc^2(\betah_e-\alphah_e)\\
  \hline 
  f&\displaystyle{(-1)^{|f|}\prod_{i=1}^{|f|}\ns(\betah_{i}-\alphah_i)} &\displaystyle{(-1)^{|f|} \prod_{i=1}^{|f|}\frac{1}{k'}\ds(\betah_{i}-\alphah_{i})}   \\
  \hline 
  v&\displaystyle{(-1)^{|v|}\prod_{i=1}^{|v|}k'\sd(\betath_{i}-\alphath_{i})} &  \displaystyle{(-1)^{|v|} \prod_{i=1}^{|v|}\sn(\betath_{i}-\alphath_{i})}\\
\hline 
  \end{array}
\end{equation}
\item[$\bullet$] When $\rho=\frac{1}{2}+\frac{\tau}{2}$, and $t=\frac{1}{2}$, the dual graph $G^*$ is bipartite. For every angle map $\mapalpha\in \Acal_{\frac{1}{2}+\frac{\tau}{2},0}$ the differences of angles $(\alpha_{i-1}-\alpha_i)$ at faces corresponding to dual vertices belong to $\RR~[\Lambda]$, and the differences of angles $(\betat_i-\alphat_{i})$ at faces corresponding to primal vertices belong to $\RR~[\Lambda]$.
 When $\rho=\frac{1}{2}+\frac{\tau}{2}$, and $t=\frac{\tau}{2}$, the graph $G$ is bipartite. For every angle map $\mapalpha\in \Acal_{\frac{1}{2}+\frac{\tau}{2},\frac{1}{2}}$ the differences of angles $(\beta_{i}-\alpha_i)$ belong to 
 $\RR~[\Lambda]$, and the differences of angles $(\alphat_i-\alphat_{i-1})$ belong to $\RR~[\Lambda]$.

\renewcommand\arraystretch{2}
\begin{equation}\label{table2}
  \begin{array}{|c|c|c|}
   \hline 
   \W_{\mapalpha,\rho,t}(\,\cdot\,)& t=\frac{1}{2}, \text{ $G^*$ bipartite } & t=\frac{\tau}{2}, \text{ $G$ bipartite } \\
  \hline \hline 
  y&-\displaystyle{\frac{(k')^2}{k^2} \nc^2(\R(\betah_e-\alphah_e))}&-k^2 \sd^2(\R(\betah_e-\alphah_e))\\
  \hline 
  f&\displaystyle{(-1)^{|f|} \prod_{i=1}^{|f|}\nd(\alphah_{i-1}-\alphah_i)} &\displaystyle{(-1)^{|f|}\prod_{i=1}^{|f|}\frac{1}{k}\ns(\betah_{i}-\alphah_{i})}  \\
  \hline 
  v&\displaystyle{(-1)^{|v|}\prod_{i=1}^{|v|}k\sn(\betath_{i}-\alphath_{i})} & (-1)^{|v|}\prod_{i=1}^{|v|}\dn(\alphath_{i}-\alphath_{i-1})\\
\hline 
  \end{array}
\end{equation}
\end{enumerate}
\end{lem}

\begin{proof}\leavevmode

$\bullet$ Case $\rho=\frac{1}{2}$ and $t\in\{0,\frac{1}{2}\}$. Setting the value $(j,\ell)=(1,0)$ and $t=0$, $t=\frac{1}{2}$ in Equation~\eqref{equ:cond_square_face_bis} and using Identity~\eqref{equ:theta_ratios} in the second case gives:
\begin{align*}
\W_{\mapalpha,\rho,t}(y)
&=-\frac{\theta_{0,0}^2(0)}{\theta_{0,1}^2(0)}\frac{\theta_{1,1}^2(\beta_e-\alpha_e)}{
      \theta_{1,0}^2(\beta_e-\alpha_e)}=\sc^2(\betah_e-\alphah_e)\\
\W_{\mapalpha,\rho,t}(y)
&=-\frac{\theta_{0,0}^2(\frac{1}{2})}{\theta_{0,1}^2(\frac{1}{2})}\frac{\theta_{1,1}^2(\beta_e-\alpha_e)}{\theta_{1,0}^2(\beta_e-\alpha_e)}=-\frac{\theta_{0,1}^2(0)}{\theta_{0,0}^2(0)}\frac{\theta_{1,1}^2(\beta_e-\alpha_e)}{\theta_{1,0}^2(\beta_e-\alpha_e)}=-(k')^2\sc^2(\betah_e-\alphah_e),      
\end{align*}
where in the last equalities we used Equations~\eqref{equ:elliptic_def_2} and~\eqref{equ:def_k_kp}. 

For the face corresponding to a dual vertex $f$, setting the values $(j,\ell)=(1,0)$ and $t=0$ in Equation~\eqref{equ:face_weight_dual_bis} gives:
\begin{align*}
\W_{\mapalpha, \rho, t}(f)&=\prod_{i=1}^{|f|}\frac{\theta_{1,1}(-\frac{1}{2})}{\theta_{0,0}(0)}
\frac{\theta_{0,0}(-\frac{1}{2}+\beta_{i}-\alpha_i)}{\theta_{1,1}(\beta_{i}-\alpha_i)}
=\prod_{i=1}^{|f|}\frac{\theta_{1,0}(0)}{\theta_{0,0}(0)}
\frac{\theta_{0,1}(\beta_{i}-\alpha_i)}{\theta_{1,1}(\beta_{i}-\alpha_i)},
\end{align*} 
using that $\theta_{m,n}$ is odd iff $(m,n)=(1,1)$, and Identity~\eqref{eq:theta:4}. The proof is concluded using Definition~\eqref{equ:elliptic_def}. 

For the face corresponding to a primal vertex $v$, setting the values $(j,\ell)=(1,0)$ and $t=0$ in Equation~\eqref{equ:face_weight_primal_bis} gives:
\begin{align*}
\W_{\mapalpha, \rho, t}(v)&=\prod_{i=1}^{|v|}\frac{\theta_{0,0}(-\frac{1}{2})}{\theta_{1,1}(-\frac{1}{2})}\frac{\theta_{1,1}(\betat_{i}-\alphat_i)}{\theta_{0,0}(\alphat_{i}-\betat_i)}=
(1)^{|v|}\prod_{i=1}^{|v|}\frac{\theta_{0,1}(0)}{\theta_{1,0}(0)}\frac{\theta_{1,1}(\betat_{i}-\alphat_i)}{\theta_{0,0}(\betat_{i}-\alphat_i)},
\end{align*} 
using again the parity properties of the function $\theta_{m,n}$, and Identity~\eqref{eq:theta:4}.
The proof is concluded using Definitions~\eqref{equ:elliptic_def} and ~\eqref{equ:elliptic_def_2}.
This ends the proof of the first column. For the second column observe that, when $t=0$, $t+\rho=\frac{1}{2}$; and the results are obtained using the duality relations~\eqref{equ:duality_relations}, with the additional observation that when theses relations are expressed with $\alpha_i,\beta_i$ instead of $\alpha_i,\alpha_{i-1}$, there is an extra sign change, implying that for every $i$, the order of the angles is not exchanged. 


$\bullet$ Case $\rho=\frac{1}{2}+\frac{\tau}{2}$, and $t\in\{\frac{1}{2},\frac{\tau}{2}\}$.  Setting the value $(j,\ell)=(1,1)$ and $t=\frac{1}{2}$, $t=\frac{\tau}{2}$ in Equation~\eqref{equ:cond_square_face_bis} and using Identity~\eqref{equ:theta_ratios} in the first case gives
\begin{align*}
\W_{\mapalpha,\rho,t}(y)&=-\frac{\theta_{0,0}^2(\frac{1}{2})}{\theta_{1,1}^{2}(\frac{1}{2})}\frac{\theta_{0,1}^2(\Re(\beta_e-\alpha_e))}{\theta_{1,0}^2(\R(\beta_e-\alpha_e))}=-\frac{\theta_{0,1}^2(0)}{\theta_{1,0}^{2}(0)}
\frac{\theta_{0,1}^2(\Re(\beta-\alpha))}{\theta_{1,0}^2(\R(\beta-\alpha))}=
-\frac{(k')^2}{k^2} \nc^2(\R(\betah_e-\alphah_e))\\
\W_{\mapalpha,\rho,t}(y)&=-\frac{\theta_{1,0}^2(0)}{\theta_{0,1}^{2}(0)}\frac{\theta_{1,1}^2(\Re(\beta_e-\alpha_e))}{\theta_{0,0}^2(\R(\beta_e-\alpha_e))}=-k^2 \sd^2(\R(\betah_e-\alphah_e)),
\end{align*}
where in the last equalities, we used Equations~\eqref{equ:elliptic_def}~\eqref{equ:elliptic_def_2} and~\eqref{equ:def_k_kp}. 

For the face corresponding to a dual vertex $f$, setting the value $(j,\ell)=(1,1)$ and $t=\frac{1}{2}$ in Equations~\eqref{eq:dual:face-weight} of Lemma~\eqref{lem:alternate_product_faces}, and using Identity~\eqref{eq:theta:4} gives
\begin{align*}
\W_{\mapalpha,\rho,t}(f)&=(-1)^{|f|}\prod_{i=1}^{|f|}\frac{\theta_{0,0}(0)}{\theta_{0,0}(\frac{1}{2})}\frac{\theta_{0,0}(\frac{1}{2}+\alpha_{i-1}-\alpha_i)}{\theta_{0,0}(\alpha_{i-1}-\alpha_i)}=
(-1)^{|f|}\prod_{i=1}^{|f|}\frac{\theta_{0,0}(0)}{\theta_{0,1}(0)}\frac{\theta_{0,1}(\alpha_{i-1}-\alpha_i)}{\theta_{0,0}(\alpha_{i-1}-\alpha_i)},
\end{align*}
and the proof is concluded using Equation~\eqref{equ:elliptic_def}.
 For the face corresponding to a primal vertex $v$, since $\alphat_i-\alphat_{i-1}$ is not real, but $\alphat_i-\betat_i$ is, we rather use expression~\eqref{equ:face_weight_primal_bis} with $(j,\ell)=(1,1)$ and $t=\frac{1}{2}$. 
\begin{align*}
\W_{\mapalpha,\rho,t}(v)&=\prod_{i=1}^{|v|}
\frac{\theta_{0,0}(\frac{1}{2}-\frac{1}{2}-\frac{\tau}{2})}{\theta_{1,1}(-\frac{1}{2}-\frac{\tau}{2})}
\frac{\theta_{1,1}(\betat_{i}-\alphat_{i})}{\theta_{0,0}(\frac{1}{2}+\alphat_i-\betat_{i})}\\
&= (-1)^{|v|}\prod_{i=1}^{|v|}
\frac{\theta_{0,0}(\frac{\tau}{2})}{\theta_{1,1}(\frac{1}{2}+\frac{\tau}{2})}
\frac{\theta_{1,1}(\betat_{i}-\alphat_{i})}{\theta_{0,0}(\betat_{i}-\alphat_i-\frac{1}{2})}, \text{ using that $\theta_{m,n}$ is odd iff $(m,n)=(1,1)$}\\
&= (-1)^{|v|}\prod_{i=1}^{|v|}
\frac{\theta_{1,0}(0)}{-\theta_{0,0}(0)}
\frac{\theta_{1,1}(\betat_{i}-\alphat_{i})}{\theta_{0,1}(\betat_{i}-\alphat_i)}, \text{ using Equations~\eqref{eq:theta:4} \eqref{eq:theta:5}}.
\end{align*}
The proof is again concluded using Equations~\eqref{equ:elliptic_def} and Equation~\eqref{equ:def_k_kp}, and observing that since the dual graph $G^*$ is bipartite, the degree of every vertex $v$ is even. 
For the second column observe that, when $t=\frac{1}{2}$, $t+\rho=\frac{\tau}{2}\ [\Lambda]$. The results are obtained using the duality relations~\eqref{equ:duality_relations}.
\end{proof}

\subsection{Equality of face-weights: Necessary Conditions V}\label{sec:Necessary_V}

Consider an Ising model on an infinite, isoradial graph $G$ with real coupling constants $\eps\Js$, and the corresponding Ising-dimer model on the minimal graph $\GQ$. Consider Fock's dimer model with parameters $k,\mapalpha,\rho,t$. Then, if the two models are gauge equivalent, the parameters $k,\mapalpha,\rho,t$ must satisfy Necessary Conditions IV, arising from  
equality of face-weights at all square faces, and from equality of \emph{squared face-weights} at all faces of $\GQ$ corresponding to primal, resp. dual vertices. If the two models are gauge equivalent, we actually have equality of \emph{face-weights} at all faces of $\GQ$ corresponding to primal, resp. dual vertices, that is, using Equation~\eqref{equ:gauge_Ising_dimers}
\begin{align}
\forall\, v\in V,&\quad \WI_{\Js_{(\mapalpha,\rho,t)}}(v)=\W_{\mapalpha,\rho,t}(v)\in(-1)^{|v|-1}\RR^+ \label{equ:condition_primal_face}\\
\forall\,f\in V^*,&\quad  \WI_{\Js_{(\mapalpha,\rho,t)}}(f)=\W_{\mapalpha,\rho,t}(f)\in\ (-1)^{|f|-1}\delta_f\RR^+, \label{equ:condition_dual_face}
\end{align}
where recall that $\delta_f$ is the frustration function of the Ising model on $G$ at the face $f$. In the following, we identify 
conditions on the angle map $\mapalpha$ so that the sign conditions given by Equations~\eqref{equ:condition_primal_face} and~\eqref{equ:condition_dual_face} are satisfied.

\begin{prop}\label{prop:Necessary_V}
Consider Fock's dimer model on an infinite, minimal graph $\GQ$ with parameters $k,\mapalpha,\rho,t$, satisfying Necessary Conditions IV, implying in particular that $k^2\in[0,1)$ up to modular transformations, that $\mapalpha$ belongs to $\A_{\rho,\,\ell_t/2}$, and that $\rho,t$ are as in Lemma~\ref{lem:face_weights_again}. Then, 
\begin{enumerate}
 \item When $\rho=\frac{1}{2}+\frac{\tau}{2}$ and $t=\frac{\tau}{2}$ (implying that the graph $G$ is bipartite), the sign condition arising from~\eqref{equ:condition_primal_face} is never satisfied. 
 \item When $\rho=\frac{1}{2}$ and $t\in\{0,\frac{1}{2}\}$, or $\rho=\frac{1}{2}+\frac{\tau}{2}$ and $t=\frac{1}{2}$ (implying that the graph $G^*$ is bipartite), the sign arising from condition~\eqref{equ:condition_primal_face} is satisfied if and only if the angle map $\mapalpha$ is such that, for every face of $\GQ$ corresponding to a vertex $v$ of $G$, using the notation of Figure~\ref{fig:primal_dual} (right), the following condition is satisfied:
\[
\prod_{i=1}^{|v|}\sn(\betath_{i}-\alphath_{i}) \in \RR^-.
\]
\begin{enumerate}
\item When $\rho=\frac{1}{2}$ and $t\in\{0,\frac{1}{2}\}$ the face-weight at the face of $\GQ$ corresponding to the face $f$ is non-frustrated, resp. frustrated if, using the notation of Figure~\ref{fig:primal_dual} (left)
\[
\prod_{i=1}^{|f|}\sn(\betah_{i}-\alphah_i)\in \RR^-,\text{ resp. }
\prod_{i=1}^{|f|}\sn(\betah_{i}-\alphah_i)\in \RR^+.
\]
\item When $\rho=\frac{1}{2}+\frac{\tau}{2}$ and $t=\frac{1}{2}$, the face-weight is frustrated at every face of $\GQ$ corresponding to the face $f$. 
\end{enumerate}
\end{enumerate}
\end{prop}

\begin{proof}
Since the parameters $k,\mapalpha,\rho,t$ satisfy Necessary Conditions IV, face-weights of Fock's dimer model on $\GQ$ are given by Tables~\eqref{table1} and~\eqref{table2}. Observe that all differences of angles involved are real. In the following, we use that $k^2\in[0,1)$, and that the Jacobi elliptic function $\dn$ is positive when $k^2\in[0,1)$ and when the argument is real.  

The sign condition arising from Equation~\eqref{equ:condition_primal_face} is: $\W_{\mapalpha,\rho,t}(v)\in(-1)^{|v|-1}\RR^+$. 
When $\rho=\frac{1}{2}+\frac{\tau}{2}$ and $t=\frac{\tau}{2}$, looking at the third row/second column of Table~\eqref{table2}, we see that this condition is never satisfied, thus proving Point 1. In all other cases, looking at the third row of Table~\eqref{table1}, and the third row/first column of Table~\eqref{table2}, we see that this condition is satisfied if and only if $\prod_{i=1}^{|v|}\sn(\betath_{i}-\alphath_{i}) \in \RR^-$ thus proving the first part of Point 2. 

By Equation~\eqref{equ:condition_dual_face} we have that the face-weight at the face of $\GQ$ corresponding to a dual vertex $f$ is non-frustrated, resp. frustrated, if 
\[
\W_{\mapalpha,\rho,t}(f)\in (-1)^{|f|-1}\RR^+,\quad \text{ resp., }\quad \W_{\mapalpha,\rho,t}(f)\in (-1)^{|f|}\RR^+.
\]
Using again that the elliptic function $\dn$ is positive, and looking at the second row of Table~\eqref{table1}, and the second row/first column of Table~\eqref{table2} proves Points (a) and (b). 
\end{proof}

This motivates the following. 
\begin{defi}\label{def:necessary_V}
We say that the parameters $k,\mapalpha,\rho,t$ of Fock's dimer model on the graph $\GQ$ satisfy \emph{Necessary Conditions V} if they satisfy Necessary Conditions IV (implying in particular that $k^2\in[0,1)$ up to modular transformations, and that the angle map $\mapalpha$ belongs to $\Acal_{\rho,\,\ell_t/2}$), and if moreover: either $\rho=\frac{1}{2}$ and $t\in\{0,\frac{1}{2}\}$, or $\rho=\frac{1}{2}+\frac{\tau}{2}$ and $t=\frac{1}{2}$ in which case the graph $G^*$ is bipartite; and if the angle map $\mapalpha$ is such that, for every face of $\GQ$ corresponding to a primal vertex $v$ of $G$ we have, using the notation of Figure~\ref{fig:primal_dual} (right):
\begin{equation}\label{equ:angle_V}
\prod_{i=1}^{|v|}\sn(\betath_{i}-\alphath_{i}) \in \RR^-.
\end{equation}
\end{defi}

\subsection{Classification and algebraic phase transition}\label{sec:classification_thm}

We are now ready to state and prove the main result of this paper, classifying Ising models whose with genus 1 spectral curves.

\subsubsection{Classification result}

\begin{thm}\label{thm:main}
Consider an Ising model on an infinite, periodic, isoradial graph $G$ with periodic real coupling constants $\eps\Js$, and the corresponding Ising-dimer model on the infinite, periodic, minimal graph $\GQ$; suppose that the spectral curve $\C$ has genus 1 and satisfies Assumptions $(\dagger)$. Consider Fock's gauge equivalent dimer model on $\GQ$ with modular parameter $\tau$, generic periodic angle map $\mapalpha$, parameters $\rho,t\in\TT(\tau)$. Then, the parameters $\tau,\mapalpha,\rho,t$ satisfy Necessary Conditions~V, and thus fall in one of the classes $(\rho,t)$ of Table~\eqref{table3}, for generic $\tau$, up to modular transformations of $\tau$. They are related to the absolute value coupling constants $\Js$, resp. the frustration function $\delta=(\delta_f)_{f\in F}$, by the third, resp. the fourth, row of Table~\eqref{table3}.

 {\small 
 \renewcommand\arraystretch{1.5}
 \begin{equation}\label{table3}
   \begin{array}{|c|c|c|c|}
  \hline 
   (\rho,t) 
   &(\frac{1}{2},0)
   & (\frac{1}{2},\frac{1}{2})
   &(\frac{1}{2}+\frac{\tau}{2},\frac{1}{2})\\
   \hline \hline 
   \text{Hyp. $G/G^*$}
   & \text{none }
   & \text{none }
   & G^* \text{ bipartite}\\
   \hline
   \sinh(2\Js_e)
     & |\sc(\beta_e-\alpha_e)|
     & k'|\sc(\beta_e-\alpha_e)|
     & \frac{k'}{k}|\nc\bigl(\Re(\betah_e-\alphah_e)\bigr)|\\
   \hline
   \Js_e
    & \frac{1}{2}\ln\bigl(\frac{1+|\sn(\beta_e-\alpha_e)|}{|\cn(\beta_e-\alpha_e)|}\bigr)
    & \frac{1}{2}\ln\bigl(\frac{\dn(\beta_e-\alpha_e) +k'|\sn(\beta_e-\alpha_e)|}{|\cn(\beta_e-\alpha_e)|}\bigr)
    & \frac{1}{2}\ln\bigl(\frac{k'+\dn\Re(\betah_e-\alphah_e)}{k|\cn\Re(\betah_e-\alphah_e)|} \bigr)\\
   \hline
   \delta_f 
   & - \sgn (\prod_{i=1}^{|f|}\sn(\betah_{i}-\alphah_i))
   & - \sgn (\prod_{i=1}^{|f|}\sn(\betah_{i}-\alphah_i))
   & -1\\
   \hline 
   \end{array}
 \end{equation}
 }

Conversely, consider Fock's dimer model on an infinite, minimal graph $\GQ$ with parameters $\tau,\mapalpha,\rho,t$ satisfying Necessary Conditions V, and suppose that the angle map $\mapalpha$ is periodic. Consider the absolute value coupling constants $\Js$ and the frustration function $\delta=(\delta_f)_{f\in F}$ given by Table~\eqref{table3}, and suppose that $\prod_{f\in F_1}\delta_f=1$. Then, there exists edge-signs $(\eps_e)_{e\in E}$ such that the Ising-dimer model on $\GQ$ arising from the coupling constants $\eps \Js$, and Fock's dimer model on $\GQ$, are gauge equivalent. Moreover, the associated spectral curve $\C$ has genus 1, is invariant by complex conjugation, is centrally symmetric, and Fock's dimer model is real. 
\end{thm}

\begin{proof}
Consider the assumptions of the first part of the statement. Then, the curve $\C$ arises from a polynomial with real coefficients, implying that it is invariant by complex conjugation and that Fock's gauge equivalent dimer model is real. By Theorems~\ref{thm:real_spectral_curve} and~\ref{thm:Fock_real_dimer_model} this implies that, generically and up to modular transformations, the modular parameter $\tau$ and the angle map $\mapalpha$ satisfy Necessary Conditions I, see Definition~\ref{def:Necessary_I}. Moreover, by Point 2. of Remark~\ref{rem:Fisher_Dub}, the spectral curve $\C$ is centrally symmetric. Using Theorem~\ref{thm:real_symmetric_spectral_curve}, this implies that the angle map $\mapalpha$ satisfies Necessary Conditions II, see Definition~\ref{def:necessary_II}. We then use the fact that the Ising-dimer model and Fock's dimer model on $\GQ$ are gauge equivalent. From equality of face-weights at square faces, we deduce from Propositions~\ref{prop:first_parameter_condition} and \ref{prop:global_condition_anglemap} that the parameters $\tau,\mapalpha,\rho$ satisfy Necessary Conditions III, see Definition~\ref{defi:necessary_III}; from equality of squared face-weights at faces corresponding to primal and dual vertices, we deduce from Proposition~\ref{prop:t} that they satisfy Necessary Conditions IV, see Definition~\ref{def:necessary_IV}; finally, from equality of face-weights at faces corresponding to primal and dual vertices, we deduce from Proposition~\ref{prop:Necessary_V} that they satisfy Necessary Conditions V, see Definition~\ref{def:necessary_V}. 
The explicit values of $\sinh(2\Js_e)$ come from Equation~\eqref{equ:cond_2}, Lemma~\ref{lem:face_weights_again} together with the fact that $\Js\geq 0$, and Proposition~\ref{prop:Necessary_V} which excludes the case $\rho=\frac{1}{2}+\frac{\tau}{2}$ and $t=\frac{\tau}{2}$. The explicit values of $\Js_e$ come from explicitly solving $\sinh(2\Js_e)$ for $\Js_e$ and using elliptic functions identities.
The computation of the frustration function comes from Equations~\eqref{equ:cond_1} and~\eqref{equ:cond_2} and Points $2(a)$, $2(b)$ of Proposition~\ref{prop:Necessary_V}. This ends the proof of the first part of Theorem~\ref{thm:main}.

Consider the assumptions of the reverse part of the statement. The absolute value coupling constants $\Js$, resp. the frustration function $\delta$, are given by the third, resp. the fourth, row of Table~\ref{table3}. 
Since $\prod_{f\in F_1}\delta_f=1$, by Lemma~\ref{lem:reconstructing_Ising_model}, there exists edge-signs $(\eps_e)_{e\in E}$ such that the Ising model on $G$ with coupling constants $\eps\Js$ lies in this equivalence class. 
Consider the corresponding Ising-dimer model on $\GQ$\footnote{Recall from Remark~\ref{rem:ising_model_gauge} that face-weights of the Ising-dimer model on $\GQ$ with coupling constants $\eps\Js$ only depend on the (Ising) gauge equivalence class of the model, that is, on the absolute value coupling constants $\Js$ and on the frustration function $(\delta_f)_{f\in F}$.}.
Now, since Necessary Conditions V are satisfied, Necessary Conditions III also are. As a consequence, by definition of the absolute value coupling constants $\Js$, by Propositions~\ref{prop:first_parameter_condition} and~\ref{prop:global_condition_anglemap}, we have equality of face-weights at square-faces. Since Necessary Condition IV is also satisfied, by Proposition~\ref{prop:t}, we have equality of squared face-weights at faces of $\GQ$ corresponding to primal and dual vertices of $G$. Finally, since Necessary Condition V is satisfied, we have equality of face-weights at faces of $\GQ$ corresponding to primal and dual vertices of $G$. Summarizing, face-weights are equal at all faces of $\GQ$ so that Fock's dimer model and the Ising-dimer model on $\GQ$ are gauge equivalent. Since the Ising-dimer model is defined on a periodic graph $\GQ$ and has periodic parameters, its spectral curve $\C$ is well defined. Since it is parameterized by elliptic functions, see Equation~\eqref{equ:param_spectral_curve_symm}, it has genus 1; and since the parameters of Fock's dimer model satisfy Necessary Conditions II, by Theorems~\ref{thm:real_spectral_curve} and~\ref{thm:real_symmetric_spectral_curve}, it is invariant by complex conjugation, and centrally symmetric. Since Necessary Conditions I are satisfied, Fock's dimer model is real. 
\end{proof}

\begin{rem}
We will be short on the comments on Theorem~\ref{thm:main} already given in the introduction; we also have some additional remarks.
\begin{enumerate}
\item An explicit parameterization of the spectral curve $\C$ is given by Equation~\eqref{equ:param_spectral_curve_symm}, helping us get more insight on the different cases. There are three family of models: Family~I consisting of the first two columns in the case where the frustration function is identically equal to 1; Family II consisting of the same two columns when the frustration also takes negative values; Family III corresponding to the third column. Generic curves in the three families are described in Figure~\ref{fig:triangular_lattice_curves}.
\item Note that Assumption $(\dagger)$ in the first part of the statement is there to ensure genericity of the angle map $\mapalpha$; this is needed in the proofs of Theorems~\ref{thm:real_spectral_curve},~\ref{thm:Fock_real_dimer_model} and~\ref{thm:real_symmetric_spectral_curve}, but not in the reverse statement. Indeed, in the reverse statement, we do not prove any regularity/irreducibility property of the curve; for example, an instance of a \emph{reducible} curve is provided by the fully frustrated isotropic Ising model on the square lattice, see Section~\ref{subsec:square_lattice}.
\item The modular parameter $\tau$ belongs to $i\RR^+$ (or equivalently $k^2$ belongs to $[0,1)$), up to modular transformations; but recall that, by Lemma~\ref{lem:modular_transo_dimers}, we know how Fock's dimer model evolves under modular transformations; hence this is not restrictive, and considering modular equivalent parameters $\tau$ will not give new families of models. On the same topic, in the paper~\cite{BdTR2} on Baxter's $Z$-invariant Ising model, we actually study the model of the first column (when $\delta\equiv 1$) in the range $k^2\in (-\infty,1)$. Using modular transformation, this is equivalent to studying the models of the first two columns (when $\delta\equiv 1$) in the range $k^2\in[0,1)$. 
\item Using Equations~\eqref{equ:weakly_dual_coupling_constants} and Lemma~\ref{lem:duality}, we know that the model of the first column defined on the graph $G$, and the model of the second column defined on the dual graph $G^*$, are \emph{weakly dual}. When the frustration function $\delta\equiv 1$, they are actually \emph{dual} models, see also~\cite{BdTR2}. For frustrated models, by Remark~\ref{rem:duality}, we can only have \emph{weak} dual models, because the sign conditions cannot be satisfied at primal and dual faces. This is why, the weakly dual model $\rho=\frac{1}{2}+\frac{\tau}{2}$, $t=\frac{\tau}{2}$ of the model $\rho=\frac{1}{2}+\frac{\tau}{2}$, $t=\frac{1}{2}$, which is fully frustrated, is not present in the classification, a fact also observed in Proposition~\ref{prop:Necessary_V}. 
\end{enumerate}

\end{rem}

\subsubsection{Algebraic phase transition}

We now give a few elements regarding phase transition. 

\begin{defi}\label{defi:algebraic_phase_transition}
We say that the Ising-dimer model on $\GQ$ undergoes an \emph{algebraic phase transition} if, when moving one parameter of the model, the spectral curve $\C$ goes from a genus 1 curve to a genus 0 curve. The corresponding \emph{algebraic critical temperature} is the value of the coupling constants when the genus 0 spectral curve is reached.
\end{defi}

When $k^2$ goes to 0, or equivalently $\Im(\tau)$ goes to infinity, the torus $\TT(\tau)$ becomes a cylinder, and the genus 1 curve $\C$ of Fock's dimer model degenerates to a genus 0 curve, see for example~\cite[Section 8]{BCdT:genus_1}.
Hence, as an immediate consequence of Theorem~\ref{thm:main}, using~\cite[2.1.18-2.1.20]{Lawden}, we obtain the following. 

\begin{cor}\label{cor:algebraic_phase_transition}
Consider Fock's dimer model on an infinite, minimal graph $\GQ$ with parameters $k,\mapalpha,\rho,t$ satisfying Necessary Conditions V, and suppose that the angle map $\mapalpha$ is periodic. Consider the absolute value coupling constants $\Js$ and the frustration function $\delta=(\delta_f)_{f\in F}$ given by Table~\eqref{table3}, and suppose that $\prod_{f\in F_1}\delta_f=1$; observe that the sign of the frustration function $\delta$ is independent of $k^2$. Consider the gauge equivalent Ising-dimer model on $\GQ$ with coupling constants $\eps\Js$ given by Theorem~\ref{thm:main}. Then, as $k^2$ tends to 0, the Ising-dimer model on $\GQ$ undergoes an algebraic phase transition, with algebraic critical absolute value coupling constants $\Js^{\mathrm{crit}}$ given by 

{\small 
 \renewcommand\arraystretch{1.5}
 \begin{equation}\label{cor_table3}
   \begin{array}{|c|c|c|c|}
  \hline 
   (\rho,t) 
   &(\frac{1}{2},0)
   & (\frac{1}{2},\frac{1}{2})
   &(\frac{1}{2}+\frac{\tau}{2},\frac{1}{2})\\
   \hline \hline 
   \text{Hyp. $G/G^*$}
   & \text{none }
   & \text{none }
   & G^* \text{ bipartite}\\
   \hline
   \sinh(2\Js_e^{\mathrm{crit}})
     & |\tan(\beta_e-\alpha_e)|
     & |\tan(\beta_e-\alpha_e)|
     & \infty\\
   \hline
   \Js_e^{\mathrm{crit}}
    & \frac{1}{2}\ln\bigl(\frac{1+|\sin(\beta_e-\alpha_e)|}{|\cos(\beta_e-\alpha_e)|}\bigr)
    & \frac{1}{2}\ln\bigl(\frac{1 +|\sin(\beta_e-\alpha_e)|}{|\cos(\beta_e-\alpha_e)|}\bigr)
    & \infty \\
   \hline
   \end{array}
 \end{equation}
 }
\end{cor}

The notion of algebraic phase transition is a priori weaker than that of phase transition. Let us explain why. As said in the introduction, one way of characterizing the phase transition is by analyzing the change of regularity in the \emph{free energy}~\cite{Baxter:exactly,Li10,CimDum12,Friedli_Velenik}. Consider the Ising model with periodic, real coupling constants $\eps\Js$ defined on an infinite, periodic isoradial graph $G$, and let $(G_n)_{n\geq 1}$ be a toroidal exhaustion of $G$. Then, the \emph{free energy}, denoted by $f$, is
\begin{equation}\label{equ:free_energy}
f=-\lim_{n\rightarrow\infty} \frac{1}{n^2} \log (\Zising(G_n,\Js)).
\end{equation}
The existence of the above limit is for example proved in~\cite{Friedli_Velenik}. Moreover, consider the Fisher dimer characteristic polynomial $\Pising(z,w)$, then we have the following. 

\begin{prop}\cite{CKP,KOS,BoutillierdeTiliere:iso_perio}\label{prop:phase:transition}
Suppose that the spectral curve $\C$ of the Ising-dimer model on $\GQ$ intersects the torus $\TT$ in a finite number of points, then the free energy of the Ising model is equal to 
\begin{equation*}
f=-\frac{1}{2}\sum_{e\in E_1}\log\cosh(\Js_e)-\frac{1}{2}\frac{1}{(2i\pi)^2}\int_{\TT^2}\log \Pising(z,w)\frac{dz}{z}\frac{dw}{w}.
\end{equation*}
\end{prop}

\begin{rem}\label{rem:phase_transition}\leavevmode
\begin{enumerate}
\item The proof is now very classical. Using Fisher's correspondence, one relates the Ising and Fisher-dimer partition functions. Then, one uses Kasteleyn theory to compute the Ising-dimer partition function, providing Riemann sums for the above integral; most of the work lies in proving convergence of the Riemann sums to the integral, the delicate point being when the characteristic polynomial has zeros on the unit torus. The proof goes through if this number is finite, hence the assumption on the intersection of the spectral curve $\C$ with the torus $\TT$. 
\item As long as the characteristic polynomial has no zero on the torus $\TT$, the free energy is smooth; a change of behavior occurs when zeros are reached, see for example~\cite{KOS,Li10}. Hence, the key information to have is whether the amoeba of the spectral curve has a hole containing the origin (smooth) or not (phase transition). When the spectral curve is Harnack, this information can be read directly on the genus of the curve but, in our setting of curves that are not necessarily Harnack, this is not the case since, in general, the amoeba is not equal to the log of the modulus of the real locus of the curve. Hence the fact that, when $k\neq 0$, the genus of the curve is 1, and when $k=0$, it is equal to 0 is a \emph{strong indication} that there is a phase transition, but not a proof. 
\item In the case of the isotropic frustrated Ising model on the square lattice, we will see in Section~\ref{subsec:square_lattice} that the spectral curve is a double copy of a Harnack curve, thus this proposition applies.
\end{enumerate}
\end{rem}

\section{Examples}\label{sec:examples}

The aim of this section is to provide examples of applications of Theorem~\ref{thm:main}. More specifically, in Section~\ref{sec:monotone_angle_maps}, we consider Fock's dimer model with parameters $k^2\in[0,1)$, $(\rho,t)\in\{(\frac{1}{2},0),(\frac{1}{2},\frac{1}{2}),(\frac{1}{2}+\frac{\tau}{2},\frac{1}{2})\}$, $\mapalpha\in\Acal_{\rho,\ell_t/2}$. We then use the notion of monotone angle maps studied in~\cite{BCdT:immersions,BCdT:genus_1} to give examples of angle maps $\mapalpha$ satisfying Necessary Condition V, \emph{i.e.}, the additional sign Condition~\eqref{equ:angle_V} of Definition~\ref{def:necessary_V}. Then, in Section~\ref{sec:triangular_lattice}, we derive a full solution of the Ising model on the triangular lattice using Theorem~\ref{thm:main}. Finally, in Section~\ref{subsec:square_lattice} we solve the case of the fully frustrated isotropic Ising model on the square lattice. 

In the whole of this section, we suppose that the graph $G$ is isoradial and periodic, meaning that the graph $\GQ$ is minimal and periodic. 

\subsection{Monotone angle maps}\label{sec:monotone_angle_maps}

We start by recalling definitions and properties of ordering of train-tracks of periodic minimal graphs, and the definition of monotone angle maps, see the original paper~\cite{BCdT:immersions} for more details. In Section~\ref{subsec:caseA}, we study implications in the case where $\rho=\frac{1}{2},t\in\{0,\frac{1}{2}\}$, then in Section~\ref{subsec:CaseB}, we study further implications when the dual graph $G^*$ is bipartite.

\subsubsection{Background}

Consider a periodic minimal graph $\Gs$, and let $\vec{\T}_1$ be the set of oriented train-tracks of the fundamental domain $\Gs_1$. The homology classes $(h_T,v_T)$ in $H_1(\TT,\ZZ)$ of the train-tracks consist of coprime integers. The total cyclic order on such pairs induces a \emph{partial cyclic order} on the set of train-tracks $\vec{\T}$~\cite{BCdT:immersions}. Since the graph $\Gs$ is assumed to be periodic, the partial cyclic order is equivalent to the \emph{global cyclic order} on the set of directed train-tracks $\vec{\T}$ of $\Gs$ defined as follows, see also Definition 6. and Remark 7. of~\cite{BCdT:immersions}. Consider a triple of train-tracks $(\vec{T}_1,\vec{T}_2,\vec{T}_3)$ of  $\vec{\T}$, pairwise non-parallel, and a compact disk $B$ outside which these train-tracks do not intersect, except possibly for anti-parallel train-tracks. Then order $(\vec{T}_1,\vec{T}_2,\vec{T}_3)$ cyclically according to the outgoing points of the corresponding oriented curves on the boundary $\partial B$. This partial cyclic order on $\vec{\T}$ is the \emph{global cyclic order} on $\vec{\T}$. 

An angle map $\mapalpha:\vec{\T}\rightarrow \RR/\ZZ$ on a globally cyclically ordered set $\vec{\T}$ is said to be \emph{monotone} if it respects the cyclic order on $\vec{\T}$ and on $\RR/\ZZ$ namely, whenever we have $[\vec{T}_1,\vec{T}_2,\vec{T}_3]$ in $\vec{\T}$, then we have $[\alpha_{\vec{T}_1},\alpha_{\vec{T}_2},\alpha_{\vec{T}_3}]$ in $\RR/\ZZ$. Following~\cite[Corollary 29]{BCdT:immersions}, we define $X'_\Gs$ to be the set of monotone angle maps $\mapalpha:\vec{\T}\rightarrow \RR/\ZZ$ mapping pairs of intersecting train-tracks to distinct angles. Since the minimal graph $\Gs$ is periodic, the set $X'_\Gs$ is equal to the set $\Ys_\Gs$ of angle maps satisfying a local ordering condition, see Definition 22 and Corollary 29 of~\cite{BCdT:immersions}, but since we do not need these considerations in this paper, we refer to the original paper for more details.  

\subsubsection{The case where $\rho=\frac{1}{2}$ and $t\in\{0,\frac{1}{2}\}$}\label{subsec:caseA}

Recall that in this case, the angle map $\mapalpha\in \Acal_{\rho,\,\ell_t/2}$ takes real values. The graph $\GQ$ is periodic and minimal. Then, as a consequence of~\cite[Theorem 31]{BCdT:immersions} and~\cite[Proposition 13]{BCdT:genus_1}, we obtain the following. 
\begin{cor}\label{cor:angle_map_1}
Consider Fock's dimer model on $\GQ$ satisfying Necessary Conditions IV in the case $\rho=\frac{1}{2}$ and $t\in\{0,\frac{1}{2}\}$. Suppose that the angle map $\mapalpha:\vec{\T}\rightarrow \RR$ belongs to $X'_{\GQ}$, then Necessary Conditions V is satisfied, and the model is non-frustrated at every face corresponding to a face $f$ of $G$. 
\end{cor}

\begin{rem}\label{rem:connexion}\leavevmode
\begin{enumerate}
\item In this case, we recover the $Z$-invariant Ising model of~\cite{Baxter:exactly}, studied in~\cite{ChelkakSmirnov2,BdTR2}.
\item The condition that the angle map $\mapalpha$ belongs to $X'_{\GQ}$ is sufficient, but clearly not necessary. A natural open question is to understand what are necessary conditions and more specifically to characterize features of the angle map leading to frustration. Concrete examples will be given in the case of the triangular lattice, see Section~\ref{sec:triangular_lattice} below. 
\end{enumerate}
\end{rem}

\subsubsection{The case where $G^*$ is bipartite}\label{subsec:CaseB}

Suppose now that the dual graph $G^*$ is bipartite. The set of train-tracks $\T$ of $G^*$ can then be consistently oriented so that white vertices of $G^*$ are on the left.  

Recall also from Section~\ref{sec:Necessary_III} and Figure~\ref{fig:dual_graph_bipartite} that that the set of oriented train-tracks 
$\vec{\T}$ of $\GQ$ is naturally split into $\vec{\T}=\vec{\T}_{\circ,*}\sqcup \vec{\T}_{\bullet,*}$, where: $\vec{\T}_{\circ,*}$ consists of the train-tracks turning counterclockwise around white vertices of $G^*$; $\vec{\T}_{\circ,*}$ is in orientation preserving bijection with $\T$; and for every $T\in\T$, $\vec{T}\in\vec{\T}_{\circ,*},\cev{T}\in\vec{\T}_{\bullet,*}$.
In the following computation, we prove that the sign condition around a face of $\GQ$ corresponding to a vertex $v$ of $G$ can be rewritten as a sign condition around the face $v$ of $G^*$. The set of oriented train-tracks $\T$ of $G^*$ around $v$ can further be split into two and labeled as $T_1,\dots,T_{|v|/2}$, $T_1',\dots,T_{|v|/2}'$, see Figure~\ref{fig:dual_graph_bipartite_bis} (left). 

\begin{figure}[h]
  \centering
  \begin{overpic}[width=12cm]{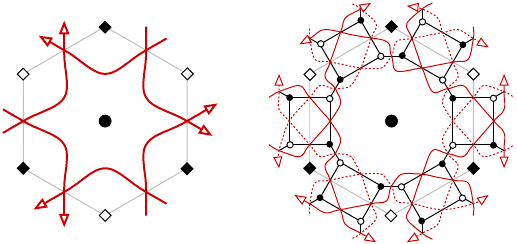} 
   \put(20,20){\scriptsize $v$}
   \put(75,20){\scriptsize $v$}
   \put(12,44){\scriptsize $T_1$}
   \put(6,41.5){\scriptsize $T_1'$}
   \put(12,1){\scriptsize $T_2'$}
   \put(3.5,6){\scriptsize $T_2$}
   \put(42,28){\scriptsize $T_3'$}
   \put(42,18){\scriptsize $T_3$}
   \put(53,12.5){\scriptsize $\alphat_1$}
   \put(95,37){\scriptsize $\betat_1$}
   \put(73,-2){\scriptsize $\alphat_2$}
   \put(72,45.5){\scriptsize $\betat_2$}
   \put(95,9){\scriptsize $\alphat_3$}
   \put(52,33.5){\scriptsize $\betat_3$}
   \put(98,33){\scriptsize $\alphat_4$}
   \put(54.5,8){\scriptsize $\betat_4$}
   \put(76,46){\scriptsize $\alphat_5$}
   \put(76.5,-2){\scriptsize $\betat_5$}
   \put(55,38){\scriptsize $\alphat_6$}
   \put(96,12){\scriptsize $\betat_6$}
  \end{overpic}
\caption{Left: oriented train-tracks of $\T$ when $G^*$ is the hexagonal lattice, and labeling $T_1,T_2,T_3$, $T_1',T_2',T_3'$.
Right: oriented train-tracks of $\vec{\T}$ split into $\vec{\T}_{\circ,*}$ (full lines) $\vec{\T}_{\bullet,*}$ (dotted lines), and labeling $\alphat_1,\dots,\alphat_6$, $\betat_1,\dots,\betat_6$ in accordance with the notation of Figure~\ref{fig:primal_dual} (right).}  \label{fig:dual_graph_bipartite_bis}  
\end{figure}

Fix a choice of angles 
$\hat{\alpha}_{T_1},\dots,\hat{\alpha}_{T_{|v|/2}},\hat{\alpha}_{T_1'},\dots,\hat{\alpha}_{T_{|v|/2}'}$ for the oriented train-tracks of $\T$. Using the orientation preserving bijection between $\T$ and $\vec{\T}_{\circ,*}$, this defines a value of the angle map $\alpha_{\vec{T}}$ for every oriented train-track $\vec{T}$ of $\vec{\T}_{\circ,*}$. Using further that, for every $T\in \T$, $\alpha_{\cev{T}}=\alpha_{\vec{T}}+\rho$, and the notation of Figure~\ref{fig:dual_graph_bipartite_bis}, we deduce
\begin{align*}
\alphat_1=\hat{\alpha}_{T_1}+\rho,\ \alphat_2=\hat{\alpha}_{T_2'},\ \alphat_3=\hat{\alpha}_{T_2}+\rho,\dots,
\alphat_{|v|-1}=\hat{\alpha}_{T_{|v|/2}}+\rho, \ 
\alphat_{|v|}=\hat{\alpha}_{T_1'}. 
\end{align*}
Recalling that the lift of the angles $(\betat_i)$ is fixed so that $\betat_i=\alphat_{i-1}+\rho$, and the notation $\alphah=2K\alpha$, we obtain
\begin{align}
\prod_{i=1}^{|v|}\sn(\betath_i-\alphath_i)&=\prod_{i=1}^{|v|/2}\sn(\hat{\alpha}^k_{T_i'}-\hat{\alpha}^k_{T_i})\sn(\hat{\alpha}^k_{T_i}+4K\rho-\hat{\alpha}^k_{T_{i+1}'})\nonumber \\
&=(-1)^{|v|/2}\prod_{i=1}^{|v|/2}\sn(\hat{\alpha}^k_{T_i'}-\hat{\alpha}^k_{T_i})\sn(\hat{\alpha}^k_{T_i}-\hat{\alpha}^k_{T_{i+1}'}),\label{equ:useful_computation}
\end{align}
where in the last equality, we used~\cite[2.2.11]{Lawden}. Recall from Definition~\ref{defi:angle_map} that in both cases $\rho=\frac{1}{2}$ and $\rho=\frac{1}{2}+\frac{\tau}{2}$, the angles of the train-tracks of $\vec{\T}_{\circ,*}$, that is of $\T$, take values in $\RR\,[\Lambda]$. Observing that the sign behavior of the function $\theta_{1,1}$ and $\sn$ is the same, as a consequence of~\cite[Theorem 31]{BCdT:immersions} and~\cite[Proposition 13]{BCdT:genus_1}, we obtain the following. 

\begin{cor}\label{cor:angle_map_2}
Suppose that the dual graph $G^*$ is bipartite, and 
consider Fock's dimer model on $\GQ$ satisfying Necessary Conditions IV in the case where 
$(\rho,t)\in\{(\frac{1}{2},0),(\frac{1}{2},\frac{1}{2}),(\frac{1}{2}+\frac{\tau}{2},\frac{1}{2})\}$. Suppose that the angle map $\hat{\mapalpha}:\T\rightarrow\RR$ belongs to $X'_{G^*}$, then Necessary Conditions V is satisfied. Moreover, when $\rho=\frac{1}{2}+\frac{\tau}{2}$ and $t=\frac{1}{2}$, Fock's dimer model is frustrated at every face $f$ of $G$. 
\end{cor}

\begin{rem}\leavevmode
\begin{enumerate}
 \item Let us emphasize that the angle map $\hat{\mapalpha}$ in the above statement is that of oriented train-tracks $\T$ of $G^*$ and not of the graph $\GQ$. Comparing to Corollary~\ref{cor:angle_map_1}, it is a condition on only half of the angles. 
 \item Again, the condition on the angle map $\hat{\mapalpha}$ is sufficient but not necessary, and understanding necessary conditions is a natural open question. In this case, there is more structure arising from the bipartitedness of $G^*$. For example, if the angle map $\hat{\mapalpha}$ belongs to $X_{G*}'$, is it true that the angle map $\mapalpha$ belongs to $X_{\GQ}'$? We believe that this is not the case, and that this depends on whether $G^*$ is isoradially embedded or minimally immersed. This will be explored further in a subsequent paper. 
\end{enumerate}
\end{rem}

\subsection{The triangular lattice case}\label{sec:triangular_lattice}

This section is dedicated to the case of the triangular lattice with smallest fundamental domain, see Figure~\ref{fig:dual_graph_bipartite_ter} (left). First, we consider Fock's dimer model on $\GQ$ with parameters $k,\mapalpha,\rho,t$ satisfying Necessary Conditions IV, then Proposition~\ref{prop:triangular_always_ising}, proves that when $(\rho,t)\in\{(\frac{1}{2},0),(\frac{1}{2},\frac{1}{2}),(\frac{1}{2}+\frac{\tau}{2},\frac{1}{2})\}$, 
for every choice of angle map $\mapalpha\in\Acal_{\rho,\,\ell_t/2}$, Necessary Condition V is always satisfied. By Theorem~\ref{thm:main}, this implies that, if $\prod_{f\in F_1}\delta_f=1$, Fock's dimer model is always gauge equivalent to an Ising model. 

Next, in Section~\ref{subsec:txo_examples} we provide two examples of families of frustrated Ising model using the second part of Theorem~\ref{thm:main}. In particular, we give a parameterization of the isotropic frustrated Ising model on the triangular lattice.

%
%
%

Finally in Section~\ref{subsec:frustrated_triang}, we consider the frustrated Ising model on $G$ with coupling constants $-\Js_1,-\Js_2,-\Js_3$ as in Figure~\ref{fig:dual_graph_bipartite_ter}, where $\Js_1,\Js_2  ,\Js_3 > 0$; the frustration function is $-1$ at every face, and the product over the two faces is equal to 1. In Theorem~\ref{thm:antiferro:triangular:general}, we classify all frustrated Ising models on the triangular lattice using Theorem~\ref{thm:main}, and making the construction explicit. 


\begin{figure}[h]
  \centering
  \begin{overpic}[width=13cm]{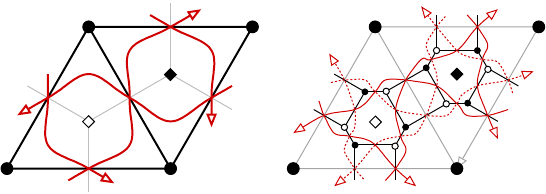} 
\put(37,34.5){\scriptsize $T_1$}
\put(0,14){\scriptsize $T_2$} 
\put(22,1){\scriptsize $T_3$} 
\put(37.5,10){\scriptsize $T_3$}
\put(22,31.5){\scriptsize $-\Js_1$}
\put(4,2){\scriptsize $-\Js_1$}
\put(16,22.5){\scriptsize $-\Js_2$}
\put(7,22.5){\scriptsize $-\Js_3$}
\put(42.7,22.5){\scriptsize $-\Js_3$}
\put(92,34){\scriptsize $\alpha_1$}
\put(98.3,22){\scriptsize $\beta_3$}
\put(51,11){\scriptsize $\alpha_2$}
\put(59,27){\scriptsize $\beta_1$}
\put(74,34){\scriptsize $\beta_1$}
\put(75,-0.5){\scriptsize $\alpha_3$}
\put(90,8){\scriptsize $\alpha_3$}
\put(61,-1){\scriptsize $\beta_2$}
  \end{overpic}
\caption{Left: $G$ is the triangular lattice with the smallest possible fundamental domain; oriented train-tracks of $\T$ are pictured in red. Right: the corresponding graph $\GQ$; oriented train-tracks $\vec{T}$ of $\T_{*,\circ}$, in orientation preserving bijection with those of $\T$, are pictured with full lines; oriented train-tracks $\cev{T}$ of $\T_{*,\bullet}$ are pictured with dotted lines.}\label{fig:dual_graph_bipartite_ter}
\end{figure}
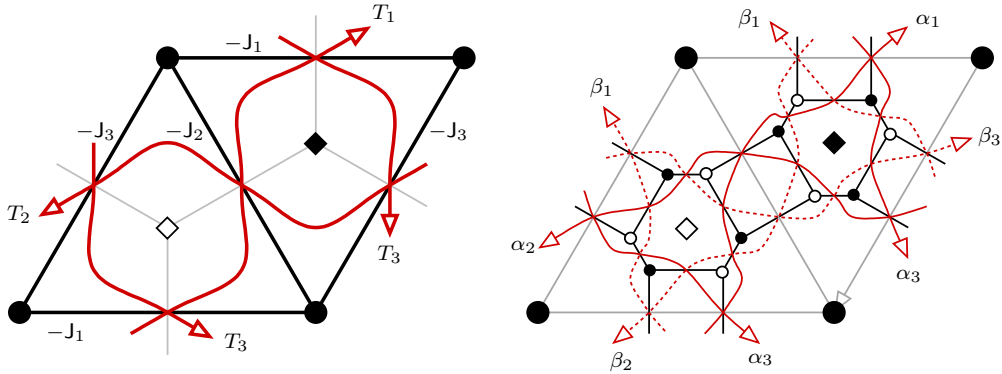

To our knowledge, this is the first attempt to classify frustrated Ising models on the triangular lattice with general coupling constants.
As already mentioned in the introduction of this article, when two of the coupling constants are equal, it was derived by Houtappel~\cite{Hou50} (and confirmed with rigorous mathematical tools in~\cite{AthUel25}) 
that the model with coupling constants $-\beta \Js_1, -\beta \Js_2, -\beta \Js_3$ undergoes a phase transition for some $\beta = \beta_c \in (0,\infty)$ when $\Js_3 < \Js_2= \Js_1$ and on the contrary does not feature a phase transition for $\Js_3 = \Js_2 \leq \Js_1$. 

Our classification is quite different in spirit and more general: we are not restricted to the case where two coupling constants are equal in strength. We partition the set of all coupling constants $\Js \in \RR_+^3$ into one-dimensional families $(\Js_1(k), \Js_2(k), \Js_3(k))_{k^2 \in [0,1]}$. 
These one-dimensional families are of three types, with three possible kinds of spectral curves, thus featuring different properties.

These families are in general not of the form $(\beta \Js_1, \beta \Js_2, \beta \Js_3)_{\beta \in \RR_+}$ (except in the isotropic case $\Js_1 = \Js_2 = \Js_3$, see Proposition~\ref{prop:antiferro:triangular:isotropic}).
As we will see in Remark~\ref{rem:extends}, even though our classification is quite different in spirit from that of Houtappel, it is coherent and extends it in some sense, since we prove that the spectral curve associated with $(\beta \Js_1, \beta \Js_2, \beta \Js_3)$ is 
\begin{itemize}
\item of genus $1$ for all values of $\beta \in \RR_+$ when $\Js_3 < \Js_2 = \Js_1$ 
\item of genus $1$ for all $\beta \neq \beta_c$ and of genus $0$ at $\beta = \beta_c \in (0,\infty)$ when $\Js_3 = \Js_2 < \Js_1$.
\end{itemize}

\subsubsection{Notation and first result}

Let us start by recalling the labeling of the oriented train-tracks, and fix notation. The dual graph $G^*$ of $G$ is the hexagonal lattice, it is bipartite. As in Section~\ref{subsec:CaseB}, see also Section~\ref{sec:Necessary_III}, the set of train-tracks $\T$ of $G^*$ can be consistently oriented so that white vertices of $G^*$ are on the left, see Figure~\ref{fig:dual_graph_bipartite_ter} (left). The fundamental domain $G_1$ has three oriented train-tracks denoted by $\{T_1,T_2,T_3\}$. Moreover, the set of oriented train-tracks $\vec{\T}$ of $\GQ$ is split into $\vec{\T}_{*,\circ}\sqcup\vec{\T}_{*,\bullet}$, where $\vec{\T}_{*,\circ}$ is in orientation preserving bijection with $\T$. The fundamental domain $\GQ_1$ contains the three train-tracks $\{\vec{T}_1,\vec{T}_2,\vec{T}_3\}$ of $\vec{\T}_{*,\circ}$, and the three train-tracks $\{\cev{T}_1,\cev{T}_2,\cev{T}_3\}$ of $\vec{\T}_{*,\bullet}$, see Figure~\ref{fig:dual_graph_bipartite_ter} (right). Suppose that angles $(\hat{\alpha}_{T_i})_{i=1}^3$ are assigned to the oriented train-tracks $(T_i)_{i=1}^3$. Using the orientation preserving bijection between $\T$ and $\vec{\T}_{*,\circ}$, and the fact that $\alpha_{\cev{T}}=\alpha_{\vec{T}}+\rho$, we have for every $i\in\{1,2,3\}$,
\[
\alpha_{\vec{T}_i}={\hat{\alpha}}_{T_i},\quad \alpha_{\cev{T}_i}=\hat{\alpha}_{T_i}+\rho.
\]
Suppose that the parameters $k,\mapalpha,\rho,t$ satisfy Necessary Conditions IV, and that $(\rho,t)\in\{(\frac{1}{2},0),(\frac{1}{2},\frac{1}{2}),(\frac{1}{2}+\frac{\tau}{2},\frac{1}{2})\}$. Then, in both cases, by definition of the set $\Acal_{\rho,\ell_t/2}$, the angle map of the train-tracks of $\vec{\T}_{*,\circ}$, that is of $\T$, takes values in $\RR\,[\Lambda]$. We know prove the following.

%
%
%

\begin{prop}\label{prop:triangular_always_ising}
Consider the triangular lattice $G$ with fundamental domain $G_1$ as in Figure~\ref{fig:dual_graph_bipartite_ter}. Consider Fock's dimer model on the graph $\GQ$ satisfying Necessary Conditions IV, and suppose that $(\rho,t)\in\{(\frac{1}{2},0),(\frac{1}{2},\frac{1}{2}),(\frac{1}{2}+\frac{\tau}{2},\frac{1}{2})\}$.
Then, for every choice of angle map $\mapalpha\in\Acal_{\rho,\,\ell_t/2}=\{\mapalpha\,:\, \alpha_{\vec{T}_1},\alpha_{\vec{T}_2},\alpha_{\vec{T}_3}\in\RR~[\Lambda]\}$, Fock's dimer model is gauge equivalent to an Ising-dimer model. 
\end{prop}
\begin{proof}
By Theorem~\ref{thm:main}, we need to prove that Necessary Conditions V are satisfied, see Definition~\ref{def:necessary_V}. Using the notation of Figure~\ref{fig:primal_dual} (right), this amounts to showing that, for every face of $\GQ$ corresponding to a vertex $v$ of $G$, 
$
\prod_{i=1}^{|v|}\sn(\betath_{i}-\alphath_{i})\in \RR^-.
$
The computation of the product is exactly the computation~\eqref{equ:useful_computation} in the case where $|v|=6$,  
$T_1,T_3'\mapsto T_1$, $T_2,T_1'\mapsto T_2$, $T_3,T_2'\mapsto T_3$, see Figure~\ref{fig:dual_graph_bipartite_bis}, so that we obtain:
\begin{align*}
\prod_{i=1}^{|v|}\sn(\betath_{i}-\alphath_{i})&=
(-1)^{|v|/2}\prod_{i=1}^{|v|/2}\sn(\hat{\alpha}^k_{T_i'}-\hat{\alpha}^k_{T_i})\sn(\hat{\alpha}^k_{T_i}-\hat{\alpha}^k_{T_{i+1}'})\\
&=-\sn^2(\alpha_{\vec{T}_2}^k-\alpha_{\vec{T}_1}^k)\sn^2(\alpha_{\vec{T}_1}^k-\alpha_{\vec{T}_3}^k)\sn^2(\alpha_{\vec{T}_3}^k-\alpha_{\vec{T}_2}^k) \in \RR^-.\qedhere
\end{align*}
%
\end{proof}

\subsubsection{Two natural families of frustrated Ising models}\label{subsec:txo_examples}
Let us give concrete examples of families of frustrated Ising models that can be parameterized by Fock's dimer model. The proof of both propositions is postponed to Section~\ref{subsec:proofs}.

\begin{prop}\label{prop:antiferro:triangular:isotropic}
Fock's dimer model with parameters $k^2\in[0,1), \rho=\frac{1}{2}+\frac{\tau}{2}$, angle-map  
$\alpha_{\vec{T}_1} = 0, \alpha_{\vec{T}_2} = \frac{1}{3}, \alpha_{\vec{T}_3} = \frac{2}{3}$ is gauge equivalent to the Ising-dimer model arising from the isotropic frustrated Ising model with coupling constants $-\Js$ given by 
$$
\sinh(2\Js) = \frac{1}{k}\ds(2K/3)
$$
As $k$ runs from $0$ to $1$, this parameterizes all possible coupling constants $\Js \in (0,\infty)$.
This family of models features no \emph{algebraic phase transition} for finite values of the coupling constants since the spectral curve remains of genus one, however the model is \emph{critical} when $k=0$, corresponding to infinite values of the coupling constants, in the sense that the spectral curve has genus $0$. 
\end{prop}
The amoeba and real locus of this model are provided on the right of Figure~\ref{fig:triangular_lattice_curves} for some $k^2 \in (0,1)$.
This result sheds a new light on the physics literature on the isotropic antiferromagnetic Ising model on the triangular lattice, and in particular on~\cite{Ste70a,Ste70b} 
which obtains that at zero temperature, the model is \emph{critical} in the sense that correlation have polynomial decay. 



We also provide an example of a family of anisotropic frustrated Ising model featuring an algebraic phase transition.
\begin{prop}\label{prop:triangular:anisotropic}
    Fock's “dimer” model with $\rho = \frac{1}{2}$, $t=0$, $k^2 \in [0,1)$ and $\alpha_{\vec{T}_1} = 0, \alpha_{\vec{T}_2} = \frac{1}{6}, \alpha_{\vec{T}_3} = \frac{1}{3}$ is gauge equivalent to the anisotropic frustrated model with coupling constants $\eps\Js$ such that the frustration function $\delta_f = -1$ at every face and with
    $$
      -\Js_{e_1} = -\Js_{e_2} = - \frac12\argsinh(\sc(2K/3)) \quad ; \quad -\Js_{e_3} = - \frac12\argsinh(\sc(K/3)).
    $$
    This family of models features an \emph{algebraic phase transition} since at $k=0$ the coupling constants have well-defined real positive values and the spectral degenerates from genus $1$ to genus $0$. 
\end{prop}
The amoeba and real locus of this model are provided in the middle of Figure~\ref{fig:triangular_lattice_curves} for some $k^2 \in (0,1)$.

\subsubsection{Classification of the frustrated Ising models}\label{subsec:frustrated_triang}
Consider, on the one hand, the Ising model on the triangular lattice with coupling constants $\eps_1\Js_1$, $\eps_2\Js_2,\eps_3\Js_3$, and the corresponding Ising-dimer model on $\GQ$. 
Denote by: 
\[
  s_i: = \sinh(2\Js_i)>0 \quad ; \quad i \in \{1,2,3\}.
\]

Consider, on the other hand, Fock's Ising-dimer model with parameters $k,\mapalpha,\rho,t$ satisfying Necessary Conditions IV, and suppose that $(\rho,t)\in\{(\frac{1}{2},0),(\frac{1}{2},\frac{1}{2}),(\frac{1}{2}+\frac{\tau}{2},\frac{1}{2})\}$.
Denote by
\[
  \gamma_1  := \alphah_{\vec{T}_3} - \alphah_{\vec{T}_1}, \gamma_2 := \alphah_{\vec{T}_1} - \alphah_{\vec{T}_2}, \gamma_3 := \alphah_{\vec{T}_2} - \alphah_{\vec{T}_3}.
\]
By Theorem~\ref{thm:main} and Proposition~\ref{prop:triangular_always_ising}, for every choice of angle map $\mapalpha\in\Acal_{\rho,\,\ell_t/2}$, Fock's dimer model is always gauge equivalent to an Ising-dimer model.
Let us now describe the maps induced on the parameters of the models. Let
\[
  \Gamma = \Bigl\{(k,\gamma_1, \gamma_2, \gamma_3) \in [0,1) \times (\RR/4K\ZZ)^3 : \sum_{i=1}^3 \gamma_i = 0\Bigr\},
\]
and
\[
  \Gamma_\pm = \Bigl\{(k,\gamma_1, \gamma_2, \gamma_3) \in \Gamma: \prod_{i=1}^3 \sn(\gamma_i +K) \in \RR_\pm\Bigr\}.
\]
Let $\cS = (\RR^+)^3$ be the set of possible values for $(s_1,s_2,s_3)$, and let 
$\phi_1, \phi_2, \phi_3: \Gamma \to \cS$ be defined by
\[
  \begin{aligned}
    \phi_1(k,\gamma_1, \gamma_2, \gamma_3) &= \Bigl\{\frac{1}{k'|\sc(\gamma_1)|}, \frac{1}{k'|\sc(\gamma_2)|}, \frac{1}{k'|\sc(\gamma_3)|}\Bigr\}\\
    \phi_2(k,\gamma_1, \gamma_2, \gamma_3) &= \Bigl\{\frac{1}{|\sc(\gamma_1)|}, \frac{1}{|\sc(\gamma_2)|}, \frac{1}{|\sc(\gamma_3)|}\Bigr\}\\
    \phi_3(k,\gamma_1, \gamma_2, \gamma_3) &= \Bigl\{\frac{|\ds(\gamma_1)|}{k}, \frac{|\ds(\gamma_2)|}{k}, \frac{|\ds(\gamma_3)|}{k}\Bigr\}.
  \end{aligned}
\] 
By Theorem~\ref{thm:main} and Proposition~\ref{prop:triangular_always_ising} (see the proof of the following theorem for details), Fock's dimer model with parameters $k,\mapalpha,\rho,t$ is gauge equivalent to an Ising model with absolute value coupling constants $(\Js_1,\Js_2,\Js_3)$, given through $(s_1,s_2,s_3)$ by
\[
(s_1,s_2,s_3)=\phi_i(k,\gamma_1,\gamma_2,\gamma_3),
\]
with 
\[
\begin{cases}
i=1, \text{ when } (\rho,t)=(\frac{1}{2},0) \text{ and } (k,\gamma_1,\gamma_2,\gamma_3) \in \Gamma_\pm\\
i=2, \text{ when } (\rho,t)=(\frac{1}{2},\frac{1}{2}) \text{ and }
(k,\gamma_1,\gamma_2,\gamma_3) \in \Gamma_\pm \\
i=3, \text{ when } (\rho,t)=(\frac{1}{2}+\frac{\tau}{2},\frac{1}{2}) \text{ and } (k,\gamma_1,\gamma_2,\gamma_3) \in \Gamma.
\end{cases} 
\]
Moreover, in Case $i$, with $i\in\{1,2\}$, the Ising model is non-frustrated, resp. frustrated, when $(k,\gamma_1,\gamma_2,\gamma_3)\in\Gamma_-$, resp. $\Gamma_+$. In Case 3., it is frustrated. 

%

The next theorem is the main result of this section. We restrict to the frustrated case, \emph{i.e.}, to the case where $\eps_i=-$ for $i\in\{1,2,3\}$. Given coupling constants $-\Js_1,-\Js_2,-\Js_3$, we identify the parameters of Fock's corresponding dimer models. Before stating our result, we need some additional notation.
For $(s_1,s_2,s_3) \in \cS$, let $(s_{(1)}, s_{(2)}, s_{(3)})$ denote the associated ordered triple, that is 
\[
  \{s_1,s_2,s_3\} = \{s_{(1)}, s_{(2)}, s_{(3)}\} \text{ and } s_{(1)} \geq s_{(2)} \geq s_{(3)}.
\]
Let us also introduce, for $s_1, s_2 \in \RR^+$,
$$
  S(s_1,s_2) = \max\left(\frac{s_2 s_1 -1}{s_2+s_1 },0\right) \quad ; \quad T(s_1,s_2) =  s_2 s_1 \frac{\sqrt{s_2^2+1}\sqrt{s_1 ^2+1}-1}{s_2^2\sqrt{s_1 ^2+1}+s_1 ^2\sqrt{s_2^2+1}}.
$$
Then we partition $\cS$ into 
\[
  \begin{aligned}
    \cS_1 &= \{(s_1, s_2,s_3) \in \cS: s_{(3)} \leq S(s_{(1)}, s_{(2)})\}\\
    \cS_2 &= \{(s_1, s_2,s_3) \in \cS: S(s_{(1)}, s_{(2)}) \leq s_{(3)} < T(s_{(1)}, s_{(2)})\}\\
    \cS_3 &= \{(s_1, s_2,s_3) \in \cS: T(s_{(1)}, s_{(2)}) < s_{(3)}\}.
  \end{aligned}
\]
Note that this is a partition of $\cS$ except for the negligible set of points points with $s_{(3)} = S(s_{(1)}, s_{(2)})$ which belong to $\cS_1 \cap \cS_2$ and those with $s_{(3)} < T(s_{(1)}, s_{(2)})$ which belong to none of these sets.
We are now able to state the main result of this section.
\begin{thm}\label{thm:antiferro:triangular:general}
  For all coupling constants $\Js_1, \Js_2, \Js_3$, there exists $k,\mapalpha,\rho,t$ satisfying Necessary Conditions IV, with $(\rho,t)\in\{(\frac{1}{2},0),(\frac{1}{2},\frac{1}{2}),(\frac{1}{2}+\frac{\tau}{2},\frac{1}{2})\}$,
  such that the Ising-dimer model on $\GQ$ arising from the frustrated Ising dimer model with coupling constants $-\Js_1, -\Js_2, -\Js_3$ is gauge equivalent to Fock's dimer model with parameters $k,\mapalpha,\rho,t$.
  More precisely, 
  \[
    \phi_1(\Gamma_+) = \cS_1 \quad ; \quad \phi_2(\Gamma_+) = \cS_2 \quad ; \quad \phi_3(\Gamma) = \cS_3.
  \]
  This is a partition of all possible frustrated Ising models on the triangular lattice in three disjoint families of parameters featuring different properties, namely the properties of the spectral curve and the existence (or not) of an algebraic phase transition.
\end{thm}

\begin{rem}\label{rem:extends}\leavevmode
\begin{enumerate}
\item The proof of Theorem~\ref{thm:antiferro:triangular:general} is postponed to Section~\ref{subsec:proofs}
\item In the two cases corresponding to $\rho = \frac{1}{2}$, $t$ is not located on the same part of the real locus.
In the third case $\rho = \frac{1}{2} + \frac{\tau}{2}$, the curve has different properties since the zeros and poles are located on two different components of the real locus.
This explains why we expect the corresponding sets of Ising models to be disjoint.

However, when $s_{(3)} = S(s_{(1)}, s_{(2)})$, two Fock's dimer model correspond to this Ising model, one with $\rho = \frac{1}{2}, t=0$ and one with $\rho = \frac{1}{2}, t= \frac{1}{2}$.
We will see in the proof that in both cases, $k=0$: the genus one curve degenerates into a genus zero curve, and in this case the effect of the parameter $t$ disappears.
The same thing happens for the non-frustrated case.
\item 
The maps $\phi_i$ are not bijective due to the absolute values in the definition, but on well-chosen domains of $\Gamma$ where the signs of $\gamma_1, \gamma_2, \gamma_3$ is controlled it becomes a bijection.
This is contained in the proof but rather painful to write down explicitly.
\item Our classification extends the classification of Eggarter~\cite{Egg75}, in the following way. 
  By definition of $\S_3$, the models with $\Js_3 = \Js_2 < \Js_1$ are such that $(s_1,s_2,s_3) \in \S_3$, hence the spectral curve associated with $\beta \Js$ is of genus $1$ for all values of $\beta >0$.
  On the contrary, for $\Js_3 < \Js_2 = \Js_1$, one can show by direct computation that there exists a unique $\beta_c \in \RR_+$ such that, denoting $s_i(\beta) = \sinh(2\beta \Js_i)$, we have $(s_1(\beta_c), s_2(\beta_c), s_3(\beta_c)) \in \S_1 \cap \S_2$, which as we will see in the proof corresponds to $k=0$, i.e., the associated spectral curve has genus $0$.
\item Proposition~\ref{prop:antiferro:triangular:isotropic} can be seen as a corollary of this classification results. 
In fact, if $\Js_1 = \Js_2 = \Js_3$, then $s = s_{(1)} = s_{(2)} = s_{(3)}$ thus $T(s_{(1)},s_{(2)}) < s_{(3)}$ and by definition $(s_1,s_2,s_3) \in \S_3$.
Thus the parameterization of the isotropic Ising model must satisfy $(\rho, t) = (\frac{1}{2} + \frac{\tau}{2})$ and one recover the values of $k$ and of the angle map by solving
$$
  s = s_i = \frac{|\ds(\gamma_i)|}{k} \quad ; \quad \gamma_1 + \gamma_2 + \gamma_3 = 0.
$$
\end{enumerate}
\end{rem}

\subsubsection{Proof of the results}\label{subsec:proofs}


Before proving Proposition~\ref{prop:antiferro:triangular:isotropic}, Proposition~\ref{prop:triangular:anisotropic}, and Theorem~\ref{thm:antiferro:triangular:general}, we make a few general considerations.
Consider Fock's Ising-dimer model with parameters $k,\mapalpha,\rho,t$ satisfying Necessary Conditions IV and such that $(\rho,t)\in\{(\frac{1}{2},0),(\frac{1}{2},\frac{1}{2}),(\frac{1}{2}+\frac{\tau}{2},\frac{1}{2})\}$. Recall that by Proposition~\ref{prop:triangular_always_ising}, for every choice of angle map $\mapalpha\in\Acal_{\rho,\,\ell_t/2}$, Fock's dimer model always corresponds to an Ising model.

Consider a face of $\GQ$ corresponding to a black vertex of $G^*$, then recalling the notation $\alpha_i,\beta_i$ for the train-tracks of Figure~\ref{fig:primal_dual} (left), and using our labeling of the oriented train-tracks $T_1,T_2,T_3$ of $G^*$, see Figure~\ref{fig:dual_graph_bipartite_ter}, we have 
\[
\alpha_1 = \alpha_{\vec{T}_1}~;~ \alpha_2 = \alpha_{\vec{T}_2} ~;~ \alpha_3 = \alpha_{\vec{T}_3}.
\]
Recall the notation $\gamma_1  := \alphah_{3} - \alphah_1, \gamma_2 := \alphah_{1} - \alphah_2$, $\gamma_3 := \alphah_{2} - \alphah_3$; we have $\gamma_1 + \gamma_2 + \gamma_3 = 0$. 
Moreover, by Equation~\eqref{equ:angle_relation}, the lifts of the angles $(\beta_i)$ are fixed so that $\beta_i=\alpha_{i-1}+\rho$, with $\rho=\frac{j}{2}+i\frac{\tau}{2}$, implying that, for $i \in \{1,2,3\}$,
\begin{equation}\label{equ:beta_alpha}
\betah_i-\alphah_i=\gamma_i +jK+i\ell K'.
\end{equation}
%

Consider now an Ising dimer model with coupling constants $-\Js_1, -\Js_2, -\Js_3$, and the corresponding Ising-dimer model on $\GQ$.
Then, using the notation of Figure~\ref{fig:dual_graph_bipartite_ter} and Equation~\eqref{equ:beta_alpha},
by Theorem~\ref{thm:main}, the two dimer models are gauge equivalent if and only the following is satisfied,

1. Case $(\rho,t)\in\{(\frac{1}{2},0),(\frac{1}{2},\frac{1}{2})\}$.
\begin{enumerate}
\item[$\bullet$] Equation for $(s_i)_{i=1}^3=(\sinh(2\Js_i))_{i=1}^3$, see the second row of Table~\ref{table3}, 

$\circ$ Case $t=0$. 
\begin{equation}\label{eq:system:bis} 
\begin{split}
\forall i \in \{1,2,3\},&\ s_i^2= \sc^2(\betah_i - \alphah_i|k) =\sc^2(\gamma_i +K|k)\\ \Leftrightarrow & \ (s_1,s_2,s_3) = \phi_1(k,\gamma_1,\gamma_2,\gamma_3).
\end{split}
\end{equation}
$\circ$ Case $t=\frac{1}{2}$. 
\begin{equation}\label{equ:square_face_triangular_2}
\begin{split}
\forall i \in \{1,2,3\},&\ s_i^2 = (k')^2 \sc^2(\betah_i - \alphah_i|k) = (k')^2\sc^2(\gamma_i +K|k)\\
\Leftrightarrow & \ (s_1,s_2,s_3) = \phi_2(k,\gamma_1,\gamma_2,\gamma_3),  
\end{split}
\end{equation}
where in the equivalences we used identities for $\sc(u+K)$ given by see~\cite[(2.2.17)-(2.2.18)]{Lawden}.
\item[$\bullet$] Sign condition at every face of $\GQ$ corresponding to a dual vertex $f$ of $G$ required to have a frustrated Ising model, see fourth row of Table~\ref{table3},
\begin{align}
&\prod_{i=1}^3 \sn(\betah_i-\alphah_i) \in \RR^+
\Leftrightarrow (k,\gamma_1,\gamma_2,\gamma_3) \in \Gamma_+\label{eq:sign:dual:triangular},
\end{align} 
where we observe that the sign condition at the face corresponding to the black vertex and the white vertex is the same.
\end{enumerate}

2. Case $(\rho,t)=(\frac{1}{2}+\frac{\tau}{2},\frac{1}{2})$. Recall that the model is always frustrated, so that the two models are gauge equivalent if and only if, 
{\small 
\begin{equation}\label{eq:system}
\begin{split}
\forall i\in\{1,2,3\} &\ s_i^2=\frac{(k')^2}{k^2}\nc(\Re(\betah_i-\alphah_i)) 
=\frac{(k')^2}{k^2}\nc^2(\alphah_{i-1}-\alphah_i+K)=\frac{1}{k^2}\ds^2(\alphah_{i-1}-\alphah_i) \\
\Leftrightarrow & \ (s_1,s_2,s_3) = \phi_3(k,\gamma_1,\gamma_2,\gamma_3),
\end{split}
\end{equation}
}
where in the second equality of the first line we used that $\alphah_{i-1}-\alphah_i$ is real and that $\Re(K+iK')=K$, and in the third we used~\cite[(2.2.18)]{Lawden}. 

We first prove the (easier) Propositions~\ref{prop:antiferro:triangular:isotropic} and \ref{prop:triangular:anisotropic} and then prove the theorem. 

\begin{proof}[proof of Proposition~\ref{prop:antiferro:triangular:isotropic}]
  Let $\rho = \frac{1}{2} + \frac{\tau}{2}$, $t = \frac{1}{2}$, $\alpha_1 = 0, \alpha_2 = 1/3, \alpha_3 = 2/3$ as in the statement of the proposition.
  We are in Point 2. of the above discussion, so the model is frustrated and we only have to check the value of the coupling constants.
  We first compute 
  $$
    \gamma_1 = \frac{4K}{3} \quad ; \quad \gamma_2 = \gamma_3 = -\frac{2K}{3}.
  $$
  By Equation~\eqref{eq:system:bis}, the coupling constants for the three types of edges satisfy
  \begin{align*}
    s_1^2 &= \frac{1}{k^2}\ds^2(4K/3) = \ds^2(-2K/3) = \ds^2(2K/3)\\
    s_2^2 = s_3^2 &= \frac{1}{k^2}\ds^2(-2K/3) = \ds^2(2K/3) 
  \end{align*}
  Hence, Fock's “dimer” model corresponds to the frustrated Ising model with coupling constants
  $$
    - \sinh^2(2\Js_e) = -\frac{1}{k^2}\ds^2(2K/3).
  $$
  On the one hand, as $k \to 0$,
  \[
    \frac{1}{k}\ds(2K/3) \sim \frac{1}{k\sin(\pi/3)} \overset{k \to 0}{\longrightarrow} \infty.
  \]
  Moreover, as $k \to 1$, since $\ds$ is decreasing on $[0,K]$ and by~(2.4.10) of~\cite{Lawden}
  \[
    \ds(2K/3) \leq \ds(K/2) = \sqrt{k'(1+k')} \overset{k \to 1}{\longrightarrow} 0. 
  \]
  Hence as $k$ runs from $0$ to $1$, $\frac{1}{k}\ds(2K/3)$ runs from $0$ to $\infty$ so $\Js$ runs from $0$ to $\infty$.
\end{proof}

\begin{proof}[Proof of Proposition~\ref{prop:triangular:anisotropic}]
  Let $\rho = \frac{1}{2}$, $t=0$, $\alpha_1 = 0$, $\alpha_2 = 1/6$, $\alpha_3 = 1/3$ as in the statement of the proposition.
  We are in Point 1 of the above discussion, so we only have to check the sign at dual faces and the values of the coupling constants.
  We first compute 
  $$
    \gamma_1 = \frac{2K}{3} \quad ; \quad \gamma_2 = \gamma_3 = -\frac{K}{3}
  $$
  so $(k, \gamma_1,\gamma_2, \gamma_3) \in \Gamma_+$ and the model is frustrated.
  
  Let us now consider the coupling constants.
  By Equation~\eqref{eq:system:bis}, the coupling constants for the three types of edges satisfy 
  \begin{align*}
    s_1^2 &= \sc^2(2K/3+K) = \sc^2(K/3)\\
    s_2^2 =  s_3^2 &= \sc^2(-K/3+K) = \sc^2(2K/3)
    \qedhere
  \end{align*}
\end{proof}

We finally prove Theorem~\ref{thm:antiferro:triangular:general}.
Let us fix $\Js_1, \Js_2, \Js_3 > 0$ and recall that $s_i = \sinh(2\Js_i)$.
Upon renumbering, we can assume that $s_1 \geq s_2 \geq s_3$ i.e. $s_{(i)} = s_i$ for every $i \in \{1,2,3\}$.
We study when the frustrated Ising dimer model with coupling constants $-\Js_1, -\Js_2, -\Js_3$ corresponds to Fock's dimer model with parameters $k,\mapalpha,\rho,t$.
We split the proof in two cases.

\begin{proof}[Proof of the $(\rho,t)\in\{(\frac{1}{2},0),(\frac{1}{2},\frac{1}{2})\}$ case of  Theorem~\ref{thm:antiferro:triangular:general}.]\leavevmode
To deal simultaneously with the two cases $t=0$ and $t=\frac{1}{2}$, let us introduce the notation
$$
  \dt
  =
  \left\{
  \begin{aligned}
    0 &\text{ if }t=0\\
    1 &\text{ if }t=\frac{1}{2}
  \end{aligned}
  \right.
$$
To prove this part of the theorem, we need only show that the system of equations~\eqref{eq:system:bis} and~\eqref{eq:sign:dual:triangular} can be solved if and only if $(s_1,s_2,s_3) \in \cS_{\delta(t)+1}$.

Let us recall the system of equations~\eqref{eq:system:bis} and~\eqref{eq:sign:dual:triangular}, that we want to solve when seen as equations of the four variables $\gamma_i  \in \RR/(4K \ZZ), k \in [0,1)$:
\begin{equation*}
\begin{aligned}
  s_i^2= (k')^{2\dt}\sc^2(\gamma_i +K|k) \quad ; \quad \prod_{i=1}^3 \sn(\gamma_i  +K) \in \RR^+ \quad ; \quad \sum_{i=1}^3 \gamma_i = 0.
\end{aligned}
\end{equation*}
Writing $\gamma_3 = -\gamma_1 - \gamma_2$, this is equivalent to the set of equations of the three variables $\gamma_1, \gamma_2 \in \RR/(4K\ZZ), k \in [0,1)$:
\begin{equation}\label{eq:system:bisbis}
\begin{aligned}
  &s_1^2= (k')^{2\dt}\sc^2(\gamma_1 +K|k); \  s_2^2= (k')^{2\dt}\sc^2(\gamma_2 +K|k); \  s_3^2= (k')^{2\dt}\sc^2(-\gamma_1-\gamma_2 +K|k)\\
  &\sn(\gamma_1  +K)\sn(\gamma_2+K)\sn(-\gamma_1 - \gamma_2+K)>0.
\end{aligned}
\end{equation}
Note that since these equations are left unchanged by translating $\gamma_1$ or $\gamma_2$ by $2K$, we can consider that $\gamma_1, \gamma_2 \in \RR / (2K\ZZ)$.
For every $k\in[0,1)$, the function $\sc : (0,K) \to (0, +\infty)$ is invertible, hence we can define $d_i:= \argsc((k')^{-\dt}s_i|k) \in (0,K)$, and we must have, for $i \in \{1,2\}$,
$$
\gamma_i + K = \eps_id_i \ [2K],\text{ with $\eps_i = \pm 1$}.
$$
Since $\gamma_i$ is defined on $\RR/(2K \ZZ)$, we loose no generality at this point. 
Using the fact that the function $\sn$ is odd and~\cite[(2.2.12-2.2.17)]{Lawden}, the sign Condition of~\eqref{eq:system:bisbis} can be rewritten as 
$$
\sn(\gamma_2 +K)\sn(\gamma_1+K)\cd(\gamma_2 +\gamma_1+2K) < 0 \iff \eps_2\eps_1 \cd(\eps_2 d_2 +\eps_1 d_1) <0.
$$
Observing that the equations are the same for $(\eps_2,\eps_1 )$ and $(-\eps_2,-\eps_1 )$, we can restrict to the case $\eps_2 =1, \eps_1  = \eps = \pm 1$. 
Hence, the system of equations~\eqref{eq:system:bisbis} is
equivalent to solving for $\eps \in \{\pm 1\}$, $k' \in (0,1]$,
  \begin{equation*}
    \begin{aligned}
    s_3 ^2 
    = (k')^{2\dt}\sc^2\big[d_2 +\eps d_1-K\big|k\big]
    = \frac{(k')^{2\dt-2}}{\sc^2(d_2 +\eps d_1)} \quad ; \quad \eps \cd(d_2 +\eps d_1)<0.
    \end{aligned}
  \end{equation*}
where we used~\cite[(2.2.17)-(2.2.18)]{Lawden}. 
This is equivalent to
\begin{equation}\label{eq:triangular:cns:bis}
    \sc(d_2 +\eps d_1) = \pm \frac{(k')^{\dt-1}}{s_3 } \quad ; \quad \eps \cd(d_2 +\eps d_1)<0.
  \end{equation}
Let us now solve this equation as a function of $k'\in(0,1]$. 
  On the one hand, the addition formula for $\sc$ is, see~\cite[(2.4.1)-(2.4.3)]{Lawden},
\begin{equation}\label{equ:addition_formula_sc}
    \sc(d_2 +\eps d_1) = \frac{\sn(d_2 )\cn(d_1)\dn(d_1)+\eps\sn(d_1)\cn(d_2 )\dn(d_2 )}
    {\cn(d_2 )\cn(d_1)-\eps\sn(d_2 )\sn(d_1)\dn(d_2 )\dn(d_1)}.
\end{equation}
  On the other hand, using Equations~\cite[(2.1.9)-(2.1.11)]{Lawden}, we obtain (omitting the variables)
  $$
    \sn^2 = \frac{\sc^2}{\sc^2+1} \quad ; \quad \cn^2 = \frac{1}{\sc^2+1} \quad ; \quad \dn^2 = \frac{(k')^2\sc^2+1}{\sc^2+1}.
  $$
  Since for $i\in\{1,2\}$, $d_i\in(0,K)$, we can take the square roots in the above identities giving
{\small 
  $$
\begin{aligned}
    \sn(d_i) 
    &= \frac{\sc(d_i)}{\sqrt{\sc^2(d_i)+1}} 
    = \frac{ s_i}{\sqrt{s_i^2+(k')^{2\dt}}},\quad 
    \cn(d_i) 
    &=\frac{1}{\sqrt{\sc^2(d_i)+1}} 
    = \frac{(k')^{\dt}}{\sqrt{s_i^2+(k')^{2\dt}}} \\
    \dn(d_i) 
    &= \frac{\sqrt{(k')^2\sc^2(d_i)+(k')^{2\dt}}}{\sqrt{\sc^2(d_i)+(k')^{2\dt}}}
    = \frac{\sqrt{(k')^2s_i^2+1}}{\sqrt{s_i^2+1}}.
\end{aligned}
  $$
}  
For $i\in\{1,2\}$, define $g_i(x)=\frac{\sqrt{x^{2(1-\dt)}s_i^2+1}}{\sqrt{s_i^2+x^{2\dt}}}$, then we have  
\begin{equation}\label{eq:sc(d1+d2)}
  {\small 
    \begin{aligned}
      \sc(d_2 +\eps d_1) 
      &= 
      \frac{
        \frac{ s_2}{\sqrt{s_2^2+(k')^{2\dt}}}
        \frac{(k')^{\dt}}{\sqrt{s_1 ^2+(k')^{2\dt}}}
        \frac{\sqrt{(k')^2s_1 ^2+1}}{\sqrt{s_1 ^2+1}}
        +\eps
        \frac{ s_1 }{\sqrt{s_1 ^2+(k')^{2\dt}}}
        \frac{(k')^{\dt}}{\sqrt{s_2^2+(k')^{2\dt}}}
        \frac{\sqrt{(k')^2s_2^2+1}}{\sqrt{s_2^2+1}}
      }
      {
        \frac{(k')^{\dt}}{\sqrt{s_2^2+(k')^{2\dt}}}
        \frac{(k')^{\dt}}{\sqrt{s_1 ^2+(k')^{2\dt}}}
        - \eps
        \frac{ s_2}{\sqrt{s_2^2+(k')^{2\dt}}}
        \frac{ s_1 }{\sqrt{s_1 ^2+(k')^{2\dt}}}
        \frac{\sqrt{(k')^2s_2^2+1}}{\sqrt{s_2^2+1}}
        \frac{\sqrt{(k')^2s_1 ^2+1}}{\sqrt{s_1 ^2+1}}
      }\\
      &=
       \frac{s_2 g_1 (k')+\eps s_1 g_2 (k')}{1-\eps s_2 s_1  g_2 (k')g_1 (k')}.
    \end{aligned}}
  \end{equation}

  Let us start with the case where $\eps = 1$.
Since $d_2 +d_1 \in [0,2K]$, $\sn(d_2 +d_1) >0$ so $\sc(d_2 +d_1)$ and $\cd(d_2 +d_1)$ have the same sign
, and solving Equation~\eqref{eq:triangular:cns:bis} is equivalent to solving the unique equation
\begin{equation*}
    \sc(d_2 +d_1)
    = -\frac{1}{(k')^{1-\dt}s_3 },
\end{equation*} 
  which by Equation~\eqref{eq:sc(d1+d2)} is equivalent to solving 
  \begin{equation}\label{eq:cns:s3}
    s_3  
    = \frac{1}{(k')^{1-\dt}}
      \frac
      {
        s_2 s_1 
        g_2 (k')
        g_1 (k')
        -
        1
      }
      {
        s_2 g_1 (k')
        +
        s_1  g_2 (k')
      }
      =: f(k')
      .
  \end{equation}
  When $\dt = 0$, one can show by a careful analysis of the derivative that $f$ is increasing.
  Since $f(0) < 0$ and $s_3  \geq 0$, the set of possible values for $s_3 $ for which Equation~\eqref{eq:cns:s3} can be solved is 
  $$
    s_3  \in [0, \max(0,f(1))] = \left[0, \max\left(0,\frac{s_2 s_1 -1}{s_2+s_1 }\right)\right] = [0,S].
  $$
  When $\dt = 1$, one can show by a careful analysis of the derivative that $f$ is strictly decreasing. 
  We have 
  $f(0)=\frac{\sqrt{s_2^2+1}\sqrt{s_1 ^2+1}-1}{\frac{s_2}{s_1 }\sqrt{s_1 ^2+1}+\frac{s_1 }{s_2}\sqrt{s_2^2+1}}=s_2 s_1 \frac{\sqrt{s_2^2+1}\sqrt{s_1 ^2+1}-1}{s_2^2\sqrt{s_1 ^2+1}+s_1 ^2\sqrt{s_2^2+1}} = T >0$, and $f(1)=\frac{s_2 s_1  -1}{s_2+s_1 }$. 
  This implies that $S \leq T$.
  On the other hand, an explicit computation shows that $T \leq s_2$ (see below).
  Hence, since $0 \leq s_3  \leq s_2$, the set of possible values for $s_3 $ for which Equation~\eqref{eq:cns:s3} can be solved is $[S,T]$.
  
  To prove that $T \leq s_2$, we first bound
  $$
    \begin{aligned}
    T 
    = \frac{s_2 s_1 \left(\sqrt{s_2^2+1}\sqrt{s_1 ^2+1}-1\right)}{s_1 ^2\sqrt{s_2^2+1}+s_2^2\sqrt{s_1 ^2+1}}
    \leq \frac{s_2\left(\sqrt{s_2^2+1}\sqrt{s_1 ^2+1}-1\right)}{s_1 \sqrt{s_2^2+1}},
    \end{aligned}
  $$
  and then note that 
  \begin{equation}\label{eq:T<s1}
    \begin{aligned}
    &\sqrt{s_2^2+1}\sqrt{s_1 ^2+1}-1 < s_1 \sqrt{s_2^2+1}\\
    &\iff (s_2^2+1)(s_1 ^2+1)-2\sqrt{s_2^2+1}\sqrt{s_1 ^2+1}+1 \leq s_1 ^2(s_2^2+1)\\
    &\iff s_2^2+2 \leq 2\sqrt{s_2^2+1}\sqrt{s_1 ^2+1}\\
    &\iff s_2^4+4s_2^2 + 4 \leq 4(s_2^2+1)(s_1 ^2+1),
    \end{aligned}
  \end{equation}
  which is true since $s_2 \leq s_1 $.

  Let us now consider the case $\eps = -1$. 
  Since $d_2 +d_1 \in (0,2K)$, $\cd(d_2 +d_1) > 0$ so the second equation is always satisfied and solving Equation~\eqref{eq:triangular:cns:bis} is equivalent to solving for $k' \in (0,1]$
\begin{equation*}
    s_3  
    = \pm\frac{1}{(k')^{1-\dt}}
      \frac{1+ s_2 s_1  g_2 (k')g_1 (k')}{s_2 g_1 (k')-\eps s_1 g_2 (k')},
\end{equation*}
Since $s_2 < s_1 $, we observe that $s_2 g_1 (k') \leq s_1 g_2 (k')$ for all $k$ by comparing the squares in both cases $\dt = 0$ and $\dt = 1$.
Hence, solving Equation~\eqref{eq:triangular:cns:bis} is equivalent to solving for $k' \in (0,1]$
\begin{equation*}
    s_3  
    = \frac{1}{(k')^{1-\dt}}
      \frac{1+s_2 s_1  g_2 (k')g_1 (k')}{s_1 g_2 (k')-s_2 g_1 (k')}=: f(k').
\end{equation*}
Now, note that for $x \in [0,1]$,
\begin{equation*}
  \begin{aligned}
    \frac{1}{x^{1-\dt}}\frac{1+ s_2 s_1  g_2 (x) g_1 (x)}{s_1  g_2 (x)-s_2 g_1 (x)}
    &\geq \frac{1}{x^{1-\dt}}\frac{s_2 s_1 g_2 (x) g_1 (x)}{s_1 g_2 (x)}\\
    &= s_2 x^{\dt-1}\frac{\sqrt{x^{2(1-\dt)}s_2^2+1}}{\sqrt{s_2^2+x^{2\dt}}}
    \geq s_2,
  \end{aligned}
\end{equation*}
and since by assumptions $s_3  \leq s_2$ we find no solution in this regime.
\end{proof}

\begin{rem}
  This argument might look strange at first sight since we are excluding solutions with $s_3  \geq s_2$ which seem to exist.
  However, since we wrote $\sc^2(d_2 |k) = s_2^2$, $\sc^2(d_1|k) = s_1 ^2$, we have forced $s_3 ^2 = \sc^2(d_2 +d_1-K)$ which is (by our computations) not compatible with the condition $s_3  \leq s_2 \leq s_1 $.
\end{rem}

\begin{proof}[Proof of the $(\rho,t)=(\frac{1}{2}+\frac{\tau}{2},\frac{1}{2})$ case of Theorem~\ref{thm:antiferro:triangular:general}.]\leavevmode
By Point 2. of the discussion above, to prove this part of the theorem, we only need to show that the system of equations~\eqref{eq:system} can be solved if and only if $(s_1,s_2,s_3) \in \cS_{3}$.
Let us recall the system of equations~\eqref{eq:system} for the four variables $\gamma_i  \in \RR/(4K \ZZ), k \in [0,1)$:
$$
  \forall i \in \{1,2,3\}, \quad s_i ^2 = \frac{1}{k^2}\ds^2(\gamma_i|k) \quad ; \quad \sum_{i=1}^3 \gamma_i = 0.
$$
Since $|\ds(u|k)| \geq k'$ for all $u \in \RR$, we obtain a first necessary condition on $k$: 
\begin{equation}\label{eq:condition:x}
  s_i^2 \geq \frac{(k')^2}{k^2} =: x^2 \quad ; \quad i \in \{1,2,3\}.
\end{equation}
so we can restrict to $x^2 \leq \min(s_1 ^2, s_2^2,s_3 ^2)$.
Writing $\gamma_3 = -\gamma_1 - \gamma_2$, this is equivalent to the set of equations of the three variables $\gamma_1, \gamma_2 \in \RR/(4K\ZZ)$, $k \in [0,1)$:
$$
  s_1^2 = \frac{1}{k^2}\ds^2(\gamma_1|k) \quad ; \quad s_2^2 = \frac{1}{k^2}\ds^2(\gamma_2|k) \quad ; \quad s_3^2 = \frac{1}{k^2}\ds^2(\gamma_1+\gamma_2|k).
$$
Note that since these equations are left unchanged when translating $\gamma_1, \gamma_2$ by $2K$, we can assume that $\gamma_1, \gamma_2 \in \RR/(2K\ZZ)$. 
Since $\ds : (0,K] \to [k', +\infty)$ is invertible for each $k$, we can define $d_i = \argds(k s_i|k) \in (0,K]$.
Hence, we must have 
$$
  \gamma_i = \eps_id_i \ [2K] \quad ; \quad i \in \{1,2\},
$$
with $\eps_i = \pm 1$. 
Since $\gamma_i \in \RR/(2K\ZZ)$, we loose no generality at this point.
Recall that $s_2 \leq s_1  $, hence, under condition~\eqref{eq:condition:x}, the system of equations is equivalent to solving for $x^2 = (k')^2/k^2 \in (0,s_2^2], \eps_2 = \pm 1, \eps_1  = \pm 1$, the unique equation
\begin{equation*}
  \begin{aligned}
    s_3 ^2 
    &= \frac{1}{k^2}\ds^2\big[\eps_2 d_2 +\eps_1 d_1\big|k\big].
  \end{aligned}
\end{equation*}
Since the equation is the same for $(\eps_2,\eps_1 )$ and $(-\eps_2,-\eps_1 )$, we can restrict to the case where $\eps_2=1,\eps := \eps_1 =\pm 1$ and solve
\begin{equation}\label{eq:triangular:cns}
    \begin{aligned}
    s_3 ^2 
    &= \frac{1}{k^2}\ds^2\big[d_2 +\eps d_1\big|k\big].
    \end{aligned}
  \end{equation}
  On the one hand, the addition formula for $\ds$ is (see~\cite[(2.4.1)-(2.4.3)]{Lawden}):
  $$
    \ds(d_2 +\eps d_1) = \frac{\dn(d_2 )\dn(d_1)-\eps k^2\sn(d_2 )\sn(d_1)\cn(d_2 )\cn(d_1)}{\sn(d_2 )\cn(d_1)\dn(d_1)+\eps\sn(d_1)\cn(d_2 )\dn(d_2 )}.
  $$
  On the other hand, using~\cite[(2.1.9)-(2.1.11)]{Lawden}, we obtain (omitting the variables)
  $$
    \sn^2 = \frac{1}{\ds^2+k^2} \quad ; \quad \cn^2 = \frac{\ds^2-(k')^2}{\ds^2+k^2} \quad ; \quad \dn^2 = \frac{\ds^2}{\ds^2+k^2}.
  $$
  hence for $i \in \{1,2\}$, taking the square root and using the fact that $d_i \in [0,K/2]$, we obtain
  {\small 
  $$
    \begin{aligned}
    \sn(d_i) &
    = \frac{1}{\sqrt{\ds^2(d_i)+k^2}}
    =  \frac{1}{\sqrt{k^2s_i^2+k^2}}
    = \frac{1}{k\sqrt{s_i^2+1}}\\
    \cn(d_i) 
    &= \sqrt{\frac{\ds^2(d_i)-(k')^2}{\ds^2(d_i)+k^2}} 
    = \sqrt{\frac{k^2s_i^2 - (k')^2}{k^2s_i^2 + k^2}}
    = \frac{\sqrt{s_i^2-x^2}}{\sqrt{s_i^2+1}}\\
    \dn(d_i) 
    &= \sqrt{\frac{\ds^2(d_i)}{\ds^2(d_i)+k^2}}
    = \frac{\sqrt{k^2s_i^2}}{\sqrt{k^2s_i^2+k^2}}
    = \frac{s_i}{\sqrt{s_i^2+1}}.
    \end{aligned}
  $$
  }
  The addition formula then becomes
  $$
    \begin{aligned}
      \ds(d_2 +\eps d_1)
      &= \frac
      {
        \frac{s_2}{\sqrt{s_2^2+1}}
        \frac{s_1 }{\sqrt{s_1 ^2+1}}
        -\eps k^2
        \frac{1}{k\sqrt{s_2^2+1}}
        \frac{1}{k\sqrt{s_1 ^2+1}}
        \frac{\sqrt{s_2^2-x^2}}{\sqrt{s_2^2+1}}
        \frac{\sqrt{s_1 ^2-x^2}}{\sqrt{s_1 ^2+1}}
      }
      {
        \frac{1}{k\sqrt{s_2^2+1}}
        \frac{\sqrt{s_1 ^2-x^2}}{\sqrt{s_1 ^2+1}}
        \frac{s_1 }{\sqrt{s_1 ^2+1}}
        + \eps
        \frac{1}{k\sqrt{s_1 ^2+1}}
        \frac{\sqrt{s_2^2-x^2}}{\sqrt{s_2^2+1}}
        \frac{s_2}{\sqrt{s_2^2+1}}
      }
      \\
      &= 
      k
      \frac{
      s_2 s_1 
      -
      \eps
      \frac
      {\sqrt{s_2^2-x^2}\sqrt{s_1 ^2-x^2}}
      {\sqrt{s_2^2+1}\sqrt{s_1 ^2+1}}
      }
      {
      s_1 
      \frac
      {\sqrt{s_1 ^2-x^2}}
      {\sqrt{s_1 ^2+1}}
      +
      \eps s_2
      \frac
      {\sqrt{s_2^2-x^2}}
      {\sqrt{s_2^2+1}}
      }.
    \end{aligned}
  $$
  Hence, solving Equation~\eqref{eq:triangular:cns} is equivalent to solving for $x^2 =(k')^2/k^2 \in (0,s_2^2], \eps = \pm 1$
  $$
    s_3  
    =\pm f_{\eps}(x^2),\quad \text{ with } 
    f_{\eps}(x^2):=
      \frac{
      s_2 s_1 
      -
      \eps
      \frac
      {\sqrt{s_2^2-x^2}\sqrt{s_1 ^2-x^2}}
      {\sqrt{s_2^2+1}\sqrt{s_1 ^2+1}}
      }
      {
      s_1 
      \frac
      {\sqrt{s_1 ^2-x^2}}
      {\sqrt{s_1 ^2+1}}
      +
      \eps s_2
      \frac
      {\sqrt{s_2^2-x^2}}
      {\sqrt{s_2^2+1}}
      } =: f_{\eps}(x^2).
  $$
  This equation has a solution if and only if
  there exists $x^2 \in (0,s_2^2]$ such that 
  \begin{equation*}
    s_3 =\pm f_{1}(x^2) \quad \text{or} \quad s_3 =\pm f_{-1}(x^2).
  \end{equation*}
  We can finally conclude by identifying the image by $f_{1}$ and $f_{-1}$ of the interval $(0,s_2^2]$.
  
  \emph{Study of $f_{1}$.}
  The denominator of $f$ is non-negative and decreasing, while the numerator of $f_{1}$ is increasing, and non-negative since
  \[
  \frac
      {\sqrt{s_2^2-x^2}\sqrt{s_1 ^2-x^2}}
      {\sqrt{s_2^2+1}\sqrt{s_1 ^2+1}}
    \leq 
    \sqrt{s_2^2-x^2}\sqrt{s_1 ^2-x^2}
    \leq s_2 s_1 .
    \]
Hence, $f_1$ is increasing on $(0,s_2^2]$, so
  $$
    \begin{aligned}
    f_{1}((0,s_2^2]) = \big(f_{1}(0),f_{1}(s_2^2)\big] 
    &= \Biggl(s_2 s_1 \frac{\sqrt{s_2^2+1}\sqrt{s_1 ^2+1}-1}{s_1 ^2\sqrt{s_2^2+1}+s_2^2\sqrt{s_1 ^2+1}},\frac{s_2\sqrt{s_1 ^2+1}}{\sqrt{s_1 ^2-s_2^2}}\Biggr]\\
    &= \Biggl( T, \frac{s_2\sqrt{s_1 ^2+1}}{\sqrt{s_1 ^2-s_2^2}}\Biggr].
    \end{aligned}
  $$
  \emph{Study of $f_{-1}$.}
  The numerator of $f_{-1}$ is positive and decreasing on $(0,s_2^2]$, 
  while the denominator $g(x)$ of $f_{-1}$ is positive and increasing on $(0,s_2^2]$, 
  because
  \[g(x)\geq 0 \iff 
  s_1 
      \frac
      {\sqrt{s_1 ^2-x^2}}
      {\sqrt{s_1 ^2+1}}
      \geq s_2
      \frac
      {\sqrt{s_2^2-x^2}}
      {\sqrt{s_2^2+1}} \iff 
      \frac{\sqrt{s_1 ^2-x^2}}{\sqrt{1+s_1 ^{-2}}} \geq 
      \frac{\sqrt{s_2^2-x^2}}{\sqrt{1+s_2^{-2}}} ,
  \]
 and 
  \begin{align*}
  g'(x)& \geq 0 \iff
  \frac{1}{\sqrt{1+\frac{1}{s_1 ^2}}}\frac{1}{\sqrt{s_1 ^2-x^2}}\leq 
  \frac{1}{\sqrt{1+\frac{1}{s_2^2}}}\frac{1}{\sqrt{s_2^2-x^2}}\\
  \iff &
  s_2^2+1-x^2\Bigl(1+\frac{1}{s_ 1^2}\Bigr)\leq s_1 ^2+1-x^2\Bigl(1+\frac{1}{s_ 2^2}\Bigr),
  \end{align*}
  which are both true since $s_2\leq s_1 $.
  Hence, $f_{-1}$ is decreasing on $(0,s_2^2]$, so
  $$
    \begin{aligned}
    f_{-1}((0,s_2^2]) =
    \big[f_{-1}(s_2^2),f_{-1}(0)\big) 
    &= 
    \Biggl[
      \frac{s_2\sqrt{s_1 ^2+1}}{\sqrt{s_1 ^2-s_2^2}},
      \frac{s_2 s_1 \left(\sqrt{s_2^2+1}\sqrt{s_1 ^2+1}+1\right)}{s_1 ^2\sqrt{s_2^2+1}-s_2^2\sqrt{s_1 ^2+1}}
    \Biggr).
    \end{aligned}
  $$
  Since both $f_1$ and $f_{-1}$ are non-negative and monotone, and since $s_3 $ is positive, the values that can be obtained for $s_3 $ with $f_{1}$ and $f_{-1}$, are $s_3  \in (f_{1}(0),f_{-1}(0))$.
  By Equation~\eqref{eq:T<s1}, $T = f_1(0) < s_2$. 
  Furthermore,
  $$
    f_{-1}(0) 
    = \frac{s_2 s_1 \left(\sqrt{s_2^2+1}\sqrt{s_1 ^2+1}+1\right)}{s_1 ^2\sqrt{s_2^2+1}-s_2^2\sqrt{s_1 ^2+1}}
    > \frac{s_2 s_1 \sqrt{s_2^2+1}\sqrt{s_1 ^2+1}}{s_1 ^2\sqrt{s_2^2+1}}
    = s_2\frac{\sqrt{s_1 ^2+1}}{s_1 } > s_2.
  $$
  Hence,
  \begin{equation*}
    0 \leq f_{1}(0) \leq s_2 < f_{-1}(0).
  \end{equation*}
  Since $s_3  \leq s_2$, we can conclude that the Ising models parameterized in this way are exactly those with $s_3  \in (f_{1}(0), s_2] = (T, s_2]$.
\end{proof}

\subsection{The square lattice.}\label{subsec:square_lattice}

In this section, we establish an explicit parameterization of the isotropic frustrated Ising model on the square lattice by Fock's dimer model. This allows us to prove phase transition at 0 temperature, a fact previously obtained in the physics literature cited in the introduction. A subject of further interest is exploring the connections to~\cite[Section III]{Forgacs}. We will not study in details all Fock's dimer models that can be obtained, since this involves a lot of computations, and we have already done this for the triangular lattice in Section~\ref{sec:triangular_lattice}.

\begin{defi}
  The isotropic frustrated Ising model on the square lattice is the (equivalence class of) Ising model with real coupling constants $\eps\Js$ such that $\Js$ is the same for all edges 
  and the frustration function is $\delta_f = -1$ at every face.
  An example of model in this equivalence class is the Ising model on the square lattice with $2 \times 2$ fundamental domain and signs given by $\eps_1 = \eps_2 = \eps_3 = 1$, $\eps_4 = -1$.
\end{defi}
Note that even though the face-weights of this model are $1 \times 1$ periodic, it is impossible to find $1 \times 1$-periodic coupling constants realizing this model by Lemma~\ref{lem:reconstructing_Ising_model}. 

We consider the square lattice $G$ with $2 \times 2$ fundamental domain; then the dual graphs $G^*$ and $G_1^*$ are also bipartite. The set of train-tracks $\T$ of $G^*$ can be consistently oriented so that white vertices are on the left, see Figure~\ref{fig:square22}. The fundamental domain $G_1$ has four oriented train-tracks denoted by $\{T_0,\dots,T_3\}$ in counterclockwise order, and the set of oriented train-tracks $\vec{\T}$ of $\GQ_1$ is split into $\vec{\T}=\vec{\T}_{*,\circ}\sqcup \vec{\T}_{*,\bullet}$, where $\vec{\T}_{*,\circ}=\{\vec{T}_0,\dots,\vec{T}_3\}$ is in orientation preserving bijection with $\T$. 
Recall that $\alpha_{\cev{T}_i} = \alpha_{\vec{T}_i} + \rho$. Then, we obtain the following.

\begin{figure}[h]
  \centering
  \begin{overpic}[width=12cm]{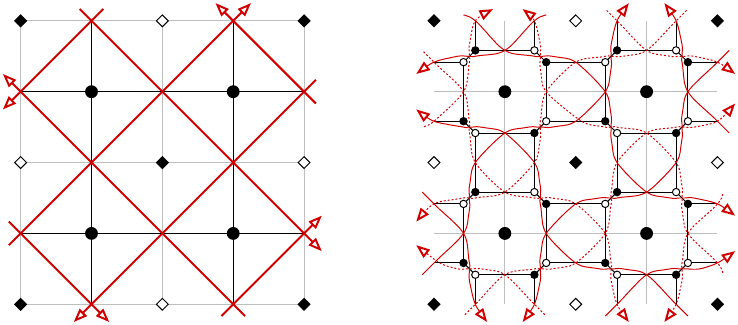} 
\put(8,-2){\scriptsize $T_3$}
\put(14,-2){\scriptsize $T_0$} 
\put(44,9){\scriptsize $T_0$}
\put(44,14){\scriptsize $T_1$} 
\put(35,43){\scriptsize $T_1$}
\put(26,43){\scriptsize $T_2$}
\put(-3,33){\scriptsize $T_2$}
\put(-3,28){\scriptsize $T_3$}
  \end{overpic}
\caption{Left: $G$ is the square lattice with a $2\times 2$ fundamental domain; oriented train-tracks of $\T$ are pictured in red. Right: the corresponding graph $\GQ$; oriented train-tracks $\vec{T}$ of $\T_{*,\circ}$, in orientation preserving bijection with those of $\T$, are pictured with full lines; oriented train-tracks $\cev{T}$ of $\T_{*,\bullet}$ are pictured with dotted lines.}\label{fig:square22}
\end{figure}
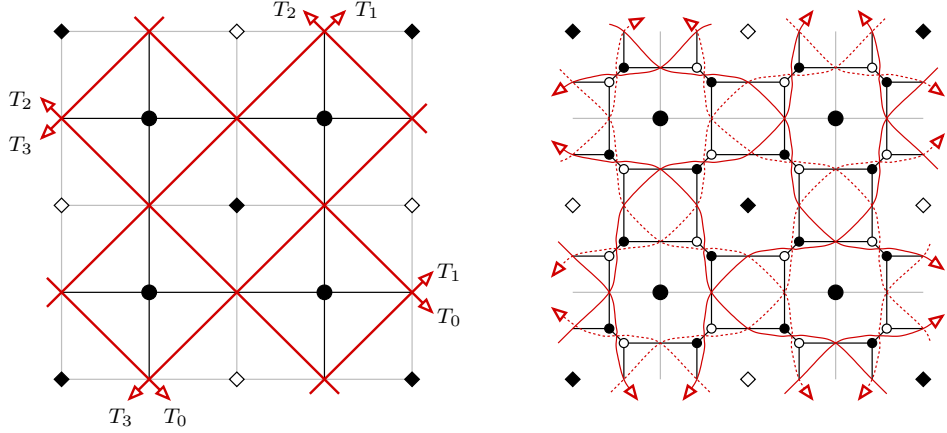

\begin{prop}\label{prop_square_lattice}
  Fock's dimer model with parameters $k^2\in[0,1)$, $\rho = \frac{1}{2}+\frac{\tau}{2}$, angle map
  $\alpha_{\vec{T}_0}=0,\alpha_{\vec{T}_1}=\frac{1}{4}, \alpha_{\vec{T}_2}=\frac{1}{2}, \alpha_{\vec{T}_3}=\frac{3}{4}$ is gauge equivalent to the Ising-dimer model arising from the isotropic frustrated Ising model with coupling constants $\eps\Js$, given by
  $$
  \sinh(2\Js) = \frac{\sqrt{k'(1+k')}}{k}.
  $$
  As $k$ runs from $0$ to $1$, this parameterizes all possible coupling constants $\Js \in (0, \infty)$. 
\end{prop}
\begin{proof}
Since the angle map preserves the cyclic order of the train-tracks of $\T$, we have by Corollary~\ref{cor:angle_map_2} that this model corresponds to an Ising model, frustrated at every face.
Let us identify the coupling constants.
At every square face $y$, by Table~\ref{table3} and using~\cite[(2.4.10)]{Lawden},
we obtain
  \[
  \begin{aligned}
    \sinh^2(2\Js) 
    = \frac{(k')^2}{k^2} \nc^2(K/2)
    = \frac{k'(1+k')}{k^2}
  \end{aligned}\qedhere
  \]
\end{proof}


One can use this parameterization to plot the amoeba and real locus of the spectral curve, see Figure~\ref{fig_amoeba_square} (left and center).
At first sight, this does not seem to be coherent with our set of parameters.  
On the one hand, looking at the amoeba and real locus, the curve seems to be Harnack.
On the other hand, the curve is parameterized by Equation~\eqref{equ:param_spectral_curve_symm} with the $\alpha_{\vec{T}_i}, \alpha_{\cev{T}_i}$ pictured in Figure~\ref{fig_amoeba_square} (right).
Hence we would expect the amoeba to feature two superimposed ``sheets'' as in the right of Figure~\ref{fig:triangular_lattice_curves}. 

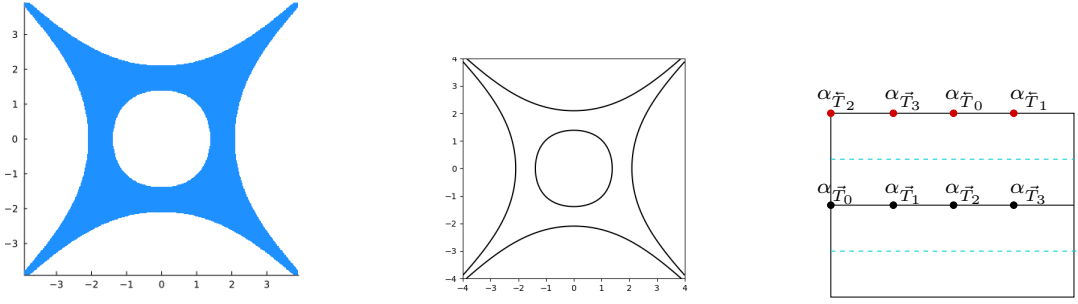
\begin{figure}[h]
\begin{subfigure}{0.32\textwidth}
        \centering
\includegraphics[scale=.3, trim = 3cm 0cm 3cm .8cm, clip]{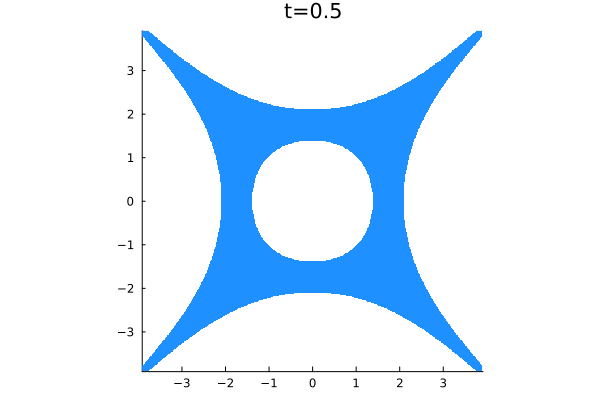}
\end{subfigure}
\hfill 
\begin{subfigure}{0.32\textwidth}
\centering
\includegraphics[scale=.31, trim = 0cm 0.5cm 0cm 1.4cm, clip]{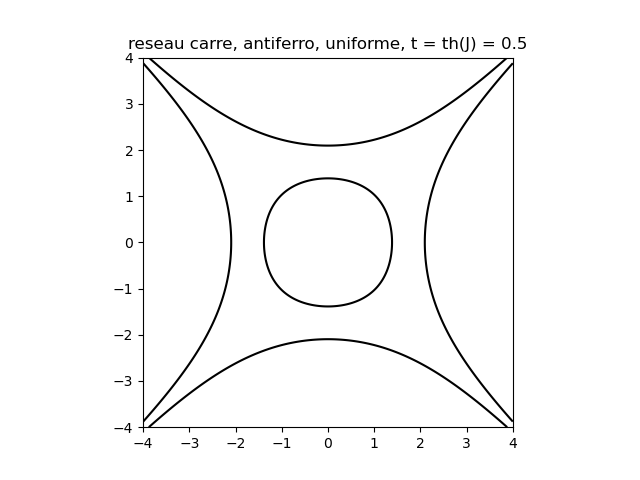}
\end{subfigure}
 \begin{subfigure}{0.32\textwidth}
         \centering
\begin{overpic}[width=3.3cm]{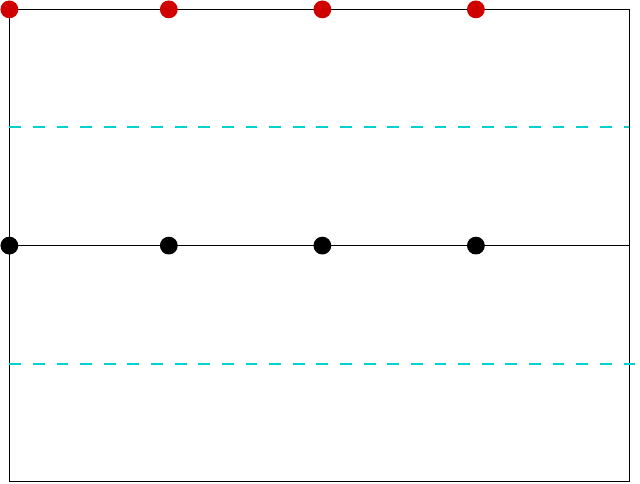}
        \put(-4,42){\scriptsize $\alpha_{\vec{T}_0}$} 
        \put(23,42){\scriptsize $\alpha_{\vec{T}_1}$} 
        \put(47,42){\scriptsize $\alpha_{\vec{T}_2}$}
        \put(73,42){\scriptsize $\alpha_{\vec{T}_3}$}
        \put(47,80){\scriptsize $\alpha_{\cev{T}_0}$} 
        \put(73,80){\scriptsize $\alpha_{\cev{T}_1}$} 
        \put(-4,80){\scriptsize $\alpha_{\cev{T}_2}$} 
        \put(23,80){\scriptsize $\alpha_{\vec{T}_3}$} 
\end{overpic}
\end{subfigure}
\caption{On the left: the amoeba of the spectral curve of the isotropic frustrated Ising model on the square lattice for some $0 < \Js < \infty$.
In the middle: the image of the real locus through the map $(z,w) \to (\log|z|, \log|w|)$. On the right: the torus and the angle map used in the parameterization of Fock's dimer model and its spectral curve.}
\label{fig_amoeba_square}
\end{figure}
To solve this apparent contradiction, we must take a closer look at the parameterization. We prove the following. 

\begin{prop}\label{prop:square}
  The spectral curve of the isotropic frustrated Ising model on the square lattice:
  \begin{itemize}
    \item is reducible, in the sense that it is the superposition of two identical copies of the same curve $\C$. In particular, it does not satisfy $(\dagger)$,
    \item is the symmetric via $(z,w) \to (z,-w)$ of the spectral curve of an isotropic ferromagnetic Ising model on the square lattice (with a different coupling constant which can be written down explicitly),
    \item each of the copy of the curve is thus a Harnack curve.
  \end{itemize}
  As a consequence, the free energy of the model is given by Proposition~\ref{prop:phase:transition}.
  By Remark~\ref{rem:phase_transition}, the algebraic phase transition at $k=0$, is in fact a phase transition in the sense that the free energy has a singularity.
  Thus by Proposition~\ref{prop_square_lattice}, the isotropic frustrated Ising model on the square lattice has a phase transition at zero temperature.
  \end{prop}
\begin{proof}
The proof consists in writing down explicitly the parameterization $\Psi= (z(u), w(u)): \TT(\tau) \to \C$ of the spectral curve given by Equation~\eqref{equ:param_spectral_curve_symm} and applying algebraic identities on theta functions. Using the values of the angle map $\mapalpha$ of Proposition~\ref{prop_square_lattice}, we have
$$
z(u) = 
\frac
{\theta_{1,1}(u)\theta_{1,1}(u-\tau/2)\theta_{1,1}(u-3/4)\theta_{1,1}(u-3/4-\tau/2)}
{\theta_{1,1}(u-1/2)\theta_{1,1}(u-1/2-\tau/2)\theta_{1,1}(u-1/4)\theta_{1,1}(u-1/4-\tau/2)}
$$
and 
$$
w(u) = 
\frac
{\theta_{1,1}(u)\theta_{1,1}(u-\tau/2)\theta_{1,1}(u-1/4)\theta_{1,1}(u-1/4-\tau/2)}
{\theta_{1,1}(u-1/2)\theta_{1,1}(u-1/2-\tau/2)\theta_{1,1}(u-3/4)\theta_{1,1}(u-3/4-\tau/2)}.
$$
Using Equation~\eqref{eq:theta:4} and Equation~(20.7.12) of~\cite{DLMF}, this can be rewritten as 
\begin{equation*}
z(u) = 
\frac
{\theta_{1,1}(u, \tau/2)\theta_{1,1}(u-3/4, \tau/2)}
{\theta_{1,1}(u-1/2, \tau/2)\theta_{1,1}(u-1/4, \tau/2)}\quad ; \quad 
w(u) = 
-\frac
{\theta_{1,1}(u, \tau/2)\theta_{1,1}(u-1/4, \tau/2)}
{\theta_{1,1}(u-1/2, \tau/2)\theta_{1,1}(u-3/4, \tau/2)}.
\end{equation*}
As the proof is very similar to Proposition~\ref{prop_square_lattice} in the (easier) $1 \times 1$ periodic case, we let the reader check that, upon applying the symmetry $(z,w) \to (z,-w)$, this is the parameterization of the isotropic ferromagnetic Ising model on the square lattice with parameter $\tau/2$.
It it thus a Harnack curve.

This result can also be obtained by computing directly the characteristic polynomials.
  Denote by $P_{i\times j}^\pm$, the characteristic polynomial of the isotropic frustrated (if $\pm = -$), resp. non-frustrated (if $\pm = +$), Ising model on the $i \times j$ periodic, square lattice, with absolute value coupling constant $\Js_-$, resp. coupling constants $\Js_+$. 
  One can check by direct computation that for all $\Js_- \in \RR_+$, there exists $\Js_+\in\RR^+, \lambda \in \RR_+$ such that
  $$
    P_{2\times 2}^-(z,w) = 
    \lambda P_{1 \times 1}^+(z,-w)^2.
  $$
  Again, this proves the results of Proposition~\ref{prop:square}, without resorting to our theory.
\end{proof}

\begin{rem}
This parameterization is coherent with Figure~\ref{fig_amoeba_square}.
In fact, the map $\Psi$ has a vertical period $\tau/2$: it is two-to-one and the two circles $\RR$ and $\RR + \tau/2$ containing the zeros and poles have the same image through the map, which is the outer boundary of the amoeba.
Moreover, the image of the circle $\RR + \tau/4$ is also real, thus implying that the real locus of $\C$ has a bounded component.
\end{rem}

\appendix
\section{Identities involving theta and elliptic functions}\label{sec:App_A}

In this section we provide statements and proofs of identities involving theta and elliptic functions that we use, and that we could not find as such in the literature. 

\begin{lem}\label{lem:appendix_1}
For $m,n\in\{0,1\}$ with indices understood mod 2, and $z,w\in\CC$,
\begin{equation*}
\begin{split}
\theta_{1,1}^2(z)\theta_{0,0}^2(w) + &(-1)^{1+m+mn}\theta_{m,n}^2(z)\theta_{m+1,n+1}^2(w)=\\
&=(-1)^{n+mn} \theta_{m+1,n+1}(z+w)\theta_{m+1,n+1}(z-w)\theta_{m,n}^2(0). 
\end{split}
\end{equation*}
\end{lem}

\begin{proof}
The following is a compact form of Equations~\cite[(1.4.19)-(1.4.22)-(1.4.29)]{Lawden},
\begin{align*}
&\theta_{1,1}^2(z)\theta_{1,1}^2(w) + (-1)^{m+n+mn}\theta_{m,n}^2(z)\theta_{m,n}^2(w) = (-1)^{m+n+mn}\theta_{m,n}^2(z+w)\theta_{m,n}^2(z-w)\theta_{m,n}^2(0),
\end{align*} 
noting that when $(m,n)=(1,1)$ the left-hand-side vanishes, and $\theta_{1,1}^2(0)=0$ in the right-hand-side. Evaluating the equation at $w\mapsto w+\frac{1}{2}+\frac{\tau}{2}$ gives: 
\begin{align*}
\theta_{1,1}^2(z)\theta_{1,1}^2&\Bigl(w+\frac{1}{2}+\frac{\tau}{2}\Bigr) + (-1)^{m+n+mn}\theta_{m,n}^2(z)\theta_{m,n}^2\Bigl(w+\frac{1}{2}+\frac{\tau}{2}\Bigr)\\
&= (-1)^{m+n+mn}\theta_{m,n}\Bigl(z+w+\frac{1}{2}+\frac{\tau}{2}\Bigr)\theta_{m,n}\Bigl(z-w-\frac{1}{2}-\frac{\tau}{2}\Bigr)\theta_{m,n}^2(0).
\end{align*}
Using Identities ~\eqref{eq:theta:4} and~\eqref{eq:theta:5} with $(j,\ell) = (1,1)$, we obtain
\begin{align*}
(q^{-1/4}e^{-i\pi w})^{2}&(i)^{2\cdot 2}\theta_{1,1}^2(z)\theta_{0,0}^2(w) + (-1)^{m+n+mn}
(q^{-1/4}e^{-i\pi w})^{2}(i)^{2(n+1)}\theta_{m,n}^2(z)\theta_{m+1,n+1}^2(w)\\
&= (-1)^{m+n+mn}(q^{-1/4}e^{-i\pi(z+w)})(q^{-1/4}e^{i\pi(z-w)})i^{2(n+1)} (-1)^{n+1+n+mn}(-1)^{m+n+mn}\\
  &\quad \theta_{m+1,n+1}(z+w)\theta_{m+1,n+1}(z-w)\theta_{m,n}^2(0).
\end{align*}
The result is obtained after simplification in the above. 
\end{proof}

\begin{lem}{\emph{\cite[\textsection 6, p.24]{MumfordTheta1}}}\label{lem:meromorphic_theta_ratio_appendix}
Let $a_1,\dots,a_r, b_1,\dots,b_r\in\CC$. Then, the function
\[
f(u)=e^{-2i\pi\ell u}\prod_{i=1}^r \frac{\theta_{m,n}(u-a_i)}{\theta_{m,n}(u-b_i)},
\]
is a meromorphic function on $\TT(\tau)$ if and only if $\ell\in\ZZ$, and $\sum_{i=1}^r (a_i-b_i)-\ell\tau \in\ZZ$.
\end{lem}

\begin{proof}
In~\cite{MumfordTheta1}, the result is written in the case where $\sum_{i=1}^r (a_i-b_i) \in \ZZ$ with no occurrence of the prefactor. Since the function $f$ is meromorphic on $\CC$, we are left with checking periodicity on $\TT(\tau)$. The function $f$ is periodic if and only if, for every $j',\ell'\in\{0,1\}$, $f(u+j'+\ell'\tau)=f(u)$. By Equation~\eqref{eq:theta:2}, we have 
\begin{align*}
\frac{f(u+j'+\ell'\tau)}{f(u)}=
e^{-2i\pi\ell(j'+\ell'\tau)}e^{-2i\pi\ell'(\sum_{i=1}^r (-a_i+b_i))}.
\end{align*}
This ratio is equal to 1 for every $j',\ell'\in\{0,1\}$ if and only if $\ell\in\ZZ$, and $\ell\tau+\sum_{i=1}^r (b_i-a_i) \in\ZZ$, thus concluding the proof.
\end{proof}

\paragraph{Ascending Landen transformation and modular transformation of theta and Jacobi elliptic functions.}
In this paragraph, we prove the following algebraic identities on Jacobi elliptic functions and theta functions.
Recall from Section~\ref{sec:jacobi:elliptic} the notation $\gamma^k = 2K(k)\gamma$.
\begin{lem}\label{lem:ascending:Landen:Jacobi}
  For $\tau \in \CC$, let $\tau_1 = 2\tau$, $\tau_2 = \tau_1 - 1 = 2\tau-1$.
  Denote also by $k_1,k_2$ the elliptic modulus associated respectively with $\tau_1$ and $\tau_2$.
  Then, for all $\gamma \in \CC$,
  $$
    \sc(\gamma^k|k) = (1+k_1)k_2'\sc(\gamma^{k_2}|k_2)\dn(\gamma^{k_2}|k_2).
  $$
\end{lem}
 
The proof of this lemma consists in performing an ascending Landen transformation (see Section~3.9 of~\cite{Lawden}) to go from $\tau$ to $\tau_1 = 2\tau$ and then the modular transformation $\tau \to \tau+1$ to go from $\tau_1$ to $\tau_2$.
The details are given below.
This implies the following.

\begin{lem}\label{lem:ascending:Landen:theta}
   For $\alpha,\beta \in \CC$, 
   \begin{equation*}
     \frac
     {\theta_{1,1}(\alpha,\tau)\theta_{1,0}(\beta,\tau)}
     {\theta_{1,0}(\alpha,\tau)\theta_{1,1}(\beta,\tau)} = 
     \frac
     {\theta_{1,1}(\alpha,2\tau-1)\theta_{0,0}(\alpha,2\tau-1)\theta_{1,0}(\beta,2\tau-1)\theta_{0,1}(\beta,2\tau-1)}
     {\theta_{1,0}(\alpha,2\tau-1)\theta_{0,1}(\alpha,2\tau-1)\theta_{1,1}(\beta,2\tau-1)\theta_{0,0}(\beta,2\tau-1)}
   \end{equation*}
 \end{lem}

This is a one-line corollary of Lemma~\ref{lem:ascending:Landen:Jacobi}, because by definition of the $\sc$ and $\dn$ functions:
{\small 
  \begin{equation*}
    \frac
    {\theta_{1,1}(\alpha|\tau)\theta_{1,0}(\beta|\tau)}
    {\theta_{1,0}(\alpha|\tau)\theta_{1,1}(\beta|\tau)} 
    =
    \frac{\sc(\alpha^k|\tau)}{\sc(\beta^k|\tau)}
    = \frac{\sc(\alpha^{k_2}|\tau_2)
  \dn(\alpha^{k_2}|\tau_2)}{\sc(\beta^{k_2}|\tau_2)
  \dn(\beta^{k_2}|\tau_2)}
  = 
  \frac
    {\theta_{1,1}(\alpha|\tau_2)\theta_{0,0}(\alpha|\tau_2)\theta_{1,0}(\beta|\tau_2)\theta_{0,1}(\beta|\tau_2)}
    {\theta_{1,0}(\alpha|\tau_2)\theta_{0,1}(\alpha|\tau_2)\theta_{1,1}(\beta|\tau_2)\theta_{0,0}(\beta|\tau_2)}
  \end{equation*}
}
Note that in principle we could obtain Lemma~\ref{lem:ascending:Landen:theta} directly with the same techniques since the ascending Landen transform also exists for theta functions, but we find the computations much easier with Jacobi elliptic functions.

\begin{proof}[Proof of Lemma~\ref{lem:ascending:Landen:Jacobi}]
%
By the first point of Equation (22.2.2) of~\cite{DLMF}, Equation (1.8.2) of \cite{Lawden}, the third point of Equation (22.2.2) of \cite{DLMF} and again the first point of Equation (22.2.2) of~\cite{DLMF}: 
$$
  2K(k) 
  = \theta_{0,0}^2(0,\tau) 
  = \theta_{0,0}^2(0,\tau_1)+\theta_{1,0}^2(0,\tau_1)
  = (1+k_1)\theta_{0,0}^2(0,\tau_1)
  = (1+k_1)2K(k_1) 
$$
so 
Using Equations (3.9.28), (3.9.29) of~\cite{Lawden} (the ascending Landen transform for $\sn$ and $\cn$), we first obtain
  \begin{equation}\label{eq:ascending:Landen}
    \sc(\gamma^k|k) 
    = (1+k_1)\frac
    {\sn\big(\frac{\gamma^k}{1+k_1}\big|k_1\big)}
    {\cn\big(\frac{\gamma^k}{1+k_1}\big|k_1\big)\dn\big(\frac{\gamma^k}{1+k_1}\big|k_1\big)}
    = (1+k_1)
    \frac
    {\sn(\gamma^{k_1}|k_1\big)}
    {\cn(\gamma^{k_1}|k_1)\dn(\gamma^{k_1}|k_1)}.
  \end{equation}
We now use the modular transformation $\tau \to \tau+1$. 
First observe that by the first and second points of Equation (22.2.2) of~\cite{DLMF} and Equation~\eqref{eq:tau_2},
$$
  2K(k_2) 
  = \theta_{0,0}^2(0,\tau_2) 
  = \frac{\theta_{0,1}^2(0,\tau_2)}{k_2'} 
  = \frac{\theta_{0,0}^2(0,\tau_1)}{k_2'} 
  = \frac{2K(k_1)}{k_2'}. 
$$
Then, the modular transformation $\tau \to \tau +1$ for elliptic functions (9.3.6)-(9.3.8) of~\cite{Lawden} (which is a direct consequence of Equation~\eqref{eq:tau_2} and the definition of the Jacobi elliptic functions) writes 
\begin{equation*}
  \sn(\gamma^{k_1}|k_1) = k_2' \sd(\gamma^{k_2}|k_2) 
  \quad ; \quad
  \cn(\gamma^{k_1}|k_1) = k_2' \cd(\gamma^{k_2}|k_2)
  \quad ; \quad
  \dn(\gamma^{k_1}|k_1) = k_2' \nd(\gamma^{k_2}|k_2)  
\end{equation*}
so Equation~\eqref{eq:ascending:Landen} gives
\begin{equation*}
  \sc(\gamma^k|k) 
  =
  (1+k_1)k_2' 
  \frac
    {\sd(\gamma^{k_2}|k_2)}
    {\cd(\gamma^{k_2}|k_2)\nd(\gamma^{k_2}|k_2)}
  = 
  (1+k_1)k_2' 
  \sc(\gamma^{k_2}|k_2)
  \dn(\gamma^{k_2}|k_2),
\end{equation*}
which concludes the proof of the lemma.
\end{proof}

\bibliographystyle{alpha}
\bibliography{anti_ferro_ising}

\end{document}